\immediate\write18{bibtex calco-2019}
\documentclass[a4paper,UKenglish,cleveref,autoref,amsmath=flneqn]{lipics-v2019}
\title{Nominal String Diagrams}

\titlerunning{Nominal String Diagrams}%optional, please use if title is longer than one line

\author{Samuel Balco}{Department of Informatics,University of Leicester, United Kingdom\and \url{https://gdlyrttnap.pl} }{sb782@leicester.ac.uk}{https://orcid.org/0000-0002-1825-0097}{}%TODO mandatory, please use full name; only 1 author per \author macro; first two parameters are mandatory, other parameters can be empty. Please provide at least the name of the affiliation and the country. The full address is optional

\author{Alexander Kurz}{Department of Computer Science, Chapman University, Orange California, USA}{akurz@chapman.edu}{[orcid]}{}

\authorrunning{S.\,Balco and A.\,Kurz}%TODO mandatory. First: Use abbreviated first/middle names. Second (only in severe cases): Use first author plus 'et al.'

\Copyright{Samuel Balco and Alexander Kurz}%TODO mandatory, please use full first names. LIPIcs license is "CC-BY";  http://creativecommons.org/licenses/by/3.0/

\ccsdesc[100]{D.3, F.1, F.3, G.m, L.1, }%TODO mandatory: Please choose ACM 2012 classifications from https://dl.acm.org/ccs/ccs_flat.cfm 

\keywords{string diagrams, nominal sets, separated product, simultaneous substitutions, internal category, monoidal category, internal monoidal categories, PROP}
\category{}%optional, e.g. invited paper

\relatedversion{}%optional, e.g. full version hosted on arXiv, HAL, or other respository/website
\supplement{}%optional, e.g. related research data, source code, ... hosted on a repository like zenodo, figshare, GitHub, ...

\acknowledgements{We are greatful to Fredrik Dahlqvist, Giuseppe Greco, Samuel Mimram, Drew Moshier, Alessandra Palmigiano, David Pym, Mike Shulman, Pawel Sobocinski, Georg Struth, Apostolos Tzimoulis and Fabio Zanasi for discussions on the topic of this paper.}%optional

\EventEditors{John Q. Open and Joan R. Access}
\EventNoEds{2}
\EventLongTitle{42nd Conference on Very Important Topics (CVIT 2016)}
\EventShortTitle{CVIT 2016}
\EventAcronym{CVIT}
\EventYear{2016}
\EventDate{December 24--27, 2016}
\EventLocation{Little Whinging, United Kingdom}
\EventLogo{}
\SeriesVolume{42}
\ArticleNo{23}
\newcommand{\hcomp}{\hspace{0.24ex};}
\newcommand{\names}{\mathcal N}

\newcommand{\id}{\mathit{id}}
\newcommand{\Id}{\mathrm{Id}}

\newcommand{\Nom}{\mathsf{Nom}}
\newcommand{\supp}{\mathsf{supp}}

\newcommand{\Fun}{\mathsf{Fun}}
\newcommand{\Inj}{\mathsf{Inj}}

\newcommand{\fresh}{{\#}}
\newcommand{\sepp}{{\ast}}

\newcommand{\Fin}{\mathsf{Fin}}

\newcommand{\Mon}{\mathsf{Mon}}
\newcommand{\Cat}{\mathsf{Cat}}
\newcommand{\V}{\mathcal{V}}

\newcommand{\Set}{\mathsf{Set}}

\newcommand{\PRO}{\mathsf{PRO}}
\newcommand{\PROP}{\mathsf{PROP}}
\newcommand{\nPROP}{\mathsf n\mathsf{PROP}}
\newcommand{\NMT}{\mathrm{NMT}}
\newcommand{\SMT}{\mathrm{SMT}}
\newcommand{\ORD}{\mathit{ORD}}
\newcommand{\NOM}{\mathit{NOM}}
\newcommand{\Ordit}{\mathit{Ord}}
\newcommand{\Nomit}{\mathit{Nom}}
\newcommand{\ordit}{\mathit{ord}}
\newcommand{\nomit}{\mathit{nom}}
\newcommand{\then}{{;}}
\newcommand{\dom}{\mathit{dom}}
\newcommand{\cod}{\mathit{cod}}
\newcommand{\comp}{\mathit{comp}}
\newcommand{\compl}{\mathit{compl}}
\newcommand{\compr}{\mathit{compr}}
\newcommand{\leftt}{\mathit{left}}
\newcommand{\rightt}{\mathit{right}}
\newcommand{\pifun}{\boldsymbol{\pi}}

\newcommand{\tw}{\sigma}
\newcommand{\gen}{\gamma}

\usepackage{bbold} % used for \mathbb 1
\usepackage{xcolor} % used for color!60
\usepackage[all]{xy}
\usepackage{tikz}
\usetikzlibrary{matrix,arrows,positioning}
\usepackage[export]{adjustbox} % used for \includegraphics[center]
\usepackage{environ}
\usepackage[many]{tcolorbox}
\usetikzlibrary{patterns}

\newif\iflongversion
\longversionfalse

\NewEnviron{longVersion}{\iflongversion\expandafter\BODY\fi}

\newif\ifcomments
\commentstrue

\NewEnviron{ak}{\ifcomments\expandafter\begin{color}{magenta}\BODY\end{color}\fi}  
\NewEnviron{sam}{\ifcomments\expandafter\begin{color}{cyan}\BODY\end{color}\fi}  
\NewEnviron{gray}{\iflongversion\expandafter\begin{color}{gray}\BODY\end{color}\fi}

\newtcbox{\stripbox}{
  enhanced,nobeforeafter,
  frame code={},
  tcbox raise base,
  boxrule=0.4pt,top=0mm,bottom=0mm,right=0mm,left=0mm,
  interior code={
    \path[
      draw=gray!40,
      pattern=north east lines,
      pattern color=gray!40
    ]
    (interior.south east) rectangle (interior.north west);
  }
}

\nolinenumbers
\usepackage{etoolbox}

\begin{document}

% amsmath is imported in the template with the flag fleqn which changes center aligned equations to indented...this turns it off...
% note this requires \usepackage{etoolbox}
\setbool{@fleqn}{false}

\maketitle

\begin{abstract}
We introduce nominal string diagrams as, string diagrams internal in the category of nominal sets. This requires us to take nominal sets as a monoidal category, not with the cartesian product, but with the separated product. To this end, we develop the beginnings of a theory of monoidal categories internal in a symmetric monoidal category. As an instance, we obtain a notion of a nominal $\PROP$ as a $\PROP$ internal in nominal sets. A 2-dimensional calculus of simultaneous substitutions is an application. 
\end{abstract}

\begin{gray}
\tableofcontents
\end{gray}

\section{Introduction}
\label{sec:intro}

One reason for the success of string diagrams, see \cite{selinger} for an overview, can be formulated by the slogan `only connectivity matters' \cite[Sec.10.1]{coecke-kissinger}. Technically, this is usually achieved by ordering input and output wires and using their ordinal numbers as implicit names. We write $\underline n = \{1,\ldots n\}$ to denote the set of $n$ numbered wires and $f:\underline n\to \underline m$ for diagrams $f$ with $n$ inputs and $m$ outputs. This approach is particularly convenient for the  generalisations of Lawvere theories known as $\PROP$s \cite{maclane:prop}. In particular, the paper on composing $\PROP$s \cite{lack} has been influential \cite{rewriting-modulo,signal-flow-1}.

\medskip On the other hand, if only connectivity matters, it is natural to consider a formalisation of string diagrams in which wires are not ordered. Thus, instead of ordering wires, we fix a countably infinite set 
\[\names\]
of `names` $a,b,\ldots$, on which the only supported operation or relation is equality.  
Mathematically, this means that we work internally in the category of nominal sets introduced by Gabbay and Pitts \cite{gabbay-pitts,pitts}. In the remainder of the introduction, we highlight some of the features of this approach.

\medskip\textbf{\large Partial commutative vs total symmetric tensor. }
One reason why ordered names are convenient is that the tensor $\oplus$ is given by the categorical coproduct (additition) in the skeleton $\mathbb F$ of the category of finite sets. Even though $\underline n\oplus \underline m= \underline m\oplus \underline n$ on objects, the tensor is not commutative but only symmetric, since the canonical arrow $\underline n\oplus \underline m \to  \underline m\oplus \underline n$ is not the identity. 

\medskip
On the other hand, in the category $\mathsf n\mathbb F$ of finite subsets of $\names$ (which is equivalent to $\mathbb F$ as an ordinary category), there is a commutative tensor $A \uplus B$ given by union of disjoint sets. The interesting feature that makes  commutativity possible is that $\uplus$ is partial with $A\uplus B$ defined if and only if $A\cap B=\emptyset$. 

\medskip
While it would be interesting to develop a general theory of partially monoidal categories, our approach in this paper is based on the observation that the partial operation $\uplus:\mathsf n\mathbb F\times \mathsf n\mathbb F\to \mathsf n\mathbb F$ is a total operation $\uplus:\mathsf n\mathbb F\ast \mathsf n\mathbb F\to \mathsf n\mathbb F$ where $\ast$ is the separated product of nominal sets \cite{pitts}.

\medskip\textbf{\large Symmetries disappear in 3 dimensions. }
From a graphical point of view, the move from ordered wires to named wires corresponds to moving from planar graphs to graphs in 3 dimensions. Instead of having a one dimensional line of inputs or outputs, wires are now sticking out of a plane \cite{joyal-street:tensor1}. As a benefit there are no wire-crossings, or, more technically, there are no symmetries to take care of. This simplifies the rewrite rules of calculi  formulated in the named setting. For example, rules such as

\vspace{-1em}
\begin{center} 
\includegraphics[page=32, width=6cm, height=3cm]{twists_new}
\end{center}
\vspace{-2em}

are not needed anymore. For more on this compare Figs~\ref{fig:smt-theories} and~\ref{fig:nmt-theories}.

%\paragraph{Running example: simultaneous subsitutions}
\medskip\textbf{\large Example: Simultaneous Substitutions. }
Substitutions $[a{\mapsto}b]$ %and $[b{\mapsto}c]$ 
can be composed sequentially and in parallel as in
\[
[a{\mapsto}b]\hcomp[b{\mapsto}c] = [a{\mapsto}c]
\quad\quad
\quad\quad
[a{\mapsto}b] \uplus [c{\mapsto}d] = [a{\mapsto}b, c{\mapsto}d].
\]

We call $\uplus$ the tensor, or the monoidal or vertical or parallel composition. Semantically, the simultaneous substitution on the right-hand side above, will correspond to the function 
$f:\{a,c\}\to \{b,d\}$
satisfying $f(a)=b$ and $f(c)=d$.
Importantly, parallel composition of simultaneous substitutions is partial. For example,
$[a{\mapsto}b] \uplus [a{\mapsto}c]$
is undefined, since there is no function $\{a\}\to\{b,c\}$ that maps $a$ simultaneously to both $b$ and $c$.

%\begin{ak}
\medskip\medskip\textbf{\large The advantages of a 2-dimensional calculus} for simultaneous substitutions over a 1-dimensional calculus are the following. 
A calculus of substitutions is an algebraic representation, up to isomorphism, of the category $\mathsf n\mathbb F$ of finite subsets of $\names$.  In a 1-dimensional calculus, operations $[a{\mapsto}b]$ have to be indexed by finite sets $S$  
\[[a{\mapsto}b]_S:S\cup \{a\}\to S\cup \{b\}\]
for sets $S$ with $a,b\notin S.$
On the other hand, in a  2-dimensional calculus with an explicit operation $\uplus$ for set-union, indexing with subsets $S$  is unnecessary. Moreover,  while the swapping 
\[\{a,b\}\to\{a,b\}\]
in the 1-dimensional calculus needs an auxiliary name such as $c$ in 
$
[a{\mapsto}c]_{\{b\}} \hcomp [b{\mapsto} a]_{\{c\}} \hcomp [c{\mapsto}a]_{\{b\}}
$
it is represented in the 2-dimensional calculus directly by 
\[
[a{\mapsto}b] \uplus [b{\mapsto}a]
\]
Finally, while it is possible to write down the equations and rewrite rules for the 1-dimensional calculus, it does not appear as particularly natural. In particular, only in the 2-dimensional calculus, will the swapping
%substitution $[a{\mapsto}c]_{\{b\}} \hcomp [b{\mapsto} a]_{\{c\}} \hcomp [c{\mapsto}a]_{\{b\}}$ 
have a simple normal form such as $[a{\mapsto}b] \uplus [b{\mapsto}a]$ (unique up to commutativity of $\uplus$).

\medskip\textbf{\large Overview. }
In order to account for partial tensors, Section~\ref{sec:internal-monoidal} %\ref{sec:nommoncat}
develops the notion of a monoidal category internal in a symmetric monoidal category. Section~\ref{sec:examples} is devoted to examples, while  Section~\ref{sec:nmts} introduces the notion of a nominal prop and Section~\ref{sec:equivalence} shows shat the categories of ordinary and of nominal props are equivalent.

\section{Setting the Scene: String Diagrams and Nominal Sets}\label{sec:scene}

We review some of the necessary terminology but need to refer to the literature for details.

\subsection{String Diagrams}

The mathematical theory of string diagrams can be formalised via $\PROP$s as defined by MacLane~\cite{maclane}. There is also the weaker notion by Lack~\cite{lack}, see Remark 2.9 of Zanasi \cite{zanasi} for a discussion.

\medskip
A $\PROP$  (\textit{products and permutation category}) is a symmetric strict monoidal category, with natural numbers as objects, where the monoidal tensor $\oplus$ is addition. Moreover, $\PROP$s, along with strict symmetric monoidal functors, that are identities on objects, form the category $\PROP$. A $\PROP$ contains all bijections between numbers as they can be be generated from the symmetry (twist) $1\oplus 1\to 1\oplus 1$ and from the parallel composition $\oplus$ and sequential composition $;$ (which we write in diagrammatic order).

%The non-symmetric versions of $\PROP$s, called $\PRO$s have also been studied and will be mentioned in this section. \begin{ak}I may need to reforumalate this\end{ak}

\medskip\noindent
$\PROP$s can be presented in algebraic form by operations and equations as \textit{symmetric monoidal theories} ($\SMT$s) \cite{zanasi}.

\medskip
An $\SMT$ $(\Sigma, E)$ has a set $\Sigma$ of generators, where each generator $\gen \in \Sigma$ is given an arity $m$ and co-arity $n$, usually written as $\gen : m \to n$ and a set $E$ of equations, which are pairs of $\Sigma$-terms. $\Sigma$-terms can be obtained by composing generators in $\Sigma$ with the unit $\id : 1 \to 1$ and symmetry $\tw : 2 \to 2$, using either the parallel or sequential composition (see Fig~\ref{fig:smt-terms}). Equations $E$ are pairs of $\Sigma$-terms with the same arity and co-arity.

\begin{figure}[h]

%\begin{tabularx}{\textwidth}{| X | X | X |}
\begin{center}
\begin{tabular}{ c c c }
\includegraphics[page=40, width=12mm]{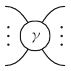} &
\qquad\qquad\includegraphics[page=41, width=12mm]{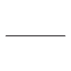}\qquad\qquad{} &
\includegraphics[page=42, width=12mm]{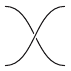} \\
\(\displaystyle\frac{}{ \gen: m \to n \in \Sigma}\) &
\(\displaystyle\frac{}{ id:1 \to 1}\) &
\(\displaystyle\frac{}{\tw : 2 \to 2}\) \\
\end{tabular}

\medskip
\medskip
\medskip

\begin{tabular}{ c c }
\includegraphics[page=39, width=60mm]{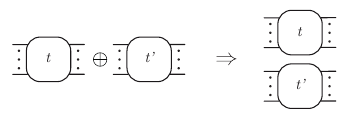} &
\includegraphics[page=38, width=70mm]{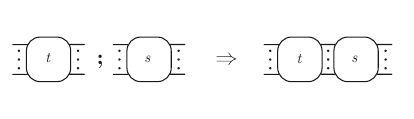} \\
\(\displaystyle\frac{ t:m\to n\quad\quad t':o\to p}{ t \oplus t' :  m+o\to n+p}\) &
\(\displaystyle\frac{ t:m\to n\quad\quad s:n\to o}{ t\hcomp s :  m \to o}\)
\end{tabular}
\end{center}
% \begin{gather*}
% \frac{}{ \gamma: m \to n \in \Sigma}
% \qquad
% \frac{}{ id:1 \to 1}
% \qquad
% \frac{}{\sigma : 2 \to 2}
% \\[1ex]
% \frac{ t:m\to n\quad\quad t':o\to p}{ t \oplus t' :  m+o\to n+p}\qquad
% \frac{ t:m\to n\quad\quad s:n\to o}{ t\hcomp s :  m \to o}
% % \\[1ex]
% % \frac{}{\delta_{ab}:\{a\}\to\{b\}}
% %\frac{ t:A\to B\uplus\{b\}}{\vdash t\then (\id_B\uplus \delta_{bc}) :  A \to B\uplus\{c\}}
% %\quad\quad\quad\quad
% %\frac{ t:\{a\}\uplus A\to B}{\vdash (\delta_{ca}\uplus\id_A)\then t  :  \{c\}\uplus A \to B}
% \end{gather*}
\caption{SMT Terms}\label{fig:smt-terms}
% \end{tcolorbox}
\end{figure}%

\noindent
Given an $\SMT$, we can freely generate a $\PROP$, by taking $\Sigma$-terms as arrows, modulo the equations of Fig~\ref{fig:symmetric-monoidal-category}, together with the smallest congruence (with respect to the two compositions) of equations in $E$.

\begin{figure}[h]
\[
\begin{array}{cc}
\id_m \hcomp t =  t = t \hcomp \id_n \qquad{}& \qquad
id_0 \oplus t = t = t \oplus id_0 \\[1ex]
(t\hcomp s)\hcomp r = t \hcomp (s \hcomp r) \qquad{}& \qquad
(t \oplus s) \oplus r = t \oplus (s \oplus r) \\[1ex]
\tw_{1,1} \hcomp \tw_{1,1} = id_2 \qquad{}& \qquad
(s \hcomp t) \oplus (u \hcomp v) = (s \oplus u) \hcomp (t \oplus v)
\\[1ex]
\multicolumn{2}{c}{(t \oplus id_z) \hcomp \tw_{n,z} = \tw_{m,z} \hcomp (id_z \oplus t)} 
\end{array}
\]
\caption{Equations of symmetric monoidal categories}\label{fig:symmetric-monoidal-category}
% \end{tcolorbox}
\end{figure}

\medskip
$\PROP$s admit a nice graphical presentation, wherein the sequential composition is modeled by horizontal composition of diagrams, and parallel/tensor composition is vertical stacking of diagrams (see Fig~\ref{fig:smt-terms}). We now present the $\SMT$s of \colorbox{black!30}{\strut bijections $\mathbb B$}, \colorbox{orange!60}{\strut injections $\mathbb I$}, \colorbox{green!40}{\strut surjections $\mathbb S$}, \colorbox{cyan!40}{\strut functions $\mathbb F$}, \colorbox{magenta!40}{\strut partial functions $\mathbb P$}, \colorbox{cyan!60!magenta!60}{\strut relations $\mathbb R$} and \colorbox{yellow!60}{\strut monotone maps $\mathbb{M}$}.\footnote{The theory of \colorbox{yellow!60}{\strut monotone maps $\mathbb{M}$} does not include equations involving the symmetry $\tw$ and is in fact presented by a so-called $\PRO$ rather than a $\PROP$. However, in this paper we will only be dealing with theories presented by $\PROP$s (the reason why this is the case is illustrated in the proof of Proposition~\ref{prop:ORD}).}
The diagram in Fig~\ref{fig:smt-theories} shows the generators and the equations that need to be added to the empty $\SMT$, to get a presentation of the given theory. 
To ease comparison with the corresponding nominal monoidal theories in Fig~\ref{fig:nmt-theories} later we also added on a  \stripbox{\strut striped} background the equations for wire-crossings that are already implied by the naturality of symmetries, that is, the last equation of Fig~\ref{fig:symmetric-monoidal-category}. These are the equations that are part of the definition of a prop in the sense of MacLane~\cite{maclane} but need to be added explicitely to the props in the sense of Lack~\cite{lack}.

%
%That is, except for the case of \colorbox{black!30}{\strut bijections $\mathbb B$}, which an empty $\SMT$ already presents, as evidenced by the equations in Fig~\ref{fig:symmetric-monoidal-category}. These equations are included in Fig~\ref{fig:smt-theories} again, for the sake of completeness. 

\begin{figure}
% \begin{tcolorbox}
\includegraphics[page=36, width=\linewidth]{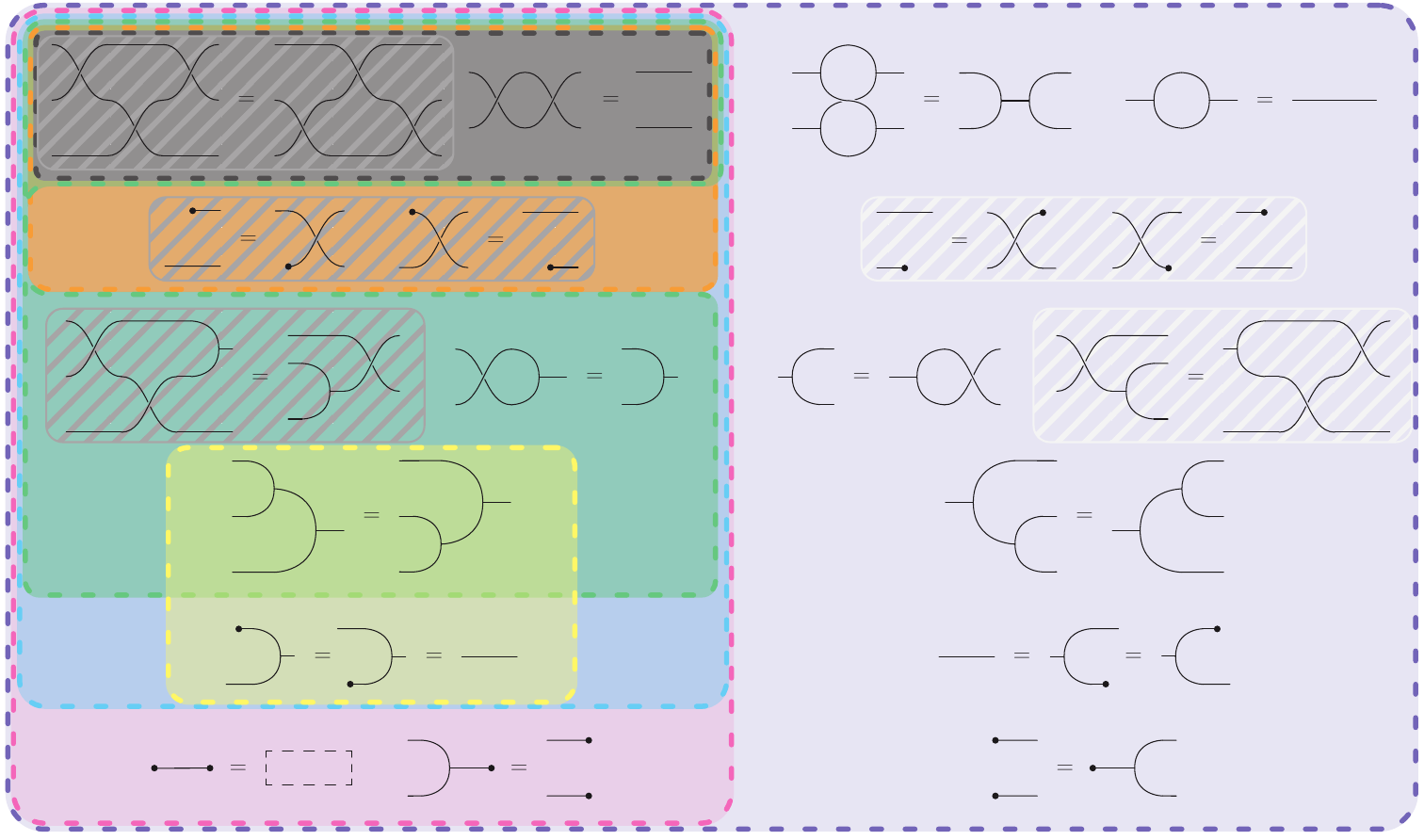}
\caption{Symmetric monoidal theories}\label{fig:smt-theories}
% \end{tcolorbox}
\end{figure}

\subsection{Nominal Sets}
Let $\names$ be a countably infinite set of `names` or `atoms`. Let $\mathfrak S$ be the group of finite\footnote{A permutation is called finite if it is generated by finitely many transpositions.} permutations $\names\to\names$. An element $x\in X$ of a group action $\mathfrak S\times X\to X$ is supported by $S\subseteq\names$ if $\pi\cdot x= x$ for all $\pi\in\mathfrak S$ such that $\pi$ restricted to $S$ is the identity. A group action $\mathfrak S\times X\to X$ such that all elements of $X$ have finite support is called a \emph{nominal set}. 
We write $\supp(x)$ for the minimal support of $x$ and $\Nom$ for the category of nominal sets, which has as maps the \emph{equivariant} functions, that is, those functions that respect the permutation action. Our main example is the category of simultaneous substitutions:

\renewcommand{\Fun}{{\mathsf n\mathbb F}}
\begin{example}[$\Fun$]\label{exle:nF}
We denote by $\Fun$ the category of finite subsets of $\names$ with all functions. While $\Fun$ is a category, it also carries additional nominal structure. In particular, both the set of objects and the set of arrows are nominal sets. with $\supp(A)=A$ and $\supp(f)=A\cup B$ for  $f:A\to B$. The categories of injections, surjections, bijections, partial functions and relations are further examples along the same lines.
\end{example}

\begin{gray}\begin{ak}we may want to refer here to the subsection on string diagrams\end{ak}
\begin{remark}
say that $\mathbb B$ bijections, $\mathbb F$ functions, injections $\mathbb I$, surjections $\mathbb S$, partial functions $\mathbb P$ and relations $\mathbb R$ are equivalent to $\Fun$ etc as categories but that these equivalences are not equivariant ... should we use 
\[\mathsf n\mathbb B, \mathsf n\mathbb F, \mathsf n\mathbb I, \mathsf n\mathbb S, \mathsf n\mathbb P, \mathsf n\mathbb R\]
as the corresponding names for the nominal versions?
\end{remark}
\end{gray}

\section{Internal monoidal categories}\label{sec:internal-monoidal}

We introduce the notion of an internal monoidal category. Given a symmetric monoidal category $(\mathcal V,I,\otimes)$ with finite limits, we are interested in categories $\mathbb C$, internal in $\mathcal V$, that carry a monoidal structure not of type $\mathbb C\times \mathbb C\to \mathbb C$ but of type $\mathbb C\otimes \mathbb C\to \mathbb C$. This will allow us to account for the partiality of $\uplus$ discussed in the introduction:

\begin{example}
\begin{itemize}
\item The symmetric monoidal (closed) category $(\Nom,1,\ast)$ of nominal sets with the separated product $\ast$ is defined as follows \cite{pitts}. $1$ is the terminal object, ie, a singleton with empty support. The separated product of two nominal sets is defined as $A\ast B = \{(a,b)\in A\times B \mid \supp(a)\cap\supp(b)=\emptyset\}$.
\item The category $\Fun$ (and its relatives) of Example~\ref{exle:nF} is an internal monoidal category with   monoidal operation given by  $A\uplus B=A\cup B$ if $A$ and $B$ are disjoint. 
\end{itemize}
\end{example}

\medskip\noindent
$(\Fun,\emptyset,\uplus)$ as defined in the previous example is not a monoidal category, since $\uplus$, being partial, is not an operation of type $\Fun\times\Fun\to\Fun$  . 
The purpose of this section is to show that $(\Fun,\emptyset,\uplus)$ is an internal monoidal category in $(\Nom,1,\ast)$ with $\uplus$ of type 

\[\uplus:\Fun\ast\Fun\to\Fun.\]

To this end we need to extend $\ast:\Nom\times\Nom\to\Nom$ to 

\[\ast:\Cat(\Nom)\times\Cat(\Nom)\to\Cat(\Nom)\] 

where we denote by $\Cat(\Nom)$, the category of (small) internal categories in $\Nom$. 

\medskip
The necessary (and standard) notation from internal categories is reviewed in Appendix A.

\begin{remark}
Let $\mathbb C$ be an internal category in a symmetric monoidal category $(\mathcal V,I,\otimes)$ with finite limits. Since $\otimes$ need not preserve finite limits, we cannot expect that defining $(\mathbb C\otimes\mathbb C)_0=\mathbb C_0\otimes\mathbb C_0$ and $(\mathbb C\otimes\mathbb C)_1=\mathbb C_1\otimes\mathbb C_1$ results in $\mathbb C\otimes\mathbb C$ being an internal category.
\end{remark}

Consequently, putting $(\mathbb C\otimes\mathbb C)_1=\mathbb C_1\otimes\mathbb C_1$ does not extend $\otimes$ to an operation $\Cat(\mathcal V)\times\Cat(\mathcal V)\to\Cat(\mathcal V)$. To show what goes wrong in a concrete instance is the purpose of the next example.

\begin{example} 
%Consider the category $\Fun$ of finite sets of names with all functions as a category internal in $(\Nom,1,\ast)$. 
Define a binary operation $\Fun\ast\Fun$ as $(\Fun\ast\Fun)_0=\Fun_0\ast\Fun_0$ and $(\Fun\ast\Fun)_1=\Fun_1\ast\Fun_1$. Then $\Fun\ast\Fun$ cannot be equipped with the structure of an internal category. Indeed, assume for a contradiction that there was an appropriate pullback $(\Fun\ast \Fun)_2$ and arrow $\comp$ such that the two diagrams commute: 
\[
\xymatrix@C=20ex{  
(\Fun\ast \Fun)_2 \  
\ar[0,1]|-{\ \comp\ }  
\ar[dd]_{\pi_1}^{\pi_2}
&  
{\ \ \Fun_1\ast \Fun_1 \ \ }   
%\ar@<1ex>[0,1]^-{\dom\times \dom}  
%\ar@<-1ex>[0,1]_-{\cod\times \cod}  
\ar[dd]_{\dom}^{\cod}
%&  
%\ \Fun_0\ast \Fun_0  
%\ar[dd]_{\uplus_0}
%
\\
&
\\
\Fun_1\ast\Fun_1\  
\ar[0,1]^{\ \dom\ }_{\cod}  
&  
{\ \ \Fun_0\ast\Fun_0 \ \ }   
%\ar@<1ex>[0,1]^-{\dom}  
%\ar@<-1ex>[0,1]_-{\cod}  
%&  
%\ \Fun_0  
}  
\]  
Let $\delta_{xy}:\{x\}\to\{y\}$ be the unique function in $\Fun$ of type $\{x\}\to\{y\}$. Then $((\delta_{ac},\delta_{bd}), (\delta_{cb},\delta_{da}))$, which can be depicted as
\[
\xymatrix@R=0.5ex{
\{a\} \ar[r]^{\delta_{ac}} & \{c\} \ar[r]^{\delta_{cb}} & \{b\}\\
\{b\} \ar[r]_{\delta_{bd}} & \{d\} \ar[r]_{\delta_{da}} & \{a\}
}
\]
 is in the pullback $(\Fun\ast \Fun)_2$, but there is no $\comp$ such that the two squares above commute, since $\comp((\delta_{ac},\delta_{bd}), (\delta_{cb},\delta_{da}))$ would have to be $(\delta_{ab},\delta_{ba})$, which do not have disjoint support and therefore are not in $\Fun_1\ast \Fun_1$. %What goes wrong here? What if we weaken the requirement  on $\delta_{ab},\delta_{ba}$ to have disjoint support to the requirement that the domains and the codomains are disjoint?
 \qed
\end{example}

The solution to the problem consists in assuming that the given symmetric monoidal category with finite limits $(\mathcal V,1,\otimes)$  is semi-cartesian (aka affine), that is, the unit $1$ is the terminal object. In such a category there are canonical 
\[j:A\otimes B\to A\times B\]
and we can use them to define arrows $j_1:(\mathbb C\otimes \mathbb C)_1\to\mathbb C_1\times \mathbb C_1$ that give us the right notion of tensor on arrows. From our example $\Fun$ above, we know that we want arrows $(f,g)$ to be in $(\mathbb C\otimes \mathbb C)_1$ if $\dom(f)\cap\dom(g)=\emptyset$ and $\cod(f)\cap\cod(g)=\emptyset$. We now turn this observation into a category theoretic definition.

\medskip
Let $\mathbb C$ and  $\mathbb D$ be internal categories in $\mathcal V$. Our first task is to define $(\mathbb C\otimes \mathbb D)_1$. This is accomplished by stipulating that  $(\mathbb C\otimes \mathbb D)_1$ is the limit in the diagram below

% \vspace{-5em}
%  \quad\includegraphics[angle=-90,width=\linewidth]{img/doublepullback2.pdf}
%  \vspace{-6em}
\begin{center}
\begin{tikzpicture}
  \matrix (m) [matrix of math nodes, row sep=1em, row 2/.style= {yshift=-6em}, column sep=0em, color=black]{
                                  & |[color=cyan]|(\mathbb C\otimes\mathbb D)_1 &                               &                                 & \mathbb C_1\times\mathbb D_1 & \\
                                  &                               & \mathbb C_0\otimes\mathbb D_0 &                                 &                               & \mathbb C_0\times\mathbb D_0  \\
    \mathbb C_0\otimes\mathbb D_0 &                               &                               & \mathbb C_0\times\mathbb D_0   &                               & \\};
  \path[-stealth, color=black]
    (m-1-2) edge [color=cyan] node [above] {$j_1$} (m-1-5) 
            edge [color=cyan, densely dotted] node [right, xshift=0.7em,yshift=-1.25em] {$\cod_{(\mathbb C\otimes\mathbb D)_1}$} (m-2-3)
            edge [color=cyan] node [left] {$\dom_{(\mathbb C\otimes\mathbb D)_1}$} (m-3-1)
    (m-2-3) edge [densely dotted] node [above, xshift=2.5em] {$j$} (m-2-6)
    (m-3-1) edge node [below] {$j$} (m-3-4)
    (m-1-5) edge [densely dotted] node [right, xshift=-0.7em,yshift=1.25em] {$\cod_{\mathbb C_1} \times \cod_{\mathbb D_1}$} (m-2-6)
            edge [-,line width=6pt,draw=white] (m-3-4) edge node [left,xshift=1.2em, yshift=2.5em] {$\dom_{\mathbb C_1} \times \dom_{\mathbb D_1}$} (m-3-4);
\end{tikzpicture}
\end{center}

\noindent
In the following we abbreviate the diagram above to
\begin{equation}
\label{equ:j1}
\vcenter{
\xymatrix@C=9ex{
(\mathbb C\otimes\mathbb D)_1 \ar[rr]^{j_1}
\ar@<-1ex>[d]_{\dom}
\ar@<1ex>[d]^{\cod}
&& \mathbb C_1\times\mathbb D_1 
\ar@<-1ex>[d]_{\dom\times \dom}
\ar@<1ex>[d]^{\cod\times \cod}
\\
(\mathbb C\otimes\mathbb D)_0 \ar[rr]_j
&& \mathbb C_0\times\mathbb D_0 
}}
\end{equation}

We are now in the position to extend the monoidal operation $\otimes:\mathcal V\times\mathcal V\to\mathcal V$ to a monoidal operation $\otimes:Cat(\mathcal V)\times Cat(\mathcal V)\to Cat(\mathcal V)$.

\begin{definition}\label{def:internal-tensor}
Let $(\mathcal V,1,\otimes)$ be a monoidal category where the unit is the terminal object. The operation $\otimes:Cat(\mathcal V)\times Cat(\mathcal V)\to Cat(\mathcal V)$ is defined as follows.
\begin{itemize}
\item $(\mathbb C\otimes\mathbb D)_0$ and $(\mathbb C\otimes\mathbb D)_1$ and $\cod,\dom: (\mathbb C\otimes\mathbb D)_1\to (\mathbb C\otimes\mathbb D)_0$ as in the diagram above.
\item $i:(\mathbb C\otimes\mathbb D)_0\to (\mathbb C\otimes\mathbb D)_1$ is the arrow into the limit $(\mathbb C\otimes\mathbb D)_1$ given by

\[
\xymatrix@C=9ex{
(\mathbb C\otimes\mathbb D)_0 
\ar@/^/[rrrd]^{(i\times i)\circ j} 
\ar@/_/@<-1ex>[rdd]|-{\ \id \ }
\ar@/_/@<-3ex>[rdd]|-{\ \id \ }
\ar@{..>}[rd]|-{\ i\ }
&&&
\\
&(\mathbb C\otimes\mathbb D)_1 \ar[rr]|-{\ j_1 \ }
\ar@<-1ex>[d]_{\cod}
\ar@<1ex>[d]^{\dom}
&& \mathbb C_1\times\mathbb D_1 
\ar@<-1ex>[d]_{\cod\times \cod}
\ar@<1ex>[d]^{\dom\times \dom}
\\
&{\ \ (\mathbb C\otimes\mathbb D)_0} \ar[rr]_j
&& \mathbb C_0\times\mathbb D_0 
}
\]
from which one reads off 
\[\dom\circ i  = \id_{(\mathbb C\otimes\mathbb D)_0} = \cod\circ i\]
\item $(\mathbb C\otimes\mathbb D)_2$ is the pullback 
\[
\xymatrix{
&(\mathbb C\otimes\mathbb D)_2\ar[dl]_{\pi_1}\ar[dr]^{\pi_2}&
\\
(\mathbb C\otimes\mathbb D)_1\ar[dr]_{\cod}
&& 
(\mathbb C\otimes\mathbb D)_1 \ar[dl]^{\dom}
\\
&(\mathbb C\otimes\mathbb D)_0&
}
\]
Recalling the definition of $j_1$ from \eqref{equ:j1}, there is also a corresponding $j_2:(\mathbb C\otimes\mathbb D)_2\to\mathbb C_2\times\mathbb D_2$ due to the fact that the product of pullbacks is a pullback of products
\begin{equation}\label{equ:j2}
\vcenter{
\xymatrix@C=2ex{
&(\mathbb C\otimes\mathbb D)_2\ar[dl]_{\pi_1}\ar[dr]^{\pi_2}\ar[rrr]^{j_2}&
&
&\mathbb C_2\times\mathbb D_2\ar[dl]_{\pi_1\times\pi_1}\ar[dr]^{\pi_2\times\pi_2}&
\\
(\mathbb C\otimes\mathbb D)_1\ar[dr]_{\cod}\ar@/^1.2pc/@{..>}[rrr]|-{ j_1 }
&& 
(\mathbb C\otimes\mathbb D)_1 \ar[dl]^{\dom}\ar@/_1.2pc/@{..>}[rrr]|-{ j_1 }
&
\mathbb C_1\times\mathbb D_1\ar[dr]_{\cod\times\cod}
&& 
\mathbb C_1\times\mathbb D_1 \ar[dl]^{\dom\times\dom}
\\
&(\mathbb C\otimes\mathbb D)_0\ar[rrr]^{j}&
&
&\mathbb C_0\times\mathbb D_0&
}}
\end{equation}
 Recall the definition of the limit $(\mathbb C\otimes\mathbb D)_1$ from \eqref{equ:j1}. 
Then $\comp:(\mathbb C\otimes\mathbb D)_2\to(\mathbb C\otimes\mathbb D)_1$ is the arrow into $(\mathbb C\otimes\mathbb D)_1$
\begin{equation}\label{equ:comp}
\vcenter{
\xymatrix@C=9ex{
(\mathbb C\otimes\mathbb D)_2 
\ar@/^/[rrrd]^{\ \ \ \ \ \ \ \ (\comp\times \comp)\circ j_2} 
\ar@/_/@<-0ex>[rdd]|-{\ \ \ \dom\circ\pi_1 \ }
\ar@/_/@<-3ex>[rdd]|-{\ \cod\circ\pi_2 \ }
\ar@{..>}[rd]|-{\ \comp\ }
&&&
\\
&(\mathbb C\otimes\mathbb D)_1 \ar[rr]|-{\ j_1\ }
\ar@<-1ex>[d]_{\cod}
\ar@<1ex>[d]^{\dom}
&& \mathbb C_1\times\mathbb D_1 
\ar@<-1ex>[d]_{\cod\times \cod}
\ar@<1ex>[d]^{\dom\times \dom}
\\
&{\ \ (\mathbb C\otimes\mathbb D)_0} \ar[rr]_j
&& \mathbb C_0\times\mathbb D_0 
}}
\end{equation}
from which one reads off 

\[\dom\circ \comp=\dom\circ \pi_1\ \quad\quad \ \cod\circ 
  \comp=\cod\circ \pi_2\]
\item The equations $\comp\circ \langle i\circ\dom,\id_{(\mathbb C\otimes\mathbb D)_1} \rangle = \id_{(\mathbb C\otimes\mathbb D)_1} = \comp\circ \langle\id_{(\mathbb C\otimes\mathbb D)_1},i\circ\cod \rangle$ are proved in Proposition~\ref{prop:comp-i}.

\item The equation $\comp\,\circ\,\compl= \comp\,\circ\,\compr$ will be shown in Proposition~\ref{prop:comp-assoc}.
\end{itemize}
\end{definition}

This ends the definition of $\mathbb C\otimes\mathbb D$ and the next few pages are devoted to showing that it is indeed an internal category. To prove the next propositions, we will need the following lemma, which can be skipped for now. It is a consequence of the general fact that the isomorphism  $[\mathcal I,\mathcal C](K_A,D)\cong\mathcal C(A,\lim D)$ defining limits is natural in $A$ and $D$.
\begin{lemma}\label{lem:hh'}
If in the diagram
\[
\xymatrix@C=9ex{
T 
\ar[rr]^{ k } 
\ar@/_4pc/@<-2ex>[dd]_{f_1}
\ar@/_4pc/@<-0ex>[dd]^{f_2}
\ar@<-0ex>[d]^{h}
&&
P
\ar@/^4pc/@<1ex>[dd]_{f'_1}
\ar@/^4pc/@<3ex>[dd]^{f'_2}
\ar@<-0ex>[d]^{h'}
\\
(\mathbb C\otimes\mathbb D)_2 
\ar[rr]^{ j_2 } 
\ar@<-1ex>[d]_{\pi_1}
\ar@<1ex>[d]^{\pi_2}
&&
\mathbb C_2\times\mathbb D_2 
\ar@<-1ex>[d]_{\pi_1\times\pi_1}
\ar@<1ex>[d]^{\pi_2\times\pi_2}
\\
(\mathbb C\otimes\mathbb D)_1 \ar[rr]^{ j_1 } 
\ar@<-1ex>[d]_{\cod}
\ar@<1ex>[d]^{\dom}
&& 
\mathbb C_1\times\mathbb D_1 
\ar@<-1ex>[d]_{\cod\times \cod}
\ar@<1ex>[d]^{\dom\times \dom}
\\
{\ \ (\mathbb C\otimes\mathbb D)_0} \ar[rr]_j
&& \mathbb C_0\times\mathbb D_0 
}
\]
$f$ and $f'$ are cones commuting with $j_1$ and $k$, that is, if
\begin{align}
\label{equ:lemjh:1} \cod\circ f_1 &= \dom\circ f_2 \\ 
\label{equ:lemjh:2} (\cod\times\cod)\circ f'_1 &= (\dom\times\dom)\circ f'_2 \\
\label{equ:lemjh:3} j_1\circ f_i &= f'_i\circ k
\end{align}
and $h,h'$ are the respective unique arrows into the pullbacks, then also 
\[h'\circ k=j_2\circ h\]
holds.
\end{lemma}

\begin{longVersion}
\begin{proof}
It suffices to calculate, in a slightly abbreviated style, 
$\pi\circ h'\circ k 
= f'\circ k 
= j_1\circ f 
= j_1\circ \pi\circ h
= \pi\circ j_2\circ h$
where the last equality is due to the definition of $j_2$ given in \eqref{equ:j2}. This implies $h'\circ k=j_2\circ h$.
\end{proof}
\end{longVersion}

\noindent
Using the lemma, the next two propositions have reasonably straight forward proofs.

\begin{proposition}\label{prop:comp-i}
$\comp\circ \langle i\circ\dom,\id_{(\mathbb C\otimes\mathbb D)_1} \rangle = \id_{(\mathbb C\otimes\mathbb D)_1} = \comp\circ \langle\id_{(\mathbb C\otimes\mathbb D)_1},i\circ\cod \rangle$.
\end{proposition}

\begin{longVersion}
\begin{proof}
We show the first equality. According to the definition of the limit $(\mathbb C\otimes\mathbb D)_1$ given in \eqref{equ:j1}, it suffices to show 
\begin{align*}
\cod\circ\comp\circ \langle i\circ\dom,\id_{(\mathbb C\otimes\mathbb D)_1} \rangle & = \cod\\
\dom\circ\comp\circ \langle i\circ\dom,\id_{(\mathbb C\otimes\mathbb D)_1} \rangle & = \dom\\
j_1\circ\comp\circ \langle i\circ\dom,\id_{(\mathbb C\otimes\mathbb D)_1} \rangle & = j_1
\end{align*}
The first two follow from $\dom\circ \comp=\dom\circ \pi_1$ and $\cod\circ 
  \comp=\cod\circ \pi_2$.
Recalling $\comp\circ\langle i\circ\dom,\id_{(\mathbb C\otimes\mathbb D)_1}\rangle=\id_{(\mathbb C\otimes\mathbb D)_1}$ and, hence,  $(\comp\times\comp)\circ\langle (i\times i)\circ(\dom\times\dom),\id_{\mathbb C_1\times\mathbb D_1} \rangle=\id_{\mathbb C_1\times\mathbb D_1}$, the third follows because the two rectangles below commute.
\begin{equation}\label{eq:prop:comp-i}
\vcenter{
\xymatrix@C=9ex{
(\mathbb C\otimes\mathbb D)_1 
\ar[rr]^{ j_1 } 
\ar@<-0ex>[d]_{\langle i\circ\dom,\id_{(\mathbb C\otimes\mathbb D)_1} \rangle}
&&
\mathbb C_1\times\mathbb D_1 
\ar@<-0ex>[d]^{\langle (i\times i)\circ(\dom\times\dom),\id_{\mathbb C_1\times\mathbb D_1} \rangle}
\\
(\mathbb C\otimes\mathbb D)_2 
\ar[rr]^{ j_2 } 
\ar[d]_{\comp}
&&
\mathbb C_2\times\mathbb D_2 
\ar[d]^{\comp\times\comp}
\\
(\mathbb C\otimes\mathbb D)_1 \ar[rr]^{ j_1 } 
&& 
\mathbb C_1\times\mathbb D_1 
}}
\end{equation}
The lower rectangle commutes by the definition of $\comp$, see \eqref{equ:comp}.
To show that the upper rectangle commutes, we instantiate the lemma with $k=j_1$ and
$f_1=i\circ\dom$ and $f_2=\id_{(\mathbb C\otimes\mathbb D)_1}$ and $h=\langle i\circ\dom,\id_{(\mathbb C\otimes\mathbb D)_1} \rangle$ and $f_1'=(i\times i)\circ(\dom\times\dom)$ and $f_2'=\id_{\mathbb C_1\times\mathbb D_1}$ and $h'=\langle (i\times i)\circ(\dom\times\dom),\id_{\mathbb C_1\times\mathbb D_1} \rangle$. Equations \eqref{equ:lemjh:1} and \eqref{equ:lemjh:2} are straightforward to check. 
%
%For \eqref{equ:lemjh:3}, we first verify  $(\cod\times\cod)\circ j_1\circ f_i=(\cod\times\cod)\circ f'_i\circ j_1$ and $(\dom\times\dom)\circ j_1\circ f_i=(\dom\times\dom)\circ f'_i\circ j_1$. For example 
%\begin{align*}
%(\cod\times\cod)\circ j_1\circ f_1 
%&= j\circ\cod\circ i\circ\dom\\
%&= j\circ\dom\\
%&= (\dom\times\dom)\circ j_1\\
%&= (\cod\times\cod)\circ(i\times i)\circ(\dom\times\dom)\circ j_1\\
%& =(\cod\times\cod)\circ f'_1\circ j_1
%\end{align*}
\begin{ak} 
For \eqref{equ:lemjh:3}, we first verify   
\begin{align*}
j\circ\cod\circ f_i & =(\cod\times\cod)\circ f'_i\circ j_1\\
j\circ\dom\circ f_i& =(\dom\times\dom)\circ f'_i\circ j_1
\end{align*}
For example,
\begin{align*}
j\circ\cod\circ f_1
&= j\circ\cod\circ i\circ\dom\\
&= j\circ\dom\\
&= (\dom\times\dom)\circ j_1\\
&= (\cod\times\cod)\circ(i\times i)\circ(\dom\times\dom)\circ j_1\\
& =(\cod\times\cod)\circ f'_1\circ j_1
\end{align*}
and the other cases are similar (CHECK). 
It follows that $(\cod\circ f_1, \dom\circ f_1, f_1'\circ j_1)$ is a cone over the limit \eqref{equ:j1}. Hence the `triangle'  $j_1\circ f_1=f'_1\circ j_1$ commutes. Similarly (CHECK), the triangle $j_1\circ f_2=f'_2\circ j_1$ commutes. We now verified the assumptions of the lemma and conclude that that the upper rectangle of Diagram \eqref{eq:prop:comp-i} commutes, 
\end{ak}
which was all that was left to show $\comp\circ \langle i\circ\dom,\id_{(\mathbb C\otimes\mathbb D)_1} \rangle = \id_{(\mathbb C\otimes\mathbb D)_1}$.  The second equality $\id_{(\mathbb C\otimes\mathbb D)_1} = \comp\circ \langle\id_{(\mathbb C\otimes\mathbb D)_1},i\circ\cod \rangle$ is proved similarly (CHECK).

\begin{sam}
Again, we need to show:
\begin{align*}
\cod \circ \comp \circ \langle\id_{(\mathbb C\otimes\mathbb D)_1},i\circ\cod \rangle & = \cod\\
\dom \circ \comp \circ \langle\id_{(\mathbb C\otimes\mathbb D)_1},i\circ\cod \rangle & = \dom\\
j_1  \circ \comp \circ \langle\id_{(\mathbb C\otimes\mathbb D)_1},i\circ\cod \rangle & = j_1
\end{align*}

We have:
\begin{align*}
\cod \circ \comp \circ \langle\id_{(\mathbb C\otimes\mathbb D)_1},i\circ\cod \rangle & = 
\cod \circ \pi_2 \circ \langle\id_{(\mathbb C\otimes\mathbb D)_1},i\circ\cod \rangle \\
& = \cod \circ i\circ\cod \\
& = \cod
\end{align*}

(Similarly for $\dom \circ \comp \circ \langle\id_{(\mathbb C\otimes\mathbb D)_1},i\circ\cod \rangle = \dom$.)

For the third equation, we have a diagram similar to the one above, where we need to check that the top square commutes. We instantiate the lemma and get the following diagram :

% \[
% \xymatrix@C=9ex{
% (\mathbb C\otimes\mathbb D)_1 
% \ar[rr]^{ j_1 } 
% \ar@<-0ex>[d]_{\langle\id_{(\mathbb C\otimes\mathbb D)_1},i\circ\cod \rangle}
% &&
% \mathbb C_1\times\mathbb D_1 
% \ar@<-0ex>[d]^{\langle \id_{\mathbb C_1\times\mathbb D_1} , (i\times i)\circ(\cod\times\cod) \rangle}
% \\
% (\mathbb C\otimes\mathbb D)_2 
% \ar[rr]^{ j_2 } 
% \ar[d]_{\comp}
% &&
% \mathbb C_2\times\mathbb D_2 
% \ar[d]^{\comp\times\comp}
% \\
% (\mathbb C\otimes\mathbb D)_1 \ar[rr]^{ j_1 } 
% && 
% \mathbb C_1\times\mathbb D_1 
% }
% \]

\[
\xymatrix@C=9ex{
T = (\mathbb C\otimes\mathbb D)_1 
\ar[rr]^{ k = j_1 } 
\ar@/_7pc/@<-2ex>[dd]_{f_1 = \id_{(\mathbb C\otimes\mathbb D)_1}}
\ar@/_7pc/@<-0ex>[dd]^{f_2 = i\circ\cod}
\ar@<-0ex>[d]^{h = \langle\id_{(\mathbb C\otimes\mathbb D)_1},i\circ\cod \rangle}
&&
P = \mathbb C_1\times\mathbb D_1 
\ar@/^7pc/@<1ex>[dd]_{f'_1 = \id_{\mathbb C_1\times\mathbb D_1}}
\ar@/^7pc/@<3ex>[dd]^{f'_2 = (i\times i)\circ(\cod\times\cod)}
\ar@<-0ex>[d]^{h' = \langle \id_{\mathbb C_1\times\mathbb D_1} , (i\times i)\circ(\cod\times\cod) \rangle}
\\
(\mathbb C\otimes\mathbb D)_2 
\ar[rr]^{ j_2 } 
\ar@<-1ex>[d]_{\pi_1}
\ar@<1ex>[d]^{\pi_2}
&&
\mathbb C_2\times\mathbb D_2 
\ar@<-1ex>[d]_{\pi_1\times\pi_1}
\ar@<1ex>[d]^{\pi_2\times\pi_2}
\\
(\mathbb C\otimes\mathbb D)_1 \ar[rr]^{ j_1 } 
\ar@<-1ex>[d]_{\cod}
\ar@<1ex>[d]^{\dom}
&& 
\mathbb C_1\times\mathbb D_1 
\ar@<-1ex>[d]_{\cod\times \cod}
\ar@<1ex>[d]^{\dom\times \dom}
\\
{\ \ (\mathbb C\otimes\mathbb D)_0} \ar[rr]_j
&& \mathbb C_0\times\mathbb D_0 
}
\]

Thus to check that the top square commutes, it is enough to show the following equations hold:
\begin{align*}
\cod\circ \id_{(\mathbb C\otimes\mathbb D)_1} &= \dom\circ i\circ\cod \\ 
(\cod\times\cod)\circ \id_{\mathbb C_1\times\mathbb D_1} &= (\dom\times\dom)\circ (i\times i)\circ(\cod\times\cod) \\
j_1\circ f_i &= f'_i\circ j_1
\end{align*}

% \begin{align*}
% \cod\circ \id_{(\mathbb C\otimes\mathbb D)_1} & = \cod = \dom\circ i\circ\cod
% \end{align*}
\end{sam}
\end{proof}
\end{longVersion}

\begin{proposition}\label{prop:comp-assoc}
$\comp\circ\compl = \comp\circ\compr$ 
\end{proposition}

\begin{longVersion}
\begin{proof}
To show that composition is associative, we need to recall the definition of $\compl$ and $\compr$ from Remark~\ref{def:internal-cat}, which leads us to consider
\[
\xymatrix@C=9ex{
(\mathbb C\otimes\mathbb D)_3 
\ar[rr]^{ j_3 } 
\ar@/_5pc/@<-2ex>[d]_{\compl}
\ar@/_5pc/@<-0ex>[d]^{\compr}
\ar@<-1ex>[d]_{\leftt}
\ar@<1ex>[d]^{\rightt}
&&
\mathbb C_3\times\mathbb D_3 
\ar@/^7pc/@<2ex>[d]_{\compl\times\compl\ }
\ar@/^7pc/@<4ex>[d]^{\compr\times\compr}
\ar@<-1ex>[d]_{\leftt}
\ar@<1ex>[d]^{\rightt}
\\
(\mathbb C\otimes\mathbb D)_2 
\ar@/_4pc/@<-0ex>[d]_{\comp}
\ar[rr]^{ j_2 } 
\ar@<-1ex>[d]_{\pi_1}
\ar@<1ex>[d]^{\pi_2}
&&
\mathbb C_2\times\mathbb D_2 
\ar@/^5pc/@<-0ex>[d]^{\comp\times\comp}
\ar@<-1ex>[d]_{\pi_1\times\pi_1}
\ar@<1ex>[d]^{\pi_2\times\pi_2}
\\
(\mathbb C\otimes\mathbb D)_1 \ar[rr]^{ j_1 } 
\ar@<-1ex>[d]_{\cod}
\ar@<1ex>[d]^{\dom}
&& 
\mathbb C_1\times\mathbb D_1 
\ar@<-1ex>[d]_{\cod\times \cod}
\ar@<1ex>[d]^{\dom\times \dom}
\\
{\ \ (\mathbb C\otimes\mathbb D)_0} \ar[rr]_j
&& \mathbb C_0\times\mathbb D_0 
}
\]
where, in analogy with the definition of $j_2$ in \eqref{equ:j2}, $j_3$ is defined as the unique arrow into the pullback $\mathbb C_3\times\mathbb D_3$.

\medskip
The rectangles commute by definition of the $j_i$. To show $\comp\circ\compl=\comp\circ\compr$ it suffices to show
\begin{align*}
\cod\circ\comp\circ\compl& =\cod\circ\comp\circ\compr\\
\dom\circ\comp\circ\compl& =\dom\circ\comp\circ\compr\\
j_1\circ\comp\circ\compl& =j_1\circ\comp\circ\compr
\end{align*}

For the first, we calculate
\begin{align*}
\cod\circ\comp\circ\compl
& =\cod\circ\pi_2\circ\compl& \cod\circ\comp=\cod\circ\pi_2\\
& =\cod\circ\pi_2\circ\rightt & \text{Def of $\compl$}\\
& =\cod\circ\comp\circ\rightt & \cod\circ\comp=\cod\circ\pi_2\\
& =\cod\circ\pi_2\circ\compr & \text{Def of $\compr$}\\
& =\cod\circ\comp\circ\compr & \cod\circ\comp=\cod\circ\pi_2
\end{align*}
and the second is similar (CHECK). The third, proceeding as in the proof of Proposition~\ref{prop:comp-i}, follows once we establish that the upper rectangle above commutes, that is, 
\begin{align*}
 j_2\circ \compl & =(\compl\times\compl)\circ j_3\\
 j_2\circ \compr & =(\compr\times\compr)\circ j_3
\end{align*}
But these two equations are instances again of Lemma~\ref{lem:hh'} (CHECK). We have shown that composition is associative.
\end{proof}
\end{longVersion}

This finishes the verification that $\mathbb C\otimes \mathbb D$ is an internal category. We next show that $1$ carries the structure of an internal monoidal category.

\begin{proposition}\label{prop:1}
Let $(\mathcal V,1,\otimes)$ be a monoidal category where the unit is the terminal object. $1$ carries the structure of an internal monoidal category $\mathbb 1$ which is the neutral element wrt to the internal tensor $\otimes: Cat(\mathcal V)\times Cat(\mathcal V)\to Cat(\mathcal V)$ of Definition~\ref{def:internal-tensor}.
\end{proposition}

\begin{longVersion}
\begin{proof}
...
\end{proof}
\end{longVersion}

The next step is to show that the $\otimes: Cat(\mathcal V)\times Cat(\mathcal V)\to Cat(\mathcal V)$ of Definition~\ref{def:internal-tensor} can be extended to a functor.

\begin{proposition}\label{prop:otimes-functor}
Let $(\mathcal V,1,\otimes)$ be a monoidal category with finite limits where the unit is the terminal object. The internal tensor $\otimes: Cat(\mathcal V)\times Cat(\mathcal V)\to Cat(\mathcal V)$ of Definition~\ref{def:internal-tensor} is functorial.
\end{proposition}

\begin{longVersion}
\begin{proof}
Let $f=(f_0,f_1):\mathbb C\to\mathbb C'$ and $g=(g_0,g_1):\mathbb D\to\mathbb D'$ be morphisms of internal categories.  $f\otimes g$ is defined as $(f\otimes g)_0=f_0\otimes g_0$ and $(f\otimes g)_1:(\mathbb C\otimes\mathbb D)_1\to(\mathbb C'\otimes\mathbb D')_1$ as the unique arrow into the limit $(\mathbb C'\otimes\mathbb D')_1$
\begin{equation}\label{equ:j1-arrows}
\vcenter{
\xymatrix@C=9ex{
(\mathbb C\otimes\mathbb D)_1 \ar[rr]^{(f\otimes g)_1}
\ar@<-1ex>[d]_{\cod}
\ar@<1ex>[d]^{\dom}
&& 
(\mathbb C'\otimes\mathbb D')_1 \ar[rr]^{j_1}
\ar@<-1ex>[d]_{\cod}
\ar@<1ex>[d]^{\dom}
&& 
\mathbb C'_1\times\mathbb D'_1 
\ar@<-1ex>[d]_{\cod\times \cod}
\ar@<1ex>[d]^{\dom\times \dom}
\\
(\mathbb C\otimes\mathbb D)_0 \ar[rr]_{(f\otimes g)_0}
&& 
(\mathbb C'\otimes\mathbb D')_0 \ar[rr]_j
&& 
\mathbb C'_0\times\mathbb D'_0 
}}
\end{equation}
CHECK THAT $f\otimes g$ satisfies the equations of an internal functor.
\end{proof}
\end{longVersion}

The main result of the section is 

\begin{theorem}
Let $(\mathcal V,1,\otimes)$ be a (symmetric) monoidal category with finite limits where the unit is the terminal object and $\otimes: Cat(\mathcal V)\times Cat(\mathcal V)\to Cat(\mathcal V)$ the internal tensor of Definition~\ref{def:internal-tensor}. Then $(Cat(\mathcal V),\mathbb 1, \otimes)$ is a (symmetric) monoidal category. 
\end{theorem}

\begin{longVersion}
\begin{proof}
After Propositions~\ref{prop:1} and \ref{prop:otimes-functor}, it remains to show that the associator etc satisfy the laws of monoidal categories ... this should follow from  $(\mathcal V,1,\otimes)$  being a monoidal category (CHECK ... there will be some big diagrams to draw).
\end{proof}
\end{longVersion}

Finally,  \emph{internal strict monoidal categories} organise themselves in a (2-)category.

\begin{definition}\label{def:internal-monoidal}
We denote by $\Mon(\Cat(\V),\mathbb 1,\otimes))$, or briefly, $\Mon(\Cat(\V))$, the category of monoids in $(\Cat(\V),\mathbb 1,\otimes))$. \begin{gray}, which we also call the category of internal strict monoidal categories.\end{gray}
\end{definition}

\begin{theorem}\label{thm:moncatv}
$\Mon(\Cat(\V),\mathbb 1,\otimes))$ is a 2-category.
\end{theorem}

\begin{longVersion}
\begin{proof}
...
\end{proof}
\end{longVersion}

\begin{gray}
THE BELOW NEEDS TO BE REWORKED ...

\begin{proposition}
Examples~\ref{exle:sets} and \ref{exle:lists} are equivalent in the 2-category $\Mon(\Cat(\V),\mathbb 1,\otimes)$.
\end{proposition}

\begin{proof}
...
\end{proof}

\begin{remark}
We summarise the terminology. 
\begin{itemize}
\item $(\mathcal V,1,\otimes)$ is a monoidal category where the unit is the terminal object. 
\item $\Cat(\V)$ is the category of small categories internal in $\V$.
\item $(\Cat(\V),\mathbb 1,\otimes)$ is the monoidal category of internal categories.
\item $\Mon(\Cat(\V),\mathbb 1,\otimes))$, or briefly, $\Mon(\Cat(\V))$ is the category of monoids in $(\Cat(\V),\mathbb 1,\otimes))$, or, as we say, the category of internal monoidal categories, where internal refers to the fact that the monoidal operation $\mathbb C\otimes \mathbb C\to\mathbb C$ is typed not using the Cartesian product but the internal tensor $\otimes$.
\end{itemize}
\end{remark}

\begin{example} \ 
\begin{itemize}
\item $(\mathcal V,1,\otimes)$ is $(\Nom,1,\ast)$
\item $\Cat(\V)$ is the category of small categories internal in $\Nom$. In particular, objects and arrows now have support.
\item $(\Cat(\Nom),\mathbb 1,\ast)$ is the monoidal category of categories internal in $(\Nom,1,\ast)$. In particular, the separated product $\mathbb C\ast\mathbb D$ of internal categories has as objects pairs $(c,d)$ with disjoint support and as arrows pairs $(f,g)$ with disjoint supports of domains and disjoint supports of codomains.
\item $\Mon(\Cat(\V),\mathbb 1,\otimes))$, or briefly, $\Mon(\Cat(\V))$ is the category of monoids in $(\Cat(\V),\mathbb 1,\otimes))$, or, as we say, the category of internal strict monoidal categories, where internal refers to the fact that the monoidal operation $\mathbb C\otimes \mathbb C\to\mathbb C$ is typed not using the Cartesian product but the internal tensor $\otimes$. 
\end{itemize}
\end{example}

\medskip still some things to show:

\medskip

... $\Mon(\Cat(\V),\mathbb 1,\otimes))$ is a 2-category

... the forgetful functor $(Cat(\mathcal V),1,\otimes)\to(\mathcal V,1,\otimes)$ is strong monoidal, 

... the left adjoint $(\mathcal V,1,\otimes)\to(Cat(\mathcal V),1,\otimes)$ ? 

\begin{itemize}

\item A monoidal category internal in $(\Nom,1,\sepp)$ is a monoid internal in $(Cat(\mathcal (\Nom,1,\sepp),1,\sepp))$.

\item A nominal PRO is a monoidal category internal in $(\Nom,1,\sepp)$ where the objects are the finite subsets of names and disjoint union is the monoidal operation. A morphism of nominal PROs is a morphism in $(Cat(\mathcal (\Nom,1,\sepp),1,\sepp))$. Thus there is a category nPRO. 

\item Given an object $A$ in an internal monoidal category $(\mathbb C,1,\otimes)$ ... 

\item An algebra for a nominal PRO is a morphism of nominal PROs as follows.    ...

\item I can use nominal string diagrams to do ordinary string diagrams ... what would be a theorem formalising this?

\item Is there a systematic way of turning results on string diagrams into results on nominal string diagrams and/or vice versa?

\item 

\end{itemize}
\end{gray}

\section{Examples}\label{sec:examples}

% We are studying some variations of Example~\ref{exle:nF}. \begin{ak}This should be revised in conjunction with Section 2 as we present here the $\nPROP$s corresponding to the $\PROP$s of Section 2.\end{ak}
Before we give a formal definition of nominal $\PROP$s and nominal monoidal theories (NMTs) in the next section, we present as examples those NMTs that correspond to the SMTs of Fig~\ref{fig:smt-theories}. The nominal monoidal theories of Fig~\ref{fig:nmt-theories} should be immediately recognizable, indeed the significant differences are that wires now carry labels and there is a new generator \includegraphics[page=43, width=20mm]{twists_new} which allows us to change the label of a wire.
\begin{figure}
% \begin{tcolorbox}
\includegraphics[page=37, width=\linewidth]{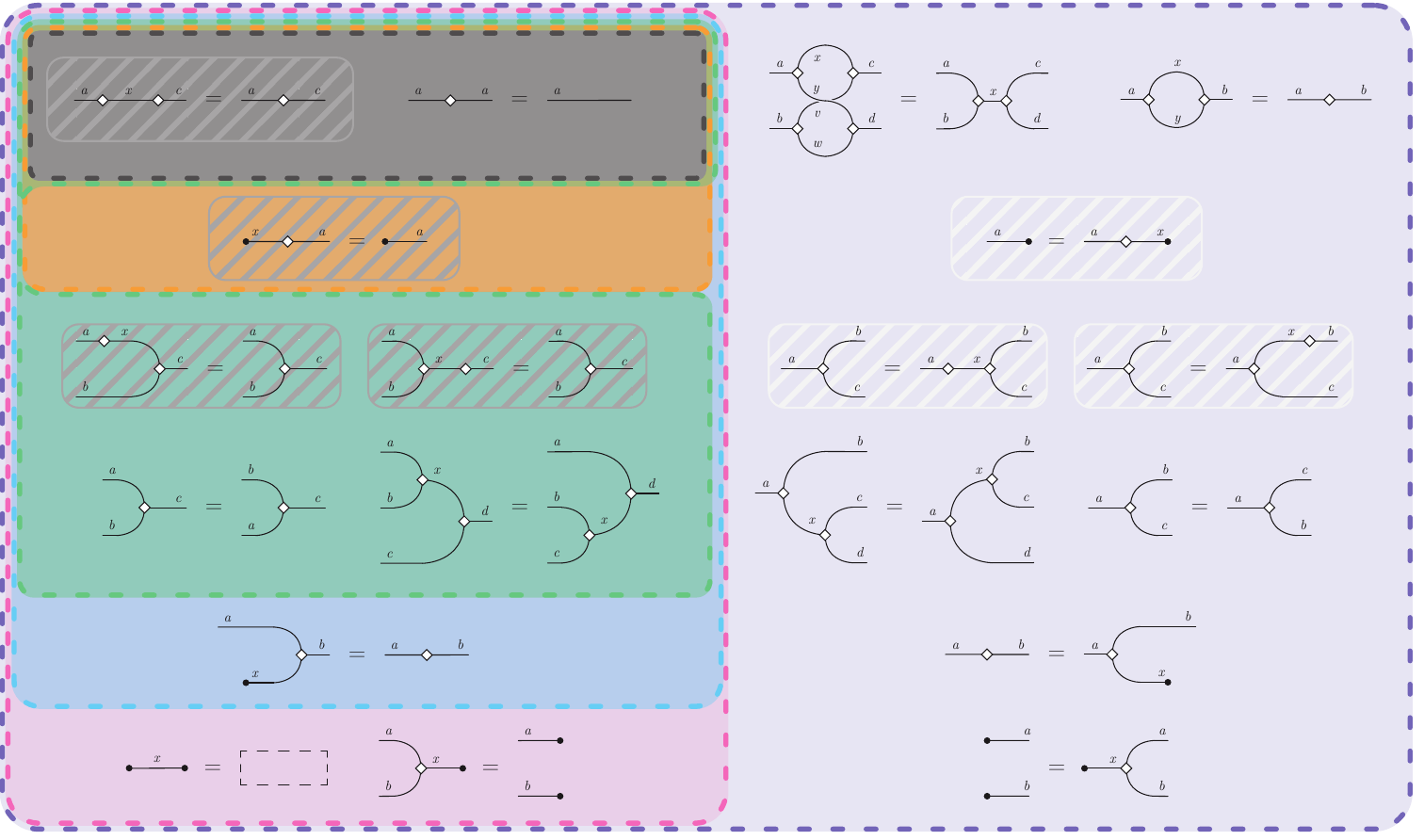}
\colorbox{black!30}{\strut bijections $\mathsf n\mathbb B$}, \colorbox{orange!60}{\strut injections $\mathsf n\mathbb I$}, \colorbox{green!40}{\strut surjections $\mathsf n\mathbb S$}, \colorbox{cyan!40}{\strut functions $\mathsf n\mathbb F$}, \colorbox{magenta!40}{\strut partial functions $\mathsf n\mathbb P$} and \colorbox{cyan!60!magenta!60}{\strut relations $\mathsf n\mathbb R$}
\caption{Nominal monoidal theories}\label{fig:nmt-theories}
% \end{tcolorbox}
\end{figure}

\begin{gray}
\medskip
The first example lists presentations of nominal monoidal theories for the nominal monoidal categories of finite sets and functions, injections, surjections, partial functions and relations, respectively.

\begin{example}
The category of finite sets and
\begin{itemize} 
\item bijections is given by the empty signature and equations.
%	$\Sigma_b=\{\delta_{ab}:\{a\} \to \{b\} \mid a,b \in \mathcal N\}$, with equations:
%  \begin{align*}
%    \delta_{aa} &= id_a &
%    \delta_{ab} \circ \delta_{bc} &= \delta_{ac} \\
%    (a\ b)\delta_{xy} &= \delta_{(a\ b)\cdot x\ (a\ b)\cdot y}
%  \end{align*}
%  \begin{figure}
%    \centering
%      \includegraphics[page=117, width=0.4\textwidth]{twists}\quad
%      \includegraphics[page=1, width=0.4\textwidth]{twists}
%    \caption{Equations of $\Sigma_b$}
%    \label{fig:Sigmab}
%  \end{figure}

\item injections is given by
  $\Sigma_i=\{\eta_a:\varnothing \to \{a\} \mid a \in \mathcal N\}$ and $E_i=\emptyset$.  The equations 
  \begin{center}
  \includegraphics[page=5, width=0.3\textwidth]{twists_new}
  \end{center}
  \vspace{-3ex}
  % \begin{figure}
  %   \centering
  %   \includegraphics[page=117, width=0.4\textwidth]{twists}
  %   \includegraphics[page=1, width=0.4\textwidth]{twists}
  %   \caption{Equations of $\Sigma_i$}
  % \end{figure}
follow from those of Fig~\ref{fig:nominal-set}.
\item surjections is given by
	$\Sigma_s=\{\mu_{abc}:\{a,b\} \to \{c\} \mid a,b,c \in \mathcal N\}$ and equations $E_s$ are 
	$(\mu_{abx} \uplus id_c) \circ \mu_{cdx} = (\mu_{bcx} \uplus id_a) \circ \mu_{adx}$ and
	\begin{sam}$\mu_{abx} = \mu_{bax}$ there is some seeming ambiguity with this rule as both $\mu_{abx}, \mu_{bax}$ have the same type $\{a,b\} \to \{x\}$, however, this does not mean we can automatically assume they are the same...this was the discussion about having the cup be a `-`, in which case, it would not be a commutative generator\end{sam}\begin{ak}should we say sth here? that two different operations have the same type is a very common situation ...  just delete the blue comment?\end{ak}
%  \begin{align*}
%    (\mu_{abx} \uplus id_c) \circ \mu_{cdx} &= (\mu_{bcx} \uplus id_a) \circ \mu_{adx} &
%    (\mu_{abx} \uplus id_c) \circ \mu_{cdx} &= (\mu_{aby} \uplus id_c) \circ \mu_{cdy} \\
%    (\delta_{ax} \uplus id_b) \circ \mu_{bcx} &= \mu_{abc} &
%    \mu_{abx} \circ \delta_{xc} &= \mu_{abc} \\
%    (a\ b)\mu_{xyz} &= \mu_{(a\ b)\cdot x\ (a\ b)\cdot y\ (a\ b)\cdot z}
%  \end{align*}
%  \begin{figure}[h]
\begin{center}      
	\includegraphics[page=12, width=0.3\textwidth]{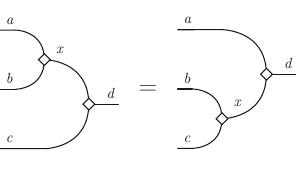}
\end{center}
\vspace{-3ex}
%\qquad
%     \includegraphics[page=15, width=0.4\textwidth]{twists_new}
%
%      \includegraphics[page=2, width=0.4\textwidth]{twists_new}\qquad
 %     \includegraphics[page=3, width=0.4\textwidth]{twists_new}
%    \caption{Equations of $\Sigma_s$}
%   \label{fig:Sigmas}
%  \end{figure}

\item functions has 
  $\Sigma_f=\Sigma_i \cup \Sigma_s$ and equations $E_f$ are $E_i\cup E_s$ plus
  $(id_a \uplus \eta_x) \circ \mu_{abx} = \delta_{ab}$
  \begin{center}
    \includegraphics[page=4, width=0.3\textwidth]{twists_new}
  \end{center}
%
  % \begin{figure}
  %   \centering
  %   \includegraphics[page=8, width=0.4\textwidth]{twists}
  %   \includegraphics[page=118, width=0.4\textwidth]{twists}
  %   \includegraphics[page=2, width=0.4\textwidth]{twists}
  %   \includegraphics[page=6, width=0.4\textwidth]{twists}
  %   \caption{Equations of $\Sigma_f$}
  % \end{figure}
\item partial functions has $\Sigma_{pf}=\Sigma_f \cup \{\hat\eta_a:\{a\} \to \varnothing \mid a \in \mathcal N\}$ and equations $E_{pf}$ are $E_f$ plus
 $\eta_x \circ \hat\eta_x = \varepsilon$ and $\mu_{abx} \circ \hat\eta_x = \hat\eta_a \uplus \hat\eta_b$ 
%
%  \begin{align*}
%    \delta_{ax} \circ \hat\eta_x &= \hat\eta_a &
%    \eta_x \circ \hat\eta_x &= \varepsilon\\
%    \mu_{abx} \circ \hat\eta_x &= \hat\eta_a \uplus \hat\eta_b
%    % \eta_a \circ \hat\eta_a &= \varepsilon
%  \end{align*}
%
%\begin{figure}
%\centering
%      \includegraphics[page=10, width=0.4\textwidth]{twists_new}
\begin{center}
      \includegraphics[page=6, width=0.3\textwidth]{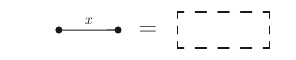} \qquad
      \includegraphics[page=16, width=0.3\textwidth]{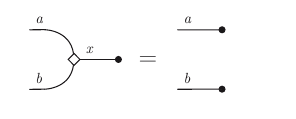}
\end{center}
%    \caption{Equations of $\Sigma_{pf}$}
%    \label{fig:Sigmapf}
%  \end{figure}
\item relations has 
 $\Sigma_{r}=\Sigma_{pf} \cup \{\hat\mu_{abc}:\{a\} \to \{b,c\} \mid a,b,c \in \mathcal N\}$, and equations $E_r$ are $E_{pf}$  plus 
%$\hat\mu_{adx} \circ (\hat\mu_{xbc} \uplus id_d) 
%= \hat\mu_{abx} \circ (\hat\mu_{xcd} \uplus id_b)$
%and
%$ (\hat\mu_{axy} \uplus \hat\mu_{bvw}) \circ (\mu_{xcv} \uplus \mu_{ydw})
%= \mu_{abx} \circ \hat\mu_{xcd}$
%and 
%$\mu_{abx} \circ \hat\mu_{xcd} = \mu_{aby} \circ \hat\mu_{ycd}$
%and
%$\hat\mu_{abx} \circ (id_b \uplus \hat\eta_{x}) = \delta_{ab}$
%and
%$\hat\mu_{axy} \circ \mu_{xyb} = \delta_{ab}$
%%  \begin{align*}
%%    \hat\mu_{adx} \circ (\hat\mu_{xbc} \uplus id_d) &= \hat\mu_{abx} \circ (\hat\mu_{xcd} \uplus id_b) &
%%    \hat\mu_{adx} \circ (\hat\mu_{xbc} \uplus id_d) &= \hat\mu_{ady} \circ (\hat\mu_{ybc} \uplus id_d) \\
%%    \hat\mu_{acx} \circ (\delta_{xb} \uplus id_c) &= \hat\mu_{abc} &
%%    \delta_{ax} \circ \hat\mu_{xbc} &= \hat\mu_{abc} \\
%%    (\hat\mu_{axy} \uplus \hat\mu_{bvw}) \circ (\mu_{xcv} \uplus \mu_{ydw}) &= \mu_{abx} \circ \hat\mu_{xcd} &
%%    \mu_{abx} \circ \hat\mu_{xcd} &= \mu_{aby} \circ \hat\mu_{ycd} \\
%%    \hat\mu_{abx} \circ (id_b \uplus \hat\eta_{x}) &= \delta_{ab} &
%%    \eta_x \circ \hat\mu_{xab} &= \eta_a \uplus \eta_b \\
%%    \hat\mu_{axy} \circ \mu_{xyb} &= \delta_{ab} &
%%    (a\ b)\hat\mu_{xyz} &= \hat\mu_{(a\ b)\cdot x\ (a\ b)\cdot y\ (a\ b)\cdot z} \\
%%  \end{align*}

% \begin{figure}
\begin{center}
      \includegraphics[page=13, width=0.3\textwidth]{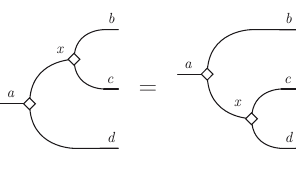}%\qquad
\qquad\qquad      
\includegraphics[page=19, width=0.3\textwidth]{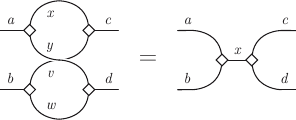}

      \includegraphics[page=9, width=0.3\textwidth]{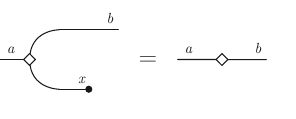}\qquad
      \includegraphics[page=18, width=0.3\textwidth]{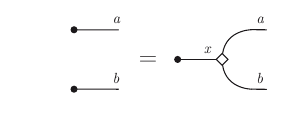}

      \includegraphics[page=11, width=0.3\textwidth]{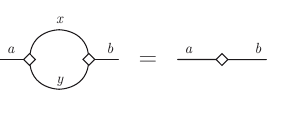}
\end{center}
%    \caption{Equations of $\Sigma_{r}$}
%    \label{fig:Sigmar}
%  \end{figure}
\end{itemize} 
\end{example}
\end{gray}

\begin{theorem}
The calculi of Fig \ref{fig:nmt-theories} are complete.
\end{theorem}

The proof of the theorem shows that the categories presented by Fig  \ref{fig:nmt-theories} are isomorphic to the categories of finite sets with the respective maps. These proofs seem easier  for NMTs than the corresponding proofs for SMTs (see eg Lafont~\cite{lafont})  because NMTs have no wire crossings. For example, in the case of bijections, it is immediate that every nominal diagram rewrites to a normal form, which is a parallel composition of diagrams of the form \includegraphics[page=43, width=20mm]{twists_new}. Completeness then follows, as usual, from the possibility to rewrite every diagram into normal form. The other cases are only slightly more complicated.

\begin{gray}
\begin{example}
The nominal monoidal theory that represents cospans of $\Fin$. ... Some remarks on the completeness proof ... 
\end{example}
\end{gray}

\section{Nominal monoidal theories and nominal PROPs}\label{sec:nmts}

\begin{gray}
change the structure? One section for each

nominal props

equivalence of nprops and props

nominal monoidal theories +  symmetric monoidal theories induce nominal monoidals

diagrammatic alpha equivalence

examples

transfer of completenes
\end{gray}

In this section, we introduce nominal $\PROP$s as  internal monoidal categories in nominal sets. We first  spell out the details of what that means in elementary terms and then discuss the notion of diagrammatic alpha-equivalence. 
%% In the next section Finally, we prove that the category $\nPROP$ of nominal props is equivalent to the category $\PROP$ of (ordinary) props.  
%\medskip
%From another point of view, we define the nominal analogues of symmetric monoidal theories and $\PROP$s, reviewed in Section~\ref{sec:scene}. More intuitively, we add labels to wires and compose by linking up wires carrying the same label. More technically, while a $\PROP$ is a monoid in the category of internal monoidal categories in $(\Set,1,\times)$ with natural numbers as objects and (at least) all bijections as arrows, a nominal $\PROP$ is a monoid in the category of internal monoidal categories in $(\Nom,1,\ast)$ with finite subsets of names as objects and (at least) all bijections as arrows. 
%%In the remainder of the section, we spell out in detail what that means and also define a coresponding notion of nominal monoidal theory, in analogy with the symmteric monoidal theories reviewed in Section~\ref{sec:scene}.

\begin{gray}
REWORK THE FOLLOWING
\medskip
The paradigmatic example of "nominal PROP" is Example~\ref{exle:sets}.

\medskip
We should also come back to this example ... (introductory remark to the next example) ... but will this be easier after the transfer theorem?

\begin{example}
Cospan(F) ...
\end{example}
\end{gray}

%\medskip
%We are interested in diagrammatic calculi for such and other categories, which we call nominal monoidal categories. A diagrammatic calculus will correspond to the presentation of a nominal monoidal category and soundness and completeness of the calculus will correspond to a representation result.

\subsection{Nominal monoidal theories} A \emph{nominal monoidal theory} $(\Sigma,E)$ is given by a nominal set $\Sigma$ of generators and a nominal set $E$ of equations.  A generator $\gen:A\to B$ has finite sets $A,B$  of names as types and $\Sigma$ is closed under permutations  $\pi\cdot\gen:\pi\cdot A\to\pi\cdot B$. 
The set of terms is given by closing under the operations of Fig~\ref{fig:terms}, which should be compared with Fig~\ref{fig:smt-terms}.
%unary operations $(a\ b)$ for all names $a\not=b$ (transpositions, where $(a\ b) = (b\ a)$), a binary operation called sequential composition written as $\circ$ or $;$ and a partial binary operation $\uplus$ called parallel composition or separated union.according to
%
\begin{figure}[h]
%\begin{tcolorbox}
\begin{gather*}
\frac{}{ \gen: A \to B \in \Sigma}
\qquad\qquad
\frac{}{ id_a:\{a\}\to\{a\}}
\qquad\qquad
\frac{}{\delta_{ab}:\{a\}\to\{b\}}
\\[2ex]
\frac{ t:A\to B\quad\quad t':A'\to B'}{ t\uplus t' :  A\uplus A'\to B\uplus B'}\qquad\qquad
\frac{ t:A\to B\quad\quad s:B\to C}{ t\hcomp s :  A \to C}
\\[2ex]
\frac{ t:A\to B}{(a\ b)\,t :(a\ b)\cdot A\to (a\ b)\cdot B}
%\frac{ t:A\to B\uplus\{b\}}{\vdash t\then (\id_B\uplus \delta_{bc}) :  A \to B\uplus\{c\}}
%\quad\quad\quad\quad
%\frac{ t:\{a\}\uplus A\to B}{\vdash (\delta_{ca}\uplus\id_A)\then t  :  \{c\}\uplus A \to B}
\end{gather*}
\caption{NMT Terms}\label{fig:terms}
%\end{tcolorbox}
\end{figure}%
%where $\uplus$ is the union of disjoint sets and  $id_A$ abbreviates $\biguplus_{a \in A} id_a$. 

\medskip\noindent
%Equations are sets of pairs of terms closed under the permutation action.
Every NMT freely generates a monoidal category internal in nominal sets by quotienting the generated terms by the equations $E$ as well as by equations describing that terms form a monoidal category and  a nominal set.  The equations of an internal monoidal category are given in Fig~\ref{fig:monoidal-category}. The main difference with the equations in Fig~\ref{fig:symmetric-monoidal-category} is that the interchange law for $\uplus$ is required to hold only if both sides are defined and that the two laws involving symmetries are replaced by the commutativity of $\uplus$.

%The equations we will need can be divided into those describing that terms form a monoidal category and those that describe that terms form a nominal set. For an NMT to present a monoidal category, we require the equations of Fig~\ref{fig:monoidal-category}
%
\begin{figure}[h]
% \begin{tcolorbox}
% \begin{gather*}
% \id_A \hcomp t = t = t \hcomp \id_B   \qquad\qquad\qquad\qquad  id_\varnothing \uplus t = t = t \uplus id_\varnothing
% \\[1ex]
% (t\hcomp s)\hcomp r = t \hcomp (s \hcomp r)\qquad\qquad
% t\uplus s = s \uplus t\qquad\qquad
% (t\uplus s)\uplus r = t \uplus (s \uplus r)
% \\[1ex]
% (s \hcomp t) \uplus (u \hcomp v) = (s \uplus u) \hcomp (t \uplus v)
% \end{gather*}
\[
\begin{array}{cc}
\id_A \hcomp t = t = t \hcomp \id_B \qquad{}& \qquad
id_\varnothing \uplus t = t = t \uplus id_\varnothing \\[1ex]
(t\hcomp s)\hcomp r = t \hcomp (s \hcomp r) \qquad{}& \qquad
(t\uplus s)\uplus r = t \uplus (s \uplus r) \\[1ex]
t\uplus s = s \uplus t \qquad{}& \qquad
(s \hcomp t) \uplus (u \hcomp v) = (s \uplus u) \hcomp (t \uplus v)
\end{array}
\]
\caption{NMT Equations of internal monoidal categories}\label{fig:monoidal-category}
% \end{tcolorbox}
\end{figure}
%
%where 
%
For terms to form a nominal set, we need the usual equations between permutations (not listed here) to hold, as well as the equations of Fig~\ref{fig:nominal-set} that specify how permutations act on terms.
\begin{figure}[h]
\begin{gather*}
(a\ b)id_x = id_{(a\ b)\cdot x}
\qquad\qquad
(a\ b)\delta_{xy} = \delta_{(a\ b)\cdot x \ (a\ b)\cdot y}
\qquad\qquad
(a\ b)\gen = (a\ b)\cdot \gen
\\[1ex]
(a\ b)(x\uplus y) = (a\ b)x\uplus (a\ b)y
\qquad\qquad
(a\ b)(x\hcomp y) = (a\ b)x\hcomp (a\ b)y
\qquad\qquad
\delta_{ab}\hcomp\delta_{bc} = \delta_{ac}
%\frac{\gen:A\to B\uplus\{b\}\qquad b,x\notin A}{\gen\,\then\,(\id_{B} \uplus \delta_{bx})  = (b\ x) \gen}
%    \qquad\qquad
%    \frac{\gen:\{a\}\uplus A\to B\qquad a,x\notin B}{(\delta_{xa} \uplus \id_A)\,\then\, \gen  = (x\ a) \gen}
%\\[1ex]
%\frac{\gen:A\uplus\{c\}\to B\uplus\{c\}}{\gen\,\then\,(\id_{B} \uplus \delta_{cx})  = (\id_A\uplus\delta_{cx})\hcomp(c\ x) \gen}
\end{gather*}
\begin{tabular}{ c c c }
\includegraphics[page=33, width=0.3\textwidth]{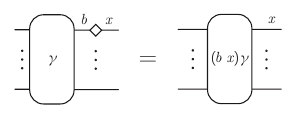} &
\includegraphics[page=34, width=0.3\textwidth]{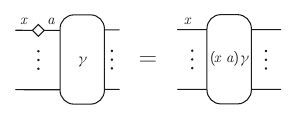} &
\includegraphics[page=35, width=0.3\textwidth]{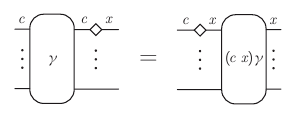}\\
$\frac{\gen:A\to B\uplus\{b\}\qquad b,x\notin A}{\gen\,\then\,(\id_{B} \uplus \delta_{bx})  = (b\ x) \gen}$ &
$\frac{\gen:\{a\}\uplus A\to B\qquad a,x\notin B}{(\delta_{xa} \uplus \id_A)\,\then\, \gen  = (x\ a) \gen}$ &
$\frac{\gen:A\uplus\{c\}\to B\uplus\{c\}}{\gen\,\then\,(\id_{B} \uplus \delta_{cx})  = (\id_A\uplus\delta_{cx})\hcomp(c\ x) \gen}$
\end{tabular}
\caption{NMT Equations of nominal sets
%Equations of permutations (not listed) and permutation action%
}\label{fig:nominal-set}
% \end{tcolorbox}
\end{figure}
These are routine, with the exception of the last three, specifying the interaction of renamings $\delta$ with renamings and generators $\gen\in\Sigma$, which we also depict  in diagrammatic form. Instances of these rules can be seen in Fig~\ref{fig:nmt-theories}, where they are distinguished by a \stripbox{\strut striped} background.

%
%\noindent
%\begin{tabular}{ c c c }
%\includegraphics[page=33, width=0.3\textwidth]{twists_new} &
%\includegraphics[page=34, width=0.3\textwidth]{twists_new} &
%\includegraphics[page=35, width=0.3\textwidth]{twists_new}\\
%$\frac{\gen:A\to B\uplus\{b\}\qquad b,x\notin A}{\gen\,\then\,(\id_{B} \uplus \delta_{bx})  = (b\ x) \gen}$ &
%$\frac{\gen:\{a\}\uplus A\to B\qquad a,x\notin B}{(\delta_{xa} \uplus \id_A)\,\then\, \gen  = (x\ a) \gen}$ &
%$\frac{\gen:A\uplus\{c\}\to B\uplus\{c\}}{\gen\,\then\,(\id_{B} \uplus \delta_{cx})  = (\id_A\uplus\delta_{cx})\hcomp(c\ x) \gen}$
%\end{tabular}
%

\subsection{Diagrammatic alpha-equivalence}

The equations of Fig~\ref{fig:nominal-set} introduce a notion of \emph{diagrammatic alpha-equivalence}, which allows us to rename `internal' names and to contract renamings.
 
\begin{definition}
Two terms of a nominal monoidal theory are alpha-equivalent if their equality follows from the equations in Fig~\ref{fig:nominal-set}.
\end{definition}

\medskip\noindent\textbf{Notation:} Every permutation $\pi$ of names gives rise to  bijective functions $\pi_A:A\to\pi[A]=\{\pi(a)\mid a\in A\}=\pi\cdot A$. Any such $\pi_A$, as well as the inverse $\pi_A^{-1}$,  are parallel compositions of $\delta_{ab}$ for suitable $a,b\in\names$. In fact, we have $\pi_A=\biguplus_{a\in A} \delta_{a\pi(a)}$.
We may therefore use the $\pi_A$ as abbreviations in  terms.

\begin{proposition} Let $t:A\to B$ be a term of a nominal monoidal theory. The equations in  Fig~\ref{fig:nominal-set} entail that $\pi\cdot t = (\pi_A)^{-1}\hcomp t\hcomp \pi_B$.
\[\xymatrix@C=50pt{
A\ar[r]^t \ar[d]_{\pi_A}& B\ar[d]^{\pi_B}\\
\pi[A]\ar[r]_{\pi\cdot t} & \pi[B]
}\]

\end{proposition}

\begin{longVersion}
\begin{proof}
\begin{ak}Wlog (CHECK) we can assume that $\pi=(x\ y)$.  There are 16 cases: $\pi_A$ can be $\id$, $\delta_{xy}$, $\delta_{yx}$ or $\delta_{xy}+\delta_{yx}$. ...  \end{ak}
\end{proof}
\end{longVersion}

\begin{corollary}
Let $t:A\uplus\{c\}\to B\uplus\{c\}$ be a term of a nominal monoidal theory and $d\fresh t$. Then $t= (\delta_{cd}\uplus\id_A)\hcomp (c\ d)\cdot t\hcomp (\delta_{dc}\uplus\id_B)$. \begin{gray}MAKE A DIAGRAM\end{gray}
\end{corollary}

\begin{longVersion}
\begin{proof}
\begin{ak} ... details missing ... \end{ak}
\end{proof}
\end{longVersion}

\begin{corollary}
Let $t:A\to B$ be a term of a nominal monoidal theory. Modulo the equations of Fig~\ref{fig:nominal-set}, the support of $t$ is $A\cup B$.
\end{corollary}

\begin{longVersion}
\begin{proof}
It follows from the proposition that $\supp(t)\subseteq A\cup B$. For the converse, suppose that there is $x\in A\cup B$ and a support $S$ of $t$ with $x\notin S\subseteq A\cup B$. Choose a permutation $\pi$ that fixes $S$ and maps $x$ to some $\pi(x)\notin A\cup B$. Then either $\pi\cdot A\not=A$ or $\pi\cdot B\not=B$, hence $\pi\cdot t \not=t$, contradicting that $S$ is a support of $t$.
\end{proof}
\end{longVersion}

The last corollary shows that internal names are bound by sequential composition. Indeed,
in a composition $A\stackrel{t}{\to} C\stackrel{s}{\to}B$, the names in $C\setminus(A\cup B)$ do not appear in the support of $t\hcomp s$.

\subsection{Nominal PROPs}
From the point of view of Section~\ref{sec:internal-monoidal}, a nominal $\PROP$ is an internal strict monoidal category in $(\Nom,1,\ast)$ that has finite sets of names as objects and at least all bijections as arrows. We spell this out in detail.

%\begin{ak}MAYBE AS IN LACK, Composing Props, define first a nominal PRO, then the nominal PRO $\mathsf n\mathbb B$, then $\nPROP$ as the category of triangles ... this would allow us to simplify the def below ... 
%\end{ak}
%

\begin{remark}
A nominal $\PROP$ $\mathbb C$ is a small category, with a set $\mathbb C_0$ of `objects' and a set $\mathbb C_1$ of `arrows', defined as follows. We write $\then$ for the `sequential' composition (in the diagrammatic order) and $\uplus$ for the `parallel' or `monoidal' composition.
\begin{itemize}
\item $\mathbb C_0$  is the set of finite subsets of $\names$. The permutation action is given by $\pi\cdot A= \pi[A]=\{\pi(a) \mid a\in A\}$. 
%%%
\item $\mathbb C_1$ contains all bijections (`renamings') $\pi_A:A\to\pi\cdot A$ for all finite permutations $\pi:\names\to\names$ and is closed under the operation mapping an arrow $f:A\to B$ to $\pi\cdot f : \pi\cdot A\to \pi\cdot B$ defined as $\pi\cdot f = (\pi_A)^{-1}\then f\then \pi_B$. 
%%%
\item $A \uplus B$ is the union of $A$ and $B$ and defined whenever $A$ and $B$ are disjoint. This makes $(\mathbb C_0,\emptyset,\uplus)$ a commutative partial monoid. On arrows, we require $(\mathbb C_1,\emptyset,\uplus)$ to be a commutative partial monoid, with $f\uplus g$ defined whenever $\dom f\cap\dom g=\emptyset$ and $\cod f\cap\cod g=\emptyset$.
\end{itemize}
\end{remark}

From this definition on can deduce the following.

\begin{remark} 
\begin{itemize}
%%%
\begin{gray}
\item Each permutation $\pi$ gives rise to an equivariant and strict monoidal functor $\boldsymbol \pi:\mathbb C\to\mathbb C$ defined by $\boldsymbol\pi(A)=\pi\cdot A$ and $\boldsymbol\pi(f)=\pi\cdot f$ \begin{gray}(CHECK)\end{gray}. The $\pi_A:A\to{\boldsymbol\pi} (A)$ are natural transformations $\Id\to\pifun$ \begin{gray}(CHECK)\end{gray}. 
\end{gray}
%%%
\item A nominal prop has a nominal set of objects and a nominal set of arrows. 
\item The support of an object $A$ is $A$ and the support of an arrow $f:A\to B$ is $A\cup B$. In particular, $\supp(f\then g) = \dom(f)\cup\cod(g)$. In other words, nominal props have diagrammatic alpha equivalence.
%%%
\item There is a category $\nPROP$ that consists of nominal props together with functors that are the identity on objects and strict monoidal and equivariant.
\begin{gray}Permutations are acting on arrows (and objects), that is,  
\begin{gather*}
\id\cdot f = f \quad\quad (\pi'\circ\pi)\cdot f= \pi'\cdot(\pi\cdot f)
\end{gather*}
where $\id$ here refers to the identity permutation $\names\to\names$. 
\end{gray}
%%%
\item Every NMT presents a $\nPROP$. Conversely, every $\nPROP$ is presented by at least one NMT given by all terms as generators and all equations as equations. \begin{gray}(CHECK, POSSIBLY ADD DETAILS.)\end{gray}
%\item An NMT is a commutative strict monoidal category internal in $(\Nom,1,\#)$.
%\item \begin{ak}...\end{ak}
\end{itemize}
\end{remark}

\begin{gray}

\subsection{Ordered nominal PROPs}

THIS IS WORTH DEVELOPPING BUT VERY SKETCHY AT THE MOMENT AND MAY BE NOT NEEDED FOR THE PAPER

Ordered nominal PROPs (onPROPs) are defined like nominal PROPs, but the types of the generators are not sets of names but lists of names, the members of which need to be pairwise distinct. onPROPs provide another example of internal monoidal categories in $(\Nom,1,\ast)$.

More precisely, the set of objects of an onPROP $\mathbb C$ is obtained by closing the names $\names$ under $\uplus$ with $\varnothing$ as neutral element. For example $\mathbf a = a_0\uplus a_1=(a_0,a_1)\in\names^{\ast 2}$ is a two element list of names. There are also new bas\texttt{}ic terms,  symmetries,  which permute the order of the names. They are all generated from arr\texttt{}ows $\gamma_{ab}:a\uplus b\to b\uplus a$ by closing under $\uplus$ and $\hcomp$. For example, the symmetry $(a,b,c)\to(c,b,a)$ can be presented by $(\gamma_{ab}\uplus c)\hcomp(\id_b\uplus\gamma_{ac})\hcomp(\gamma_{bc}\hcomp\id_a)$. Such a presentation is not unique and we will therefore require the usual equations of Fig~\ref{fig:bijections}. The operation $\uplus:\mathbb C\ast \mathbb C\to\mathbb C$ is a functor (internal in categories in $\Nom$) and the $\gamma_{ab}:a\uplus b\to b\uplus a$ are natural transformations (internal in categories in $\Nom$).
\footnote{Internal here means that $\uplus:\mathbb C\ast \mathbb C\to\mathbb C$ is equivariant and that the $\gamma_{ab}:a\uplus b\to b\uplus a$ are closed under the permutation action.}
As in classical string diagrams, naturality provides us with equations that allow us to `slide' symmetries past diagrams in the usual way. 
\end{gray}

\section{Equivalence of nominal and ordinary string diagrams}\label{sec:equivalence}

We show that the categories $\nPROP$ and $\PROP$ are equivalent.

\medskip
To define translations between ordinary and nominal monoidal theories we introduce some auxiliary notation. We denote lists that contain each letter at most once by bold letters. If $\boldsymbol a = [a_1,\ldots a_n]$ is a list, then $\underline {\boldsymbol a}=\{a_1,\ldots a_n\}$. %and $\boldsymbol b = [b_1,\ldots b_m]$. 
Given lists $\boldsymbol a$ and $\boldsymbol a'$ with $\underline{\boldsymbol a}=\underline{\boldsymbol a'}$ we abbreviate bijections in $\PROP$ (also called symmetries) mapping $i\mapsto a_i=a'_j\mapsto j$  as $\langle \boldsymbol a | \boldsymbol a'\rangle$. Given lists $\boldsymbol a$ and $\boldsymbol b$ of the same length we write $[ \boldsymbol a | \boldsymbol b] = \biguplus \delta_{a_ib_i}$ for the bijection $a_i\mapsto b_i$ in an $\nPROP$.

%\medskip\noindent
%We first define a functor $\PROP\to\nPROP$.

\begin{proposition}\label{prop:NOM}
For any $\PROP$ $\mathcal S$, there is an $\nPROP$ \[\NOM(\mathcal S)\] that has for all arrows $f:\underline n\to \underline m$ of $\mathcal S$, and for all lists $\boldsymbol a = [a_1,\ldots a_n]$ and $\boldsymbol b = [b_1,\ldots b_m]$ arrows $[\boldsymbol a\rangle f\langle\boldsymbol b]$. These arrows are subject to equations 
%(we sometimes write $[a_1,\ldots a_n\rangle$ for  $[\boldsymbol a\rangle$ and $\langle b_1,\ldots b_m]$ for $\langle\boldsymbol b]$)
%
\begin{align*}
[\boldsymbol a\rangle f\hcomp g\langle\boldsymbol c]
&= 
[\boldsymbol a\rangle f\langle\boldsymbol b]\hcomp[\boldsymbol b\rangle g\langle\boldsymbol c]
%& \textrm{for any choice of $\boldsymbol b$}
\\
[\boldsymbol a+\boldsymbol c\rangle f \oplus g\langle\boldsymbol b+\boldsymbol d]
&= 
[\boldsymbol a\rangle f \langle\boldsymbol b]\uplus[\boldsymbol c\rangle g\langle\boldsymbol d]
% \\
% [a_1, a_2\rangle tw \langle b_1, b_2] 
% & = \delta_{a_1b_1}\uplus\delta_{a_2b_1}
\\
[\boldsymbol a\rangle \id \langle \boldsymbol b]
& = [\boldsymbol a| \boldsymbol b]
\\
[\boldsymbol a\rangle\, \langle\boldsymbol b | \boldsymbol b'\rangle \hcomp f\, \langle\boldsymbol c]\,
& = 
[\boldsymbol a | \boldsymbol b]\hcomp [ \boldsymbol b'\rangle f \langle\boldsymbol c]
\\
[\boldsymbol a\rangle\, f \hcomp \langle\boldsymbol b | \boldsymbol b'\rangle\, \langle\boldsymbol c]\,
& = 
[ \boldsymbol a\rangle f \langle\boldsymbol b] \hcomp [\boldsymbol b' | \boldsymbol c]
\end{align*}
\end{proposition}

\begin{proof}
To show that $\NOM(\mathcal S)$ is well-defined, we need to check that the equations of $\mathcal S$ are respected. We only have space here for the most interesting case which is the naturality of symmetries given by the last equation in Fig~\ref{fig:symmetric-monoidal-category}.
%namely that of $(t \oplus id_z) \hcomp \tw_{n,z} = \tw_{m,z} \hcomp (id_z \oplus t)$. 
We write $\boldsymbol a^m$ for a list of $a$'s of length $m$.
\begin{align*}
[\boldsymbol a^m + \boldsymbol a^z\rangle\,(t \oplus id_z) \hcomp \tw_{n,z}\,\langle\boldsymbol b^z + \boldsymbol b^n] 
&= ([\boldsymbol a^m\rangle\,t\,\langle\boldsymbol x^n] \uplus [\boldsymbol a^z\rangle\,id_z\,\langle\boldsymbol x^z]) \hcomp [\boldsymbol x^n + \boldsymbol x^z\rangle\,\tw_{n,z}\,\langle\boldsymbol b^z + \boldsymbol b^n] 
\\
&= ([\boldsymbol a^z\rangle\,id_z\,\langle\boldsymbol x^z] \uplus [\boldsymbol a^m\rangle\,t\,\langle\boldsymbol x^n]) \hcomp [\boldsymbol x^n + \boldsymbol x^z\rangle\,\tw_{n,z}\,\langle\boldsymbol b^z + \boldsymbol b^n] 
\\
&= [\boldsymbol a^z + \boldsymbol a^m\rangle\,id_z \oplus t\,\langle\boldsymbol x^z + \boldsymbol x^n] \hcomp [\boldsymbol x^n + \boldsymbol x^z\rangle\,\tw_{n,z}\,\langle\boldsymbol b^z + \boldsymbol b^n] 
\\
&= [\boldsymbol a^z + \boldsymbol a^m\rangle\,id_z \oplus t\,\langle\boldsymbol x^z + \boldsymbol x^n] \hcomp [\boldsymbol x^n + \boldsymbol x^z\rangle\,\langle\boldsymbol x^n+\boldsymbol x^z|\boldsymbol x^z+\boldsymbol x^n\rangle\,\langle\boldsymbol b^z + \boldsymbol b^n] 
\\
&= [\boldsymbol a^z + \boldsymbol a^m\rangle\,id_z \oplus t\,\langle\boldsymbol x^z + \boldsymbol x^n] \hcomp [\boldsymbol x^n + \boldsymbol x^z|\boldsymbol x^n+\boldsymbol x^z] \hcomp [\boldsymbol x^z+\boldsymbol x^n|\boldsymbol b^z + \boldsymbol b^n] 
\\
&= [\boldsymbol a^z + \boldsymbol a^m\rangle\,id_z \oplus t\,\langle\boldsymbol x^z + \boldsymbol x^n] \hcomp [\boldsymbol x^z+\boldsymbol x^n|\boldsymbol b^z + \boldsymbol b^n] 
\\
&= [\boldsymbol a^z + \boldsymbol a^m\rangle\,id_z \oplus t\,\langle\boldsymbol b^z + \boldsymbol b^n] 
\\
&= [\boldsymbol a^m + \boldsymbol a^z|\boldsymbol a^m + \boldsymbol a^z] \ \hcomp\ [\boldsymbol a^z + \boldsymbol a^m\rangle\,id_z \oplus t\,\langle\boldsymbol b^z + \boldsymbol b^n] 
\\
&= [\boldsymbol a^m + \boldsymbol a^z\rangle\,\langle\boldsymbol a^m + \boldsymbol a^z|\boldsymbol a^z + \boldsymbol a^m\rangle \hcomp (id_z \oplus t)\,\langle\boldsymbol b^z + \boldsymbol b^n] 
\\
&= [\boldsymbol a^m + \boldsymbol a^z\rangle\,\tw_{m,z} \hcomp (id_z \oplus t)\,\langle\boldsymbol b^z + \boldsymbol b^n] \ \hspace{9em} 
\end{align*}
Note how commutativity of $\uplus$ is used to show that naturality of symmetries is respected.\end{proof}

%\vspace{-1em} 

%\noindent Next we define a functor $\nPROP\to\PROP$.

\begin{proposition}\label{prop:ORD}
For any $\nPROP$ $\mathcal T$ there is a $\PROP$ \[\ORD(\mathcal T)\] that has for all arrows $f:A\to B$ of $\mathcal T$, and for all lists $\boldsymbol a = [a_1,\ldots a_n]$ and $\boldsymbol b = [b_1,\ldots b_m]$ arrows $\langle\boldsymbol a]f[\boldsymbol b\rangle$. These arrows are subject to equations
\begin{align*}
\langle\boldsymbol a]\,f\hcomp g\,[\boldsymbol c\rangle
&= 
\langle\boldsymbol a]\, f \,[\boldsymbol b\rangle\hcomp\langle\boldsymbol b]\, g\,[\boldsymbol c\rangle
\\
%\begin{sam}
\langle\boldsymbol a_f+\boldsymbol a_g]\,f\uplus g\,[\boldsymbol b_f+\boldsymbol b_g\rangle 
%\end{sam}
&= 
\langle\boldsymbol a_f]\,f\,[\boldsymbol b_f\rangle \oplus \langle\boldsymbol a_g]\,g\,[\boldsymbol b_g\rangle
\\
\langle\boldsymbol a]\,\id\,[\boldsymbol a\rangle
& = \id
% & 
% \\
% \langle a]\delta_{ab}[b\rangle 
% & = \id
\\
\langle\boldsymbol a] \, [\boldsymbol a' | \boldsymbol b]\hcomp f\,  [\boldsymbol c\rangle
& = 
\langle\boldsymbol a | \boldsymbol a'\rangle \hcomp \langle \boldsymbol b]\, f \,[\boldsymbol c\rangle 
\\
\langle\boldsymbol a]\, f \hcomp [\boldsymbol b | \boldsymbol c] \, [\boldsymbol c'\rangle
& = 
\langle \boldsymbol a]\, f \,[\boldsymbol b\rangle \hcomp \langle\boldsymbol c | \boldsymbol c'\rangle
\end{align*}

%where in the last two equations equation $\boldsymbol a = [a_1,\ldots a_n], \boldsymbol a' = [a'_1,\ldots a'_n]$ and $\{a_1,\ldots a_n\}=\{a'_1,\ldots a'_n\}$ and $\langle\boldsymbol a | \boldsymbol a'\rangle$ is a shorthand for the symmetry obtained from composing $i\mapsto a_i$ with $a'_j\mapsto j$ where $a_i = a'_j$ (analogously for $\boldsymbol c/\boldsymbol c'$) .
\end{proposition}

\begin{proof}
To show that $\ORD$ is well defined we need to show that the equations of an NMT are respected. The most interesting case here is the commutativity of $\uplus$ since the $\oplus$ of SMTs is not commutative. 
%docus on the interesting case, $t \uplus s = s \uplus t$, or rather $\langle\boldsymbol a]\,t \uplus s \,[\boldsymbol b\rangle = \langle\boldsymbol a]\,s \uplus t\,[\boldsymbol b\rangle$ for some choice of $\boldsymbol a,\boldsymbol b$. W.l.o.g. we consider the case below:
% For $ : m \to n$, we write $\boldsymbol a^m$ as a list of $a$'s of length $m$.
\begin{align*}
\langle\boldsymbol a_t + \boldsymbol a_s]\,t \uplus s\,[\boldsymbol b_t + \boldsymbol b_s\rangle
&= \langle\boldsymbol a_t]\,t\,[\boldsymbol b_t\rangle \oplus \langle\boldsymbol a_s]\,s\,[\boldsymbol b_s\rangle
\\
&= (\langle\boldsymbol a_t]\,t \,[\boldsymbol b_t\rangle \ \hcomp\ id_{|\boldsymbol b_t|}) \oplus (id_{|\boldsymbol a_s|} \ \hcomp\, \langle\boldsymbol a_s]\,s\,[\boldsymbol b_s\rangle)
\\
&= (\langle\boldsymbol a_t]\,t \,[\boldsymbol b_t\rangle \oplus id_{|\boldsymbol a_s|}) \ \hcomp\, (id_{|\boldsymbol b_t|} \oplus \langle\boldsymbol a_s]\,s\,[\boldsymbol b_s\rangle)
\\
&= (\langle\boldsymbol a_t]\,t \,[\boldsymbol b_t\rangle \oplus id_{|\boldsymbol a_s|}) \,\hcomp\, \tw_{|\boldsymbol b_t|,|\boldsymbol a_s|} \,\hcomp\, \tw_{|\boldsymbol a_s|,|\boldsymbol b_t|} \,\hcomp\, (id_{|\boldsymbol b_t|} \oplus \langle\boldsymbol a_s]\,s\,[\boldsymbol b_s\rangle)
\\
&= \tw_{|\boldsymbol a_t|,|\boldsymbol a_s|} \,\hcomp\, (id_{|\boldsymbol a_s|} \oplus \langle\boldsymbol a_t]\,t \,[\boldsymbol b_t\rangle) \,\hcomp\, \tw_{|\boldsymbol a_s|,|\boldsymbol b_t|} \,\hcomp\, (id_{|\boldsymbol b_t|} \oplus \langle\boldsymbol a_s]\,s\,[\boldsymbol b_s\rangle)
\\
&= \tw_{|\boldsymbol a_t|,|\boldsymbol a_s|} \,\hcomp\, (id_{|\boldsymbol a_s|} \oplus \langle\boldsymbol a_t]\,t \,[\boldsymbol b_t\rangle) \ \hcomp\ (\langle\boldsymbol a_s]\,s\,[\boldsymbol b_s\rangle \oplus id_{|\boldsymbol b_t|}) \,\hcomp\, \tw_{|\boldsymbol b_s|,|\boldsymbol b_t|}
\\
&= \tw_{|\boldsymbol a_t|,|\boldsymbol a_s|} \,\hcomp\, (id_{|\boldsymbol a_s|} \ \hcomp\, \langle\boldsymbol a_s]\,s\,[\boldsymbol b_s\rangle) \oplus (\langle\boldsymbol a_t]\,t \,[\boldsymbol b_t\rangle \ \,\hcomp\ id_{|\boldsymbol b_t|}) \,\hcomp\, \tw_{|\boldsymbol b_s|,|\boldsymbol b_t|}
\\
&= \tw_{|\boldsymbol a_t|,|\boldsymbol a_s|} \,\hcomp\, \langle\boldsymbol a_s + \boldsymbol a_t]\,s \uplus t\,[\boldsymbol b_s + \boldsymbol b_t\rangle \,\hcomp\, \tw_{|\boldsymbol b_s|,|\boldsymbol b_t|}
\\
&= \langle \boldsymbol a_t + \boldsymbol a_s|\boldsymbol a_s + \boldsymbol a_t\rangle \,\hcomp\, \langle\boldsymbol a_s + \boldsymbol a_t]\,s \uplus t\,[\boldsymbol b_s + \boldsymbol b_t\rangle \,\hcomp\, \langle \boldsymbol b_s + \boldsymbol b_t|\boldsymbol b_t + \boldsymbol b_s\rangle
\\
&= \langle \boldsymbol a_t + \boldsymbol a_s]\,[\boldsymbol a_s + \boldsymbol a_t|\boldsymbol a_s + \boldsymbol a_t] \ \hcomp\, s \uplus t \ \hcomp\ [\boldsymbol b_s + \boldsymbol b_t|\boldsymbol b_s + \boldsymbol b_t]\,[\boldsymbol b_t + \boldsymbol b_s\rangle
\\
&= \langle\boldsymbol a_t + \boldsymbol a_s]\,s \uplus t\,[\boldsymbol b_t + \boldsymbol b_s\rangle
\end{align*}
Note how naturality of symmetries is used to show that the definition of $\ORD$ respects commutativity of $\uplus$.
\end{proof}

\begin{remark}
The following equations can be obtained from the ones above:
\begin{align*}
[\boldsymbol a\rangle\, f \hcomp \langle\boldsymbol b|\boldsymbol b'\rangle \hcomp g \,\langle\boldsymbol c]
&= [\boldsymbol a\rangle\, f \,\langle\boldsymbol b] \hcomp [\boldsymbol b'\rangle\, g \,\langle\boldsymbol c]
\\
[\boldsymbol a\rangle\, \langle\boldsymbol b|\boldsymbol b'\rangle \,\langle\boldsymbol c]
&= [\boldsymbol a|\boldsymbol b] \hcomp [\boldsymbol b'|\boldsymbol c]
\\
\langle\boldsymbol a]\, f \hcomp [\boldsymbol b|\boldsymbol c] \hcomp g \,[\boldsymbol d\rangle
&= \langle\boldsymbol a]\, f \,[\boldsymbol b\rangle \hcomp \langle\boldsymbol c]\, g \,[\boldsymbol d\rangle
\\
\langle\boldsymbol a]\, [\boldsymbol a'|\boldsymbol b'] \,[\boldsymbol b\rangle
&= \langle\boldsymbol a|\boldsymbol a'\rangle \hcomp \langle\boldsymbol b'|\boldsymbol b\rangle
\\
\langle\boldsymbol a][\boldsymbol a\rangle\, f \,\langle\boldsymbol b][\boldsymbol b\rangle
&= \langle\boldsymbol c][\boldsymbol c\rangle\, f \,\langle\boldsymbol d][\boldsymbol d\rangle
\\
% \langle a]\id[a\rangle 
% & = \id
% \\
% \langle a,b]\delta_{ab}\uplus \delta_{ba}[a,b\rangle 
% & = tw
% \\
% [\boldsymbol a\rangle\langle\boldsymbol a | \boldsymbol b\rangle\langle\boldsymbol b]
% & = \id
% \\
% [\boldsymbol a\rangle\langle\boldsymbol a] f[\boldsymbol b\rangle\langle\boldsymbol b]
% & = f
% \\
% \langle\boldsymbol a][\boldsymbol a\rangle f\langle\boldsymbol b][\boldsymbol b\rangle
% & = f
% \\
% \langle\boldsymbol a]f\uplus g[\boldsymbol c\rangle
% &= 
% \langle\boldsymbol a] id[\boldsymbol a'\rangle \hcomp \langle\boldsymbol a'] f\uplus g [\boldsymbol b'\rangle\hcomp\langle\boldsymbol b'] id[\boldsymbol c\rangle
\end{align*}

\begin{longVersion}
\begin{proof}
\begin{align*}
[\boldsymbol a \rangle f \hcomp \langle\boldsymbol b | \boldsymbol b'\rangle \hcomp g\langle\boldsymbol c]
&= [\boldsymbol a \rangle\, f \,\hcomp \langle\boldsymbol b | \boldsymbol b'\rangle \,\langle\boldsymbol b' ]\, \hcomp\, [\boldsymbol b'\rangle\, g \,\langle\boldsymbol c]
\\
&= [\boldsymbol a \rangle\, f \,\langle\boldsymbol b ]\, \hcomp \,[\boldsymbol b' | \boldsymbol b']\, \hcomp\, [\boldsymbol b'\rangle\, g \,\langle\boldsymbol c]
\\
&= [\boldsymbol a \rangle\, f \,\langle\boldsymbol b]\, \hcomp \,[\boldsymbol b'\rangle\, g \,\langle\boldsymbol c]
\end{align*}

\begin{align*}
\langle\boldsymbol a]\, f \hcomp [\boldsymbol b|\boldsymbol c] \hcomp g \,[\boldsymbol d\rangle
&= \langle\boldsymbol a]\, f \,\hcomp\, [\boldsymbol b|\boldsymbol c] \,[\boldsymbol c\rangle\ \hcomp\ \langle\boldsymbol c]\, g \,[\boldsymbol d\rangle
\\
&= \langle\boldsymbol a]\, f \,[\boldsymbol b\rangle\, \hcomp \,\langle\boldsymbol c | \boldsymbol c\rangle\, \hcomp\, \langle\boldsymbol c]\, g \,[\boldsymbol d\rangle
\\
&= \langle\boldsymbol a]\, f \,[\boldsymbol b\rangle\, \hcomp \,\langle\boldsymbol c]\, g \,[\boldsymbol d\rangle
\end{align*}

\noindent
The choice of $\boldsymbol a, \boldsymbol b$ is arbitrary, because we can prove that for any other choice $\boldsymbol c, \boldsymbol d$, we have $\langle\boldsymbol a][\boldsymbol a\rangle\, f \,\langle\boldsymbol b][\boldsymbol b\rangle = \langle\boldsymbol c][\boldsymbol c\rangle\, f\,\langle\boldsymbol d][\boldsymbol d\rangle$:
\begin{align*}
\langle\boldsymbol a][\boldsymbol a\rangle\, f\,\langle\boldsymbol b][\boldsymbol b\rangle
& = \langle\boldsymbol a](\,[\boldsymbol a\rangle \langle \boldsymbol c|\boldsymbol c \rangle \hcomp f \hcomp \langle \boldsymbol d|\boldsymbol d \rangle \,\langle\boldsymbol b]\,)[\boldsymbol b\rangle
\\
& = \langle\boldsymbol a](\,[\boldsymbol a | \boldsymbol c] \hcomp  [ \boldsymbol c\rangle \, f \hcomp \langle \boldsymbol d|\boldsymbol d \rangle \,\langle\boldsymbol b]\,)[\boldsymbol b\rangle
\\
& = \langle\boldsymbol a | \boldsymbol a \rangle \hcomp \langle \boldsymbol c]([ \boldsymbol c\rangle \, f \hcomp \langle \boldsymbol d|\boldsymbol d \rangle \,\langle\boldsymbol b]\,)[\boldsymbol b\rangle
\\
& = \langle \boldsymbol c](\,[ \boldsymbol c\rangle \, f \hcomp \langle \boldsymbol d|\boldsymbol d \rangle \ \langle\boldsymbol b]\,)[\boldsymbol b\rangle
\\
& = \langle \boldsymbol c](\,[ \boldsymbol c\rangle \, f \,\langle \boldsymbol d] \,\hcomp\, [\boldsymbol d |\boldsymbol b]
\,)[\boldsymbol b\rangle
\\
& = \langle \boldsymbol c]\ \,[ \boldsymbol c\rangle \, f \,\langle \boldsymbol d]\ \,[\boldsymbol d\rangle \,\hcomp \langle\boldsymbol b |\boldsymbol b\rangle
\\
& = \langle \boldsymbol c][ \boldsymbol c\rangle \, f \,\langle \boldsymbol d][\boldsymbol d\rangle
\\
\end{align*}
\end{proof}
\end{longVersion}
\end{remark}

\begin{gray}
\begin{remark}
While technically $[\boldsymbol a\rangle\, f\,\langle\boldsymbol b]$ is just an operation symbol, the intended meaning is that 
$\langle\boldsymbol a]$  is mapping the index $i$ to the name $a_i$ and that 
$[\boldsymbol a\rangle$  is mapping the name $a_i$ to the index $i$.
\end{remark} 
\end{gray}

\begin{proposition}$\NOM:\PROP\to\nPROP$ is a functor mapping an arrow of $\PROP$s $F:\mathcal S\to\mathcal S$ to an arrow of $\nPROP$s 
$\NOM(F):\NOM(\mathcal S)\to\NOM(\mathcal S)$ defined by 
\[\NOM(F)([\boldsymbol a \rangle\, g\,\langle\boldsymbol b]) = [\boldsymbol a \rangle\, Fg\,\langle\boldsymbol b].\]
\end{proposition}

\begin{longVersion}
\begin{proof}
%Let $[\boldsymbol a \rangle g\langle\boldsymbol b]:\underline{\boldsymbol a}\to \underline{\boldsymbol b}$
We need to show that $\NOM(F)$ preserves sequential and parallel composition and is equivariant.
\begin{align*}
\NOM(F)([\boldsymbol c | \boldsymbol c']) &=\NOM(F)([\boldsymbol c \rangle \id \langle\boldsymbol c']) \\&= [\boldsymbol c \rangle F\id \langle\boldsymbol c']
\\
&= [\boldsymbol c \rangle \id \langle\boldsymbol c']
\\
&=[\boldsymbol c | \boldsymbol c']
\end{align*}

\begin{align*}
\NOM(F)([\boldsymbol a \rangle f \langle\boldsymbol c] \hcomp [\boldsymbol c'\rangle g\langle\boldsymbol b])
&=
\NOM(F)([\boldsymbol a \rangle f \hcomp \langle\boldsymbol c | \boldsymbol c'\rangle \hcomp g\langle\boldsymbol b])\\
&= [\boldsymbol a \rangle F(f \hcomp \langle\boldsymbol c | \boldsymbol c'\rangle \hcomp g)\langle\boldsymbol b]
\\
&= [\boldsymbol a \rangle Ff \hcomp F\langle\boldsymbol c | \boldsymbol c'\rangle \hcomp Fg\langle\boldsymbol b] %\color{red}\text{ this only works if we have } F\langle\boldsymbol c | \boldsymbol c'\rangle =\langle\boldsymbol c | \boldsymbol c'\rangle
\\
&= [\boldsymbol a \rangle Ff \hcomp \langle\boldsymbol c | \boldsymbol c'\rangle \hcomp Fg\langle\boldsymbol b]
\\
&= [\boldsymbol a \rangle Ff \langle\boldsymbol c] \hcomp [\boldsymbol c'\rangle Fg\langle\boldsymbol b]
\\
&= \NOM(F)([\boldsymbol a \rangle f \langle\boldsymbol c]) \hcomp \NOM(F)([\boldsymbol c'\rangle g\langle\boldsymbol b])
\end{align*}

\begin{align*}
\NOM(F)([\boldsymbol a \rangle f \langle\boldsymbol b] \uplus [\boldsymbol c \rangle g \langle\boldsymbol d])
&= \NOM(F)([\boldsymbol a + \boldsymbol c\rangle f + g \langle\boldsymbol b + \boldsymbol d])
\\
&= [\boldsymbol a + \boldsymbol c\rangle F(f + g) \langle\boldsymbol b + \boldsymbol d]
\\
&= [\boldsymbol a + \boldsymbol c\rangle Ff + Fg \langle\boldsymbol b + \boldsymbol d]
\\
&= [\boldsymbol a \rangle Ff \langle\boldsymbol b] \uplus [\boldsymbol c\rangle Fg\langle\boldsymbol d]
\\
&= \NOM(F)([\boldsymbol a \rangle f \langle\boldsymbol b]) \uplus \NOM(F)([\boldsymbol c\rangle g\langle\boldsymbol d])
\end{align*}
\end{proof}
\end{longVersion}

\begin{proposition}$\ORD$ is a functor mapping an arrow of $\nPROP$s $F:\mathcal T\to\mathcal T$ to an arrow of $\PROP$s 
$\ORD(F):\ORD(\mathcal T)\to\ORD(\mathcal T)$ defined by 
\[\ORD(F)(\langle\boldsymbol a ]\, f \,[\boldsymbol b\rangle) = \langle\boldsymbol a ]\,  Ff  \,[\boldsymbol b\rangle\]
\end{proposition}

\begin{longVersion}
\begin{proof}
We need to show that $\ORD(F)$ preserves sequential and parallel composition and symmetries.

\begin{align*}
\ORD(F)(\langle \boldsymbol c | \boldsymbol c' \rangle) &=\NOM(F)(\langle \boldsymbol c] \id [\boldsymbol c'\rangle) \\&= \langle \boldsymbol c] F\id [\boldsymbol c'\rangle
\\
&= \langle \boldsymbol c] \id [\boldsymbol c'\rangle
\\
&=\langle \boldsymbol c | \boldsymbol c'\rangle
\end{align*}

% \begin{ak}We also need to check the case $\ORD(F)(\langle\boldsymbol a ] f [\boldsymbol c \rangle \hcomp \langle \boldsymbol c'] g [\boldsymbol b\rangle)$ where the $c,c'$ have the same set of names but may be different lists\end{ak}

\begin{align*}
\ORD(F)(\langle\boldsymbol a ] f [\boldsymbol c \rangle \hcomp \langle \boldsymbol d] g [\boldsymbol b\rangle)
&=
\ORD(F)(\langle\boldsymbol a ] f \hcomp [\boldsymbol c|\boldsymbol d] \hcomp g [\boldsymbol b\rangle)\\
&= \langle\boldsymbol a ] F(f \hcomp [\boldsymbol c|\boldsymbol d] \hcomp g) [\boldsymbol b\rangle
\\
&= \langle\boldsymbol a ] Ff \hcomp F[\boldsymbol c|\boldsymbol d] \hcomp Fg [\boldsymbol b\rangle
\\
&= \langle\boldsymbol a ] Ff \hcomp [\boldsymbol c|\boldsymbol d] \hcomp Fg [\boldsymbol b\rangle
\\
&= \langle\boldsymbol a ] Ff [\boldsymbol c \rangle \hcomp \langle \boldsymbol d] Fg [\boldsymbol b\rangle
\\
&= \ORD(F)(\langle\boldsymbol a ] f [\boldsymbol c \rangle) \hcomp \ORD(F)(\langle \boldsymbol d] g [\boldsymbol b\rangle)
\end{align*}

% Again, we use the fact that the choice of internal names doesn't matter:

% \begin{align*}
% \langle\boldsymbol a ] f [\boldsymbol c \rangle \hcomp \langle \boldsymbol c] g [\boldsymbol b\rangle
% &= \langle\boldsymbol a ] f \hcomp [\boldsymbol x|\boldsymbol x] [\boldsymbol c \rangle \hcomp \langle \boldsymbol c] [\boldsymbol x|\boldsymbol x] \hcomp g [\boldsymbol b\rangle
% \\
% &= \langle\boldsymbol a ] f [\boldsymbol x \rangle 
% \langle\boldsymbol x | \boldsymbol c \rangle \hcomp \langle\boldsymbol c | \boldsymbol x \rangle
% \hcomp \langle \boldsymbol x] g [\boldsymbol b\rangle
% \\
% &= \langle\boldsymbol a ] f [\boldsymbol x \rangle \hcomp \langle \boldsymbol x] g [\boldsymbol b\rangle
% \end{align*}

\begin{align*}
\ORD(F)(\langle\boldsymbol a ] f [\boldsymbol b\rangle + \langle\boldsymbol c] g [\boldsymbol d\rangle)
&= \ORD(F)(\langle\boldsymbol a + \boldsymbol c] f \uplus g [\boldsymbol b + \boldsymbol d\rangle)
\\
&= \langle\boldsymbol a + \boldsymbol c] F(f \uplus g) [\boldsymbol b + \boldsymbol d\rangle
\\
&= \langle\boldsymbol a + \boldsymbol c] Ff \uplus Fg [\boldsymbol b + \boldsymbol d\rangle
\\
&= \langle\boldsymbol a ] Ff [\boldsymbol b\rangle + \langle\boldsymbol c] Fg [\boldsymbol d\rangle
\\
&= \ORD(F)(\langle\boldsymbol a ] f [\boldsymbol b\rangle) + \ORD(F)(\langle\boldsymbol c] g [\boldsymbol d\rangle)
\end{align*}

\end{proof}
\end{longVersion}

\begin{proposition}
For each  $\PROP$ $\mathcal S$, there is an isomorphism of $\PROP$s, natural in $\mathcal S$, 
\[%\Delta : 
\mathcal S \to \ORD(\NOM(\mathcal S))\]
mapping $f \in \mathcal S$ to $\langle\boldsymbol a][\boldsymbol a\rangle\, f\,\langle\boldsymbol b][\boldsymbol b\rangle$ for some choice of $\boldsymbol a, \boldsymbol b$. 
\end{proposition}

\begin{longVersion}
\begin{proof}
In order to show the map above is an iso, we define an inverse
\[ \Gamma: \ORD(\NOM(\mathcal S))\to\mathcal S\] 
mapping the $\langle\boldsymbol a'][\boldsymbol a\rangle\, f\,\langle\boldsymbol b][\boldsymbol b'\rangle$ generated by an $f:\underline n\to\underline n$ in $\mathcal S$ to 
$\langle\boldsymbol a'|\boldsymbol a\rangle \, \hcomp \,f \, \hcomp \,\langle\boldsymbol b|\boldsymbol b'\rangle$

\medskip\noindent
We now verify that $\Gamma (\Delta (f)) = f$ for any $f$:
\begin{align*}
\Gamma (\Delta (f))
& = \Gamma (\langle\boldsymbol a](\,[\boldsymbol a\rangle f\langle\boldsymbol b]\,)[\boldsymbol b\rangle)
\\
& = \langle\boldsymbol a|\boldsymbol a\rangle \, \hcomp \,f \, \hcomp \,\langle\boldsymbol b|\boldsymbol b\rangle
\\
& = f
\end{align*}
and
\begin{align*}
\Delta (\Gamma (\langle\boldsymbol a][\boldsymbol a'\rangle\,f\,\langle\boldsymbol b'][\boldsymbol b\rangle))
& = \Delta (\,\langle\boldsymbol a|\boldsymbol a'\rangle \hcomp f \hcomp \langle\boldsymbol b'|\boldsymbol b\rangle\,)
\\
& = \langle\boldsymbol x]\,(\ [\boldsymbol x\rangle(\,\langle\boldsymbol a|\boldsymbol a'\rangle \hcomp f \hcomp \langle\boldsymbol b'|\boldsymbol b\rangle\,)\langle\boldsymbol y]\ )\,[\boldsymbol y\rangle
\\
& = \langle\boldsymbol x]\,(\ [\boldsymbol x | \boldsymbol a]\,\hcomp\,[\boldsymbol a'\rangle(\,f \hcomp \langle\boldsymbol b'|\boldsymbol b\rangle\,)\langle\boldsymbol y]\ )\,[\boldsymbol y\rangle
\\
& = \langle\boldsymbol x | \boldsymbol x \rangle\,\hcomp\,\langle\boldsymbol a](\ [\boldsymbol a'\rangle(\,f \hcomp \langle\boldsymbol b'|\boldsymbol b\rangle\,)\langle\boldsymbol y]\ )\,[\boldsymbol y\rangle
\\
& = \langle\boldsymbol a](\,[\boldsymbol a'\rangle(\,f \hcomp \langle\boldsymbol b'|\boldsymbol b\rangle\,)\langle\boldsymbol y]\,)[\boldsymbol y\rangle
\\
& = \langle\boldsymbol a](\,[\boldsymbol a'\rangle\ f\ \langle\boldsymbol b']\ \hcomp\ [\boldsymbol b | \boldsymbol y]\,)[\boldsymbol y\rangle
\\
& = \langle\boldsymbol a]\ \,[\boldsymbol a'\rangle\ f\ \langle\boldsymbol b']\ \,[\boldsymbol b \rangle\ \hcomp\ \langle \boldsymbol y | \boldsymbol y\rangle
\\
& = \langle\boldsymbol a][\boldsymbol a'\rangle\,f\,\langle\boldsymbol b'][\boldsymbol b\rangle
\end{align*}
%

% we could have something like
% \[\Gamma_{\boldsymbol a \boldsymbol b} : \ORD(\NOM(\mathcal S)) \to \mathcal S\]
% Which is a map $\langle\boldsymbol a](\,[\boldsymbol a'\rangle f\langle\boldsymbol b']\,)[\boldsymbol b\rangle \mapsto \,\langle\boldsymbol a|\boldsymbol a'\rangle \, ; \,f \, ; \,\langle\boldsymbol b'|\boldsymbol b\rangle\,$

% \[\Delta_{\boldsymbol a \boldsymbol b} : \mathcal S \to \ORD(\NOM(\mathcal S))\]
% Which is a map $f \mapsto \langle\boldsymbol a](\,[\boldsymbol a\rangle f\langle\boldsymbol b]\,)[\boldsymbol b\rangle$

% Then we would have $\Delta_{\boldsymbol a \boldsymbol b} \circ \Gamma_{\boldsymbol a \boldsymbol b} = Id$ and $\Gamma_{\boldsymbol a \boldsymbol b} \circ \Delta_{\boldsymbol a \boldsymbol b} = Id$ for all $\boldsymbol a,\boldsymbol b$??

% However, the $\Gamma$ map is now definitely not total, as I don't think we can prove any arrow can be rewritten as $\langle\boldsymbol a](\,[\boldsymbol a'\rangle f\langle\boldsymbol b']\,)[\boldsymbol b\rangle$ for a/any specific $\boldsymbol a, \boldsymbol b$...
% \begin{ak}
% We still need to show that $\Gamma$ (or, equivalently $\Delta$) preserves the two compositions and symmetries ... I think it is easy, but can you give it a CHECK?
% \end{ak}
\medskip\noindent
We also show that $\Delta$ (and, similarly, $\Gamma$, \begin{ak}you wrote "equivalently" but is is true that this is equivalent?\end{ak} \begin{sam} i think you might have written equivalently at some point, but I've gone over them and think they hold...will add them here\end{sam}) preserves the two kinds of composition and symmetries:
\begin{align*}
\Delta(f\hcomp g) &= \langle\boldsymbol a](\,[\boldsymbol a\rangle f\hcomp g \langle\boldsymbol b]\,)[\boldsymbol b\rangle
\\
&= \langle\boldsymbol a](\,[\boldsymbol a\rangle f\langle\boldsymbol c]\hcomp[\boldsymbol c\rangle g\langle\boldsymbol b]\,)[\boldsymbol b\rangle
\\
&= \langle\boldsymbol a][\boldsymbol a\rangle\, f\,\langle\boldsymbol c][\boldsymbol c\rangle \,\hcomp\, \langle\boldsymbol c][\boldsymbol c\rangle\, g\,\langle\boldsymbol b][\boldsymbol b\rangle
\\
&= \Delta(f) \hcomp \Delta(g)
\end{align*}
In the last step of the derivation above, we used the fact that we can arbitrarily choose $\boldsymbol a, \boldsymbol b$ in $\Delta$, which follows form the last equation of Remark \ref{eqs:NomOrdDerivable}.

\begin{align*}
\Gamma(\langle\boldsymbol a][\boldsymbol a'\rangle\, f \,\langle\boldsymbol b'][\boldsymbol b\rangle \hcomp \langle\boldsymbol c][\boldsymbol c'\rangle\, g \,\langle\boldsymbol d'][\boldsymbol d\rangle) 
&= \Gamma(\langle\boldsymbol a](\,[\boldsymbol a'\rangle\, f\,\langle\boldsymbol b']\ \hcomp \ [\boldsymbol b|\boldsymbol c]\ \hcomp \ [\boldsymbol c'\rangle\, g\,\langle\boldsymbol d']\,)[\boldsymbol d'\rangle)
\\
&= \Gamma(\langle\boldsymbol a](\,[\boldsymbol a'\rangle(\, f \hcomp \langle\boldsymbol b'|\boldsymbol b \rangle \,)\langle\boldsymbol c] \,\hcomp\, [\boldsymbol c'\rangle\, g\,\langle\boldsymbol d']\,)[\boldsymbol d'\rangle)
\\
&= \Gamma(\langle\boldsymbol a](\,[\boldsymbol a'\rangle\ f\, \hcomp\, \langle\boldsymbol b'|\boldsymbol b \rangle\, \hcomp\, \langle\boldsymbol c|\boldsymbol c'\rangle\, \hcomp\, g\ \langle\boldsymbol d']\,)[\boldsymbol d'\rangle)
\\
&= \langle\boldsymbol a|\boldsymbol a'\rangle\,\hcomp\, f\, \hcomp\, \langle\boldsymbol b'|\boldsymbol b \rangle\, \hcomp\, \langle\boldsymbol c|\boldsymbol c'\rangle\, \hcomp\, g\,\hcomp\, \langle\boldsymbol d'|\boldsymbol d'\rangle
\\
&= \Gamma(\langle\boldsymbol a][\boldsymbol a'\rangle\, f \,\langle\boldsymbol b'][\boldsymbol b\rangle) \hcomp \Gamma(\langle\boldsymbol c][\boldsymbol c'\rangle\, g \,\langle\boldsymbol d'][\boldsymbol d\rangle) 
\end{align*}

\begin{align*}
\Delta(f + g) &= \langle\boldsymbol a_f + \boldsymbol a_g]\,[\boldsymbol a_f + \boldsymbol a_g\rangle\,f + g\,\langle\boldsymbol b_f + \boldsymbol b_g]\ [\boldsymbol b_f + \boldsymbol b_g\rangle
\\
&= \langle\boldsymbol a_f + \boldsymbol a_g](\,[\boldsymbol a_f\rangle\, f\,\langle\boldsymbol b_f]\,\uplus\,[\boldsymbol a_g\rangle\, g\,\langle\boldsymbol b_g]\,)[\boldsymbol b_f + \boldsymbol b_g\rangle
\\
&= \langle\boldsymbol a_f][\boldsymbol a_f\rangle\, f \,\langle\boldsymbol b_f][\boldsymbol b_f\rangle \,+\, \langle\boldsymbol a_g][\boldsymbol a_g\rangle\, g\,\langle\boldsymbol b_g][\boldsymbol b_g\rangle
\\
&= \Delta(f) + \Delta(g)
\end{align*}

\begin{align*}
\Gamma(\langle\boldsymbol a_f][\boldsymbol a'_f\rangle\,f \,\langle\boldsymbol b'_f][\boldsymbol b_f\rangle + \langle\boldsymbol a_g][\boldsymbol a'_g\rangle\,g \,\langle\boldsymbol b'_g][\boldsymbol b_g\rangle) 
&= \Gamma(\langle\boldsymbol a_f + \boldsymbol a_g](\,[\boldsymbol a'_f\rangle\,f \,\langle\boldsymbol b'_f]\ \,\uplus\ \,[\boldsymbol a'_g\rangle\,g \,\langle\boldsymbol b'_g]\,)[\boldsymbol b_f + \boldsymbol b_g\rangle)
\\
&= \Gamma(\langle\boldsymbol a_f + \boldsymbol a_g](\,[\boldsymbol a'_f + \boldsymbol a'_g\rangle\,f + g\,\langle\boldsymbol b'_f + \boldsymbol b'_g]\,)[\boldsymbol b_f + \boldsymbol b_g\rangle)
\\
&= \langle\boldsymbol a_f +\boldsymbol a_g|\boldsymbol a'_f + \boldsymbol a'_g\rangle\ \, \hcomp\ \, (f + g)\ \, \hcomp\ \, \langle\boldsymbol b'_f + \boldsymbol b'_g|\boldsymbol b_f + \boldsymbol b_g\rangle
\\
&= (\langle\boldsymbol a_f|\boldsymbol a'_f\rangle + \langle\boldsymbol a_g|\boldsymbol a'_g\rangle)\,\hcomp\, (f + g)\, \hcomp\, (\langle\boldsymbol b'_f|\boldsymbol b_f\rangle + \langle\boldsymbol b'_g|\boldsymbol b_g\rangle)
\\
&= (\langle\boldsymbol a_f|\boldsymbol a'_f\rangle\, \hcomp\, f\, \hcomp\, \langle\boldsymbol b'_f|\boldsymbol b_f\rangle) + (\langle\boldsymbol a_g|\boldsymbol a'_g\rangle\, \hcomp\, g\, \hcomp\, \langle\boldsymbol b'_g|\boldsymbol b_g\rangle)
\\
&= \Gamma(\langle\boldsymbol a_f][\boldsymbol a'_f\rangle\,f \,\langle\boldsymbol b'_f][\boldsymbol b_f\rangle)\, + \,\Gamma(\langle\boldsymbol a_g][\boldsymbol a'_g\rangle\,g \,\langle\boldsymbol b'_g][\boldsymbol b_g\rangle) 
\end{align*}

\begin{align*}
\Delta(\langle \boldsymbol x| \boldsymbol x'\rangle) 
&= \langle\boldsymbol a]\ [\boldsymbol a\rangle\, \langle\boldsymbol x| \boldsymbol x'\rangle \,\langle\boldsymbol b]\ [\boldsymbol b\rangle
\\
&= \langle\boldsymbol a](\,[\boldsymbol a|\boldsymbol x]\,\hcomp[\boldsymbol x'|\boldsymbol b]\,)[\boldsymbol b\rangle
\\
&= \langle\boldsymbol a | \boldsymbol a\rangle \hcomp \langle \boldsymbol x]\ \,[\boldsymbol x'|\boldsymbol b]\ \,[\boldsymbol b\rangle
\\
&= \langle \boldsymbol x]\,[\boldsymbol x'|\boldsymbol b]\,[\boldsymbol b\rangle
\\
&= \langle \boldsymbol x| \boldsymbol x'\rangle \hcomp \langle \boldsymbol b|\boldsymbol b\rangle
\\
&= \langle \boldsymbol x| \boldsymbol x'\rangle
\end{align*}

\begin{align*}
\Gamma(\langle \boldsymbol a| \boldsymbol a'\rangle) 
&= \Gamma(\langle\boldsymbol a][\boldsymbol a\rangle\, \langle\boldsymbol a| \boldsymbol a'\rangle \,\langle\boldsymbol a'][\boldsymbol a'\rangle) 
\\
&= \langle\boldsymbol a | \boldsymbol a\rangle \hcomp \langle\boldsymbol a| \boldsymbol a'\rangle \hcomp \langle\boldsymbol a' | \boldsymbol a'\rangle
\\
&= \langle \boldsymbol a| \boldsymbol a'\rangle
\end{align*}

\begin{ak}Where do we use that $\mathcal S$ satisfies naturality? The proof that $\Gamma$ preserves equations can only work under this condition ... so we must use use it in hte proof ... \end{ak}
\begin{sam}
Should this be where we define $\ORD,\NOM$? I thought they had to preserve equations of $\PROP,\nPROP$ ... or? \begin{ak} just think about what happens if we allow $\mathcal S$ to be a $\PROP$ in the sense of Lack (no naturality required) ... then we still have an adjunction, $\Delta$ is still well defined, but $\Gamma$ is not ... \end{ak}
\end{sam}
\end{proof} 
\end{longVersion}

% \begin{proposition}
% For each  $\PROP$ $\mathcal S$, there is an isomorphism of $\PROP$s.
% \[\Gamma : \ORD(\NOM(\mathcal S))\to\mathcal S\] mapping the $\langle\boldsymbol a'][\boldsymbol a\rangle\, f\,\langle\boldsymbol b][\boldsymbol b'\rangle$ generated by an $f:\underline n\to\underline n$ in $\mathcal S$ to 
% $\langle\boldsymbol a'|\boldsymbol a\rangle \, \hcomp \,f \, \hcomp \,\langle\boldsymbol b|\boldsymbol b'\rangle$ 
% %\begin{sam}
% %%ak: not needed anymore now that we have shown that it is an ismorphism ...
% %Note that for this map to be total, we need to show that every arrow in $\ORD(\NOM(\mathcal S))$ os of the form $\langle\boldsymbol a'](\,[\boldsymbol a\rangle f\langle\boldsymbol b]\,)[\boldsymbol b'\rangle$, or rather is equivalent to one of this shape, by the equations of \ref{def:NOM} and \ref{def:ORD}. This is easy enough to show by induction on the shape of an arbitrary diagram in $\ORD(\NOM(\mathcal S))$... I think....
% %\end{sam}
% \end{proposition}

\begin{proposition}
For each  $\nPROP$ $\mathcal T$, there is an isomorphism of $\nPROP$s, natural in $\mathcal T$, 
\[%\Gamma_{\mathsf n} : 
\NOM(\ORD(\mathcal T))\to\mathcal T\]
mapping the $[\boldsymbol c\rangle\langle \boldsymbol a]\,f\,[\boldsymbol b\rangle\langle\boldsymbol d]$ generated by an $f:\underline{\boldsymbol a}\to\underline{\boldsymbol b}$ in $\mathcal T$ to $[\boldsymbol c|\boldsymbol a] \, ; f \, ; \,[\boldsymbol b|\boldsymbol d]\,$.
\end{proposition}

\begin{longVersion}
\begin{proof}
We define a converse $\Delta_{\mathsf n} : \mathcal T\to \NOM(\ORD(\mathcal T))$ mapping $f:\underline{\boldsymbol a}\to\underline{\boldsymbol b}$ to $[\boldsymbol a\rangle\langle \boldsymbol a] \,f\,[\boldsymbol b\rangle\langle\boldsymbol b]$ for some choice of $\boldsymbol a, \boldsymbol b$.

\medskip\noindent
We now verify that $\Gamma_{\mathsf n} (\Delta_{\mathsf n} (f)) = f$ for any $f$:
\begin{align*}
\Gamma_{\mathsf n} (\Delta_{\mathsf n} (f))
& = \Gamma_{\mathsf n} ([\boldsymbol a\rangle\langle \boldsymbol a]\,f\,[\boldsymbol b\rangle\langle\boldsymbol b])
\\
& = [\boldsymbol a|\boldsymbol a] \, \hcomp \,f \, \hcomp \,[\boldsymbol b|\boldsymbol b]
\\
& = f
\end{align*}
and

\begin{align*}
\Delta_{\mathsf n} (\Gamma_{\mathsf n} ([\boldsymbol c\rangle\langle\boldsymbol a]\,f\,[\boldsymbol b\rangle\langle\boldsymbol d]))
& = \Delta (\,[\boldsymbol c|\boldsymbol a] \hcomp f \hcomp [\boldsymbol b|\boldsymbol d]\,)
\\
& = [\boldsymbol c\rangle(\,\langle\boldsymbol c](\ [\boldsymbol c|\boldsymbol a] \hcomp f \hcomp [\boldsymbol b|\boldsymbol d]\,)[\boldsymbol d\rangle\,)\langle\boldsymbol d]
\\
& = [\boldsymbol c\rangle(\,\langle\boldsymbol c | \boldsymbol c\rangle\hcomp\langle\boldsymbol a](\,f \hcomp [\boldsymbol b|\boldsymbol d]\,)[\boldsymbol d\rangle\,)\langle\boldsymbol d]
\\
& = [\boldsymbol c\rangle(\,\langle\boldsymbol a](\,f \hcomp [\boldsymbol b|\boldsymbol d]\,)[\boldsymbol d\rangle\,)\langle\boldsymbol d]
\\
& = [\boldsymbol c\rangle(\,\langle\boldsymbol a]\ f\ [\boldsymbol b\rangle\,\hcomp\,\langle\boldsymbol d | \boldsymbol d\rangle\,)\langle\boldsymbol d]
\\
& = [\boldsymbol c\rangle\langle\boldsymbol a]\,f\,[\boldsymbol b\rangle\langle\boldsymbol d ]
\end{align*}

\medskip\noindent
We also show that $\Delta_{\mathsf n}$ (or equivalently $\Gamma_{\mathsf n}$\begin{ak} is it true that this is equivalent?\end{ak}) preserves the two kinds of composition and symmetries:
\begin{align*}
\Delta_{\mathsf n}(f\hcomp g) &= [\boldsymbol a\rangle(\,\langle\boldsymbol a] f\hcomp g [\boldsymbol b\rangle\,)\langle\boldsymbol b]
\\
&= [\boldsymbol a\rangle(\,\langle\boldsymbol a] f[\boldsymbol c\rangle\hcomp \langle\boldsymbol c]g[\boldsymbol b\rangle\,)\langle\boldsymbol b]
\\
&= [\boldsymbol a\rangle\langle\boldsymbol a]\, f\,[\boldsymbol c\rangle\langle\boldsymbol c]\,\hcomp\, [\boldsymbol c\rangle\langle\boldsymbol c]\, g\,[\boldsymbol b\rangle\langle\boldsymbol b]
\\
&= \Delta_{\mathsf n}(f) \hcomp \Delta_{\mathsf n}(g)
\end{align*}

\begin{align*}
\Gamma_{\mathsf n}([\boldsymbol c\rangle\langle\boldsymbol a]\, f \,[\boldsymbol b\rangle\langle\boldsymbol d] \hcomp [\boldsymbol d'\rangle\langle\boldsymbol e]\, g \,[\boldsymbol f\rangle\langle\boldsymbol h]) 
&= \Gamma_{\mathsf n}([\boldsymbol c\rangle\langle\boldsymbol a]\, f \,[\boldsymbol b\rangle \,\hcomp\, \langle\boldsymbol d|\boldsymbol d'\rangle \,\hcomp \langle\boldsymbol e]\, g \,[\boldsymbol f\rangle\langle\boldsymbol h]) 
\\
&= \Gamma_{\mathsf n}([\boldsymbol c\rangle\langle\boldsymbol a]\,(f \hcomp [\boldsymbol b |\boldsymbol d])\,[\boldsymbol d'\rangle \hcomp \langle\boldsymbol e]\, g \,[\boldsymbol f\rangle\langle\boldsymbol h]) 
\\
&= \Gamma_{\mathsf n}([\boldsymbol c\rangle\langle\boldsymbol a]\,(f \hcomp [\boldsymbol b |\boldsymbol d] \hcomp [\boldsymbol d'|\boldsymbol e]\,\hcomp\, g)\,[\boldsymbol f\rangle\langle\boldsymbol h])
\\
&= [\boldsymbol c|\boldsymbol a]\hcomp f \hcomp [\boldsymbol b |\boldsymbol d] \hcomp [\boldsymbol d'|\boldsymbol e] \hcomp g \hcomp [\boldsymbol f|\boldsymbol h]
\\
&= \Gamma_{\mathsf n}([\boldsymbol c\rangle\langle\boldsymbol a]\, f \,[\boldsymbol b\rangle\langle\boldsymbol d]) \hcomp \Gamma_{\mathsf n}([\boldsymbol d'\rangle\langle\boldsymbol e]\, g \,[\boldsymbol f\rangle\langle\boldsymbol h]) 
\end{align*}

\begin{align*}
\Delta_{\mathsf n}(f + g) &= [\boldsymbol a_f + \boldsymbol a_g\rangle\,\langle\boldsymbol a_f + \boldsymbol a_g]\,f + g\,[\boldsymbol b_f + \boldsymbol b_g\rangle\ \langle\boldsymbol b_f + \boldsymbol b_g]
\\
&= [\boldsymbol a_f + \boldsymbol a_g\rangle(\,\langle\boldsymbol a_f]\, f\,[\boldsymbol b_f\rangle\,\uplus\,\langle\boldsymbol a_g]\, g\,[\boldsymbol b_g\rangle\,)\langle\boldsymbol b_f + \boldsymbol b_g]
\\
&= [\boldsymbol a_f \rangle\langle\boldsymbol a_f]\, f \,[\boldsymbol b_f\rangle\langle\boldsymbol b_f] \,+\, [\boldsymbol a_g\rangle\langle\boldsymbol a_g]\, g\,[\boldsymbol b_g\rangle\langle\boldsymbol b_g]
\\
&= \Delta_{\mathsf n}(f) + \Delta_{\mathsf n}(g)
\end{align*}

\begin{align*}
\Gamma_{\mathsf n}([\boldsymbol a'_f\rangle\langle\boldsymbol a_f]\,f \,[\boldsymbol b_f\rangle\langle\boldsymbol b'_f] \uplus [\boldsymbol a'_g\rangle\langle\boldsymbol a_g]\,g \,[\boldsymbol b_g\rangle\langle\boldsymbol b'_g]) 
&= \Gamma_{\mathsf n}([\boldsymbol a'_f + \boldsymbol a'_g\rangle(\,\langle\boldsymbol a_f]\,f \,[\boldsymbol b_f\rangle\ +\ \langle\boldsymbol a_g]\,g \,[\boldsymbol b_g\rangle\,)\langle\boldsymbol b'_f + \boldsymbol b'_g])
\\
&= \Gamma_{\mathsf n}([\boldsymbol a'_f + \boldsymbol a'_g\rangle(\,\langle\boldsymbol a_f + \boldsymbol a_g]\,f \uplus g\,[\boldsymbol b_f + \boldsymbol b_g\rangle\,)\langle\boldsymbol b'_f + \boldsymbol b'_g])
\\
&= [\boldsymbol a'_f +\boldsymbol a'_g|\boldsymbol a_f + \boldsymbol a_g]\  \hcomp\  (f \uplus g)\  \hcomp\  [\boldsymbol b_f + \boldsymbol b_g|\boldsymbol b'_f + \boldsymbol b'_g]
\\
&= ([\boldsymbol a'_f|\boldsymbol a_f] + [\boldsymbol a'_g|\boldsymbol a_g])\,\hcomp\, (f \uplus g)\, \hcomp\, ([\boldsymbol b_f|\boldsymbol b'_f] + [\boldsymbol b_g|\boldsymbol b'_g])
\\
&= ([\boldsymbol a'_f|\boldsymbol a_f]\, \hcomp\, f\, \hcomp\, [\boldsymbol b_f|\boldsymbol b'_f]) + ([\boldsymbol a'_g|\boldsymbol a_g]\, \hcomp\, g\, \hcomp\, [\boldsymbol b_g|\boldsymbol b'_g])
\\
&= \Gamma_{\mathsf n}([\boldsymbol a'_f\rangle\langle\boldsymbol a_f]\,f \,[\boldsymbol b_f\rangle\langle\boldsymbol b'_f])\, + \,\Gamma_{\mathsf n}([\boldsymbol a'_g\rangle\langle\boldsymbol a_g]\,g \,[\boldsymbol b_g\rangle\langle\boldsymbol b'_g])
\end{align*}

\begin{align*}
\Delta_{\mathsf n}([\boldsymbol x| \boldsymbol y]) 
&= [\boldsymbol x\rangle\langle\boldsymbol x]\, [\boldsymbol x| \boldsymbol y] \,[\boldsymbol y\rangle\langle\boldsymbol y]
\\
&= [\boldsymbol x\rangle(\,\langle\boldsymbol x|\boldsymbol x\rangle\,\hcomp\,\langle\boldsymbol y|\boldsymbol y\rangle\,)\langle\boldsymbol y]
\\
&= [\boldsymbol x\rangle\, \id \,\langle\boldsymbol y]
\\
&= [\boldsymbol x| \boldsymbol y]
\end{align*}

\begin{align*}
\Gamma_{\mathsf n}([\boldsymbol a| \boldsymbol b]) 
&= \Gamma([\boldsymbol a\rangle\langle\boldsymbol a]\, [\boldsymbol a| \boldsymbol b] \,[\boldsymbol b\rangle\langle\boldsymbol b]) 
\\
&= [\boldsymbol a | \boldsymbol a] \hcomp[\boldsymbol a| \boldsymbol b] \hcomp [\boldsymbol b | \boldsymbol b]
\\
&= [\boldsymbol a| \boldsymbol b]
\end{align*}
\end{proof}
\end{longVersion}

% \medskip\noindent
% Recall that  $\langle\boldsymbol a'|\boldsymbol a\rangle$ denotes the bijection given by $i\mapsto a'_i=a_j\mapsto j$.

Since the last two propositions provide an isomorphic unit and counit of an adjunction, we obtain

\begin{theorem}\label{thm:equivalence}
The categories $\PROP$ and $\nPROP$ are equivalent.
\end{theorem}

\begin{remark}
If we generalise the notion of prop from MacLane~\cite{maclane} to Lack~\cite{lack}, in other words, if we drop the last equation of Fig~\ref{fig:symmetric-monoidal-category} expressing the naturality of symmetries, we still obtain an adjunction, in which $\NOM$ is left-adjoint to $\ORD$. Nominal props then are a full reflective subcategory of ordinary props. In other words, the (generalised) props $\mathcal S$ that satsify naturality of symmetries are exactly those for which $\mathcal S\cong \ORD(\NOM(\mathcal S))$.
\end{remark}

\begin{gray}
\begin{proof}
... is there anything left to prove? ... yes the triangle equalities
\end{proof}
\end{gray}

\begin{gray}
\section{A transfer theorem}

THIS SECTION IS EXPERIMENTAL ... NOT SURE HOW TO MAKE THE TRANSFER OF COMPLETENESS PRECISE ...

Above we have shown that the categories $\PROP$ and $\nPROP$ are equivalent. Here we will show in which sense each $\PROP$ is equivalent to an $\nPROP$ and vices versa. In other words, for each $\PROP$ $\mathcal S$ and each $\nPROP$ $\mathcal T$ we will study morphisms 
\[\NOM(\mathcal S)\to\mathcal S \quad\quad\quad\quad \mathcal T\to\ORD(\mathcal T)\]
that are equivalences of categories and ``weak'' equivalences of internal monoidal categories.  As a corollary we will obtain a theorem that allows us to transfer completeness results between $\PROP$s and $\nPROP$s. In particular, the completeness of the calculi of Section~\ref{} will follow.

\medskip\noindent
In order to define functors above we need to introduce the notion of local enumeration. The idea comes from the observation that the above equivalences are variations of the equivalence between the category of finite sets and its skeleton (with natural numbers as objects). The equivalences between these categories are in bijection with local enumerations of names:

\begin{definition} A \emph{local enumeration of names} $\boldsymbol e$ is a function mapping each finite set of names $A$ to a list $\boldsymbol e_A=\boldsymbol a = (a_1,\ldots a_n)$. 
\end{definition}

\begin{proposition}
Let $\boldsymbol e$ be a local enumeration of names. For each $\nPROP$ $\mathcal T$, $\boldsymbol e$ determines an internal functor 
\[\boldsymbol e:\mathcal T\to\ORD(\mathcal T)\]
also denoted by  $\boldsymbol e$, mapping $f:A\to B$ to $\langle\boldsymbol e_A] f[\boldsymbol e_B\rangle :\underline n\
\to \underline m$ where $n,m$ are the cardinalities of $A,B$, respectively. 
\end{proposition}

\begin{proof}
That $\boldsymbol e$ is a functor follows immediately from the definition of $\ORD(\mathcal T)$ (CHECK).
\end{proof}

\begin{proposition}
The diagram 
\[\xymatrix@C=50pt{
\mathcal T\ar[r]^F \ar[d]_{\boldsymbol e}& \mathcal T'\ar[d]^{\boldsymbol e}\\
\ORD(\mathcal T)\ar[r]_{\ORD(F)} & \ORD(\mathcal T')
}\]
commutes for all local enumeration of names $\boldsymbol e$ and all $F:\mathcal T\to\mathcal T'$ in $\nPROP$. 

\end{proposition}

\begin{proof}
$\ORD(F)$ was defined in ... as ... now ... DETAILS NEEDED (but should be straightforward)
\end{proof}

\begin{proposition}
Let $\boldsymbol e$ be a local enumeration of names.
Then for each $\PROP$ $\mathcal S$, $\boldsymbol e$ determines a functor 
\[\boldsymbol e:\NOM(\mathcal 
S)\to\mathcal S\]
\begin{gray} I need to think more about this one ...
also denoted by  $\boldsymbol e$, mapping $f:A\to B$ to $\langle\boldsymbol e_A] f[\boldsymbol e_B\rangle :\underline n\
\to \underline m$ where $n,m$ are the cardinalities of $A,B$, respectively. 

$\NOM(\mathcal S)\to\ORD(\NOM(\mathcal S))\to\mathcal S$. Moreover, this functor are strong monoidal and equivariant (as internal functors in $\Cat(\Nom,1,\ast)$).
\end{gray}
\end{proposition}

\begin{proof}
...
\end{proof} 

\begin{proposition}
Given an arrow $G:\mathcal S\to\mathcal S'$ in $\PROP$ there is a unique arrow $\NOM(G)$ in $\nPROP$ such that 
\[\xymatrix@C=50pt{
\NOM(\mathcal S)\ar[r]^{\NOM(G)} \ar[d]_{\boldsymbol e}& \NOM(\mathcal S')\ar[d]^{\boldsymbol e}\\
\mathcal S\ar[r]_G & \mathcal S'
}\]
commutes for all (equivalently, one) local enumeration of names $\boldsymbol e$.
\end{proposition}

\begin{proof}
This needs to be checked carefully:
\begin{ak}
it is clear (?) how to define $\NOM(G)$ from $G$ and $\boldsymbol e$ ... moreover, given the  specific choice of $\boldsymbol e$, $\NOM(G)$ should be unique wrt ... \end{ak}
\end{proof}

\begin{theorem}\label{thm:transfer}
The two functors form an equivalence of categories ... this needs to be more precise ...
\end{theorem}

\begin{proof}
Both $\PROP$ and $\nPROP$ are monoidal categories internal in nominal sets with the separating tensor. The latter form a 2-category as shown in Theorem~\ref{thm:moncatv}, in which the notion of equivalence of categories has a precise meaning. ...
\end{proof}

\begin{theorem}\label{thm:ordisleftadjoint}
The functor $\ORD:\nPROP\to\PROP$ is \end{theorem}

\begin{corollary} ... on completeness ... $\mathcal T$ is complete if $\ORD(\mathcal T)$ is complete ...
\end{corollary}

\medskip\noindent what about:  $\mathcal S$ is complete iff $\NOM(\mathcal S)$ is complete
\end{gray}

\section{Conclusion}

The equivalence of nominal and ordinary props (Theorem~\ref{thm:equivalence}) has a satsifactory graphical interpretation. Indeed, comparing Figs~\ref{fig:smt-theories} and~\ref{fig:nmt-theories} we see that both share, modulo different labellings of wires mediated by the functors $\ORD$ and $\NOM$, the same core of generators and equations while the difference lies only in the equations expressing, on the one hand,  that $\oplus$ has natural symmetries and, on the other hand, that generators are a nominal set.

\medskip
There are several directions for future research. First, the notion of an internal monoidal category has been developed because it is easier to prove the basic results in general rather than only in the special case of nominal sets. Nevertheless, it would be interesting to explore whether there are other interesting instances of internal monoidal categories. 

\medskip
Second, internal monoidal categories are a principled way to build monoidal categories with a partial tensor. For example, by working internally in the category of nominal sets with the separated product we can capture in a natural way constraints such as the tensor $f\oplus g$ for two partial maps $f,g:\names\to V$ being defined only if the domains of $f$ and $g$ are disjoint. This reminds us of the work initiated by O'Hearn and Pym on categorical and algebraic models for separation logic and other resource logics, see eg \cite{ohearn-pym,galmiche-etal,struth}. It seems promising to investigate how to build categorical models for resource logics based on internal monoidal theories. In one direction, one could extend the work of Curien and Mimram~\cite{curien-mimram} to partial monoidal categories.

\medskip
Third, there has been substantial progress in exploiting Lack's work on composing PROPs \cite{lack} in order to develop novel string diagrammatic calculi for a wide range of applications, see eg \cite{rewriting-modulo,signal-flow-1}. It will be interesting to explore how much of this technology can be transferred from props to nominal props. 

\medskip 
Fourth, various applications of nominal string digrams could be of interest. The orginial motivation for our work was to obtain a convenient calculus for simulataneous substitutions that can be integrated with multi-type display calculi \cite{multi-type} and, in particular, with the multi-type display calculus for first-order logic of Tzimoulis~\cite{tzimoulis}. Another direction for applications comes from the work of Ghica and Lopez~\cite{ghica-lopez} on a nominal syntax for string diagrams. In particular, it would of interest to add various binding operations to nominal props.

\newpage
\appendix
\section{Some internal category theory}

See eg Borceux, Handbook of Categorical Algebra, Volume 1, Chapter 8 and the \href{https://ncatlab.org/nlab/show/internal+category}{nlab}.

\begin{remark}[internal category] 
\label{def:internal-cat} 
In a category with finite limits an \emph{internal category}  is a diagram 
\[ 
\xymatrix@C=10ex{  
A_2 \  
%\ar@<2ex>[0,1]^-{d^2_2}  
%\ar[0,1]|-{d^2_1}  
\ar[0,1]|-{\ \comp\ }  
%\ar@<-2ex>[0,1]_-{d^2_0}  
&  
{\ \ A_1 \ \ } %\ar@(ur,ul)_{s} 
%\ar@<2ex>[0,1]^-{d^1_1}  
\ar@<2ex>[0,1]^-{\dom}  
%\ar@<-2ex>[0,1]_-{d^1_0}  
\ar@<-2ex>[0,1]_-{\cod}  
&  
\ A_0  
\ar[0,-1]|-{\ i \ }  
}  
\]  
where  
\begin{enumerate} 
\item $A_2$ is a pullback \  \ 
$\vcenter{ 
\xymatrix{ 
A_2
\ar[r]^{\pi_2}
\ar[d]_{\pi_1} 
& 
A_1\ar[d]^{\dom} \\ 
A_1\ar[r]^{\cod}
& 
A_0 
} 
}$ 
\item $\dom\circ \comp=\dom\circ \pi_1$ \ \ and \ $\cod\circ 
  \comp=\cod\circ \pi_2$, 
\item $\dom\circ i  = \id_{A_0} = \cod\circ i$, 
\item $\comp\circ \langle i\circ\dom,\id_{A_1} \rangle = \id_{A_1} = \comp\circ \langle\id_{A_1},i\circ\cod \rangle$
\item $\comp\,\circ\,\compl= \comp\,\circ\,\compr$
\end{enumerate}  
where 
\begin{itemize}
\item $\langle i\circ\dom,\id_{A_1}\rangle :A_1\to A_2$ and $\langle \id_{A_1},i\circ\cod\rangle :A_1\to A_2$ are the arrows into the pullback $A_2$  pairing $i\circ\dom,\id_{A_1}:A_1\to A_1$ and $\id_{A_1},i\circ\cod:A_1\to A_1$, respectively.
\item the ``triple of arrows''-object $A_3$ is the pullback \ \ \ \ 
\[\vcenter{ 
\xymatrix{ 
A_3
\ar[r]^{\textit{right}}
\ar[d]_{\textit{left}} 
& 
A_2\ar[d]^{\pi_1} \\ 
A_2\ar[r]^{\pi_2}
& 
A_1
} 
}\]
where $\textit{left}$ ``projects out the left two arrows'' and $\textit{right}$ ``projects out the right two arrows''  
\item $\compl$ is the arrow composing the ``left two arrows''
\[\vcenter{ 
\xymatrix@C=10ex{ 
A_3 
\ar@/^/[rrd]^{\ \ \pi_2\,\circ\, \textit{right}}
\ar@{..>}[dr]|-{\compl}
\ar@/_/[rdd]_{\comp\,\circ\, \textit{left}}
&&\\
&
A_2
\ar[r]^{\pi_2}
\ar[d]_{\pi_1} 
& 
A_1\ar[d]^{\dom} 
\\ 
&
A_1\ar[r]^{\cod}
& 
A_0 
} 
}\]
\item $\compr$ is the arrow composing the ``right two arrows''
\[\vcenter{ 
\xymatrix@C=10ex{ 
A_3 
\ar@/^/[rrd]^{\ \ \comp\,\circ\, \textit{right}}
\ar@{..>}[dr]|-{\compr}
\ar@/_/[rdd]_{\pi_1\,\circ\, \textit{left}}
&&\\
&
A_2
\ar[r]^{\pi_2}
\ar[d]_{\pi_1} 
& 
A_1\ar[d]^{\dom} 
\\ 
&
A_1\ar[r]^{\cod}
& 
A_0 
} 
}\]
\end{itemize}

1.~and 2.~define $A_2$ as the `object of composable pairs of arrows' while 3.~and 4.~express that the `object  of arrows' $A_1$ has identities and 5.~formalises associativity of composition.
\end{remark}

\begin{remark}
A morphism $f:A\to B$ between internal categories, an \emph{internal functor}, is a triple $(f_0,f_1,f_2)$ of arrows such that the two diagrams (one for $\dom$ and one for $\cod$)
\begin{gray}ADD IDENTITIES\end{gray}
\begin{equation}\label{eq:def:internal-functor}
\vcenter{
\xymatrix@C=10ex{  
A_2 \  
%\ar@<2ex>[0,1]^-{d^2_2}  
%\ar[0,1]|-{d^2_1}  
\ar[0,1]|-{\ \comp\ }  
%\ar@<-2ex>[0,1]_-{d^2_0}  
\ar[dd]_{f_2}
&  
{\ \ A_1 \ \ } %\ar@(ur,ul)_{s} 
%\ar@<2ex>[0,1]^-{d^1_1}  
\ar@<1ex>[0,1]^-{\dom}  
%\ar@<-2ex>[0,1]_-{d^1_0}  
\ar@<-1ex>[0,1]_-{\cod}  
\ar[dd]_{f_1}
&  
\ A_0  
%\ar[0,-1]|-{i^0_0}  
\ar[dd]_{f_0}
\\
&&
\\
B_2 \  
%\ar@<2ex>[0,1]^-{d^2_2}  
%\ar[0,1]|-{d^2_1}  
\ar[0,1]|-{\ \comp\ }  
%\ar@<-2ex>[0,1]_-{d^2_0}  
&  
{\ \ B_1 \ \ } %\ar@(ur,ul)_{s} 
%\ar@<2ex>[0,1]^-{d^1_1}  
\ar@<1ex>[0,1]^-{\dom}  
%\ar@<-2ex>[0,1]_-{d^1_0}  
\ar@<-1ex>[0,1]_-{\cod}  
&  
\ B_0  
%\ar[0,-1]|-{i^0_0}  
} }
\end{equation}
commute and identities are preserved.
Since $A_2$ is a pullback, $f_2$ is uniquely determined by $f_1$ (in other words, the existence of $f_2$ is a property rather than a structure). In more detail, if $\Gamma\to A_2$ is any arrow then, because $A_2$ is a pullback, it can be written as a pair $\langle l,r\rangle$ of arrows $l,r:\Gamma\to A_1$ and $f_2$ is determined by $f_1$ via
\begin{equation}\label{equ:f2}
f_2\circ\langle l,r\rangle=\langle f_1\circ l,f_1\circ r\rangle
\end{equation}
\end{remark}

\begin{remark}
A natural transformation $\alpha:f\to g$ between internal functors $f,g:A\to B$, an \emph{internal natural transformation},   is an arrow $\alpha:A_0\to B_1$ such that 
\[
\dom\circ\alpha = f_0
\quad\quad
\cod\circ\alpha = g_0
\quad\quad
\comp\circ(f_1,\alpha\circ\cod)=\comp(\alpha\circ\dom,g_1)
\]
\end{remark}

\begin{remark}
Internal categories with functors and natural transformations form a 2-category. We denote by $Cat(\mathcal V)$ the category or 2-category of categories internal in $\mathcal V$. The forgetful functor $Cat(\mathcal V)\to\mathcal C$ mapping an internal category $A$ to its object of objects $A_0$ has both left and right adjoints and, therefore, preserves limits and colimits. Moreover, a limit of internal categories is computed componentwise as $(\lim D)_j=\lim (D_j)$ for $j=0,1,2$.
\end{remark}

\begin{remark}
A monoidal category can be thought of both as a monoid in the category of categories and as a category internal in the category of monoids. To understand this in more detail, note that both cases give rise to the diagram
\[ 
\xymatrix@C=20ex{  
A_2\times A_2 \  
\ar[0,1]|-{\ \comp\times \comp\ }  
\ar[dd]_{m_2}
&  
{\ \ A_1\times A_1 \ \ }   
\ar@<1ex>[0,1]^-{\dom\times \dom}  
\ar@<-1ex>[0,1]_-{\cod\times \cod}  
\ar[dd]_{m_1}
&  
\ A_0\times A_0  
\ar[dd]_{m_0}
\\
&&
\\
A_2 \  
\ar[0,1]|-{\ \comp\ }  
&  
{\ \ A_1 \ \ }   
\ar@<1ex>[0,1]^-{\dom}  
\ar@<-1ex>[0,1]_-{\cod}  
&  
\ A_0  
}  
\] 
where 
\begin{itemize}
\item in the case of a monoid $A$ in the category of internal categories, $m=(m_0,m_1,m_2)$ is an internal functor $A\times A \to A$ and, using that products of internal categories are computed componentwise, we have $\comp\circ m_2=m_1\circ (\comp\times \comp)$, which gives us the interchange law \[(f\then\, g)\cdot(f'\then\, g')= (f\cdot f')\,\then\,(g\cdot g')\]
by using \eqref{equ:f2} with $m$ for $f$ and  writing $\then$ for $\comp$ and  $\cdot$ for $m_1$;
\item in the case of a category internal in monoids we have monoids $A_0,A_1,A_2$ and monoid homomorphisms $i,\dom,\cod,\comp$ which, if spelled out, leads to the same commuting diagrams as the previous item.
\end{itemize}
\end{remark}

\begin{gray}
WHAT FOLLOWS has been partly moved into the main part ...

\medskip
In the main part of this paper we answer the question of how to internalise the Cartesian product $A\times A$ in a monoidal category $(\mathcal V,I,\otimes)$. In other words, where above we use $A_0\times A_0$ we now want $A_0\otimes A_0$. The first idea is to replace in the diagram above $A_n\times A_n$ by $A_n\otimes A_n$, $n =0,1,2$. But there are two problems with this.

\begin{itemize}
\item First, $\otimes$ may not preserve pullbacks, in which case $A\otimes A$ may not be an internal category.
\item Second, in the motivating example $(\Nom,1,\ast)$ with $\ast$\ the separating product, we want $A_0\ast A_0$ on objects but not on arrows. 
\end{itemize}

\begin{example} Consider the category $\Fun$ of finite sets of names with all functions as a category internal in $(\Nom,1,\ast)$. Define a binary operation $\Fun\ast\Fun$ as $(\Fun\ast\Fun)_0=\Fun_0\ast\Fun_0$ and $(\Fun\ast\Fun)_1=\Fun_1\ast\Fun_1$. Then $\Fun\ast\Fun$ cannot be equipped with the structure of an internal category. Indeed, assume for a contradiction that there was an appropriate pullback $(\Fun\ast \Fun)_2$ and arrow $\comp$ such that the two diagrams
\[
\xymatrix@C=20ex{  
(\Fun\ast \Fun)_2 \  
\ar[0,1]|-{\ \comp\ }  
\ar[dd]_{\pi_1}^{\pi_2}
&  
{\ \ \Fun_1\ast \Fun_1 \ \ }   
%\ar@<1ex>[0,1]^-{\dom\times \dom}  
%\ar@<-1ex>[0,1]_-{\cod\times \cod}  
\ar[dd]_{\dom}^{\cod}
%&  
%\ \Fun_0\ast \Fun_0  
%\ar[dd]_{\uplus_0}
%
\\
&
\\
\Fun_1\ast\Fun_1\  
\ar[0,1]^{\ \dom\ }_{\cod}  
&  
{\ \ \Fun_0\ast\Fun_0 \ \ }   
%\ar@<1ex>[0,1]^-{\dom}  
%\ar@<-1ex>[0,1]_-{\cod}  
%&  
%\ \Fun_0  
}  
\]   
commute. Let $\delta_{xy}:\{x\}\to\{y\}$ be the unique functions in $\Fun$ of the indicated type. Then $((\delta_{ac},\delta_{bd}), (\delta_{cb},\delta_{da}))$, which can be depicted as
\[
\xymatrix@R=0.5ex{
\{a\} \ar[r]^{\delta_{ac}} & \{c\} \ar[r]^{\delta_{cb}} & \{b\}\\
\{b\} \ar[r]_{\delta_{bd}} & \{d\} \ar[r]_{\delta_{da}} & \{a\}
}
\]
is in the pullback $(\Fun\ast \Fun)_2$, but there is no $\comp$ such that the two squares above commute, since $\comp((\delta_{ac},\delta_{bd}), (\delta_{cb},\delta_{da}))$ would have to be $(\delta_{ab},\delta_{ba})$, which do not have disjoint support (see Remark~\cite{}) and therefore are not in $\Fun_1\ast \Fun_1$.
\end{example}

\end{gray}

\begin{gray}
\section{Diagrammatic alpha equivalence}
There is a notion of alpha-equivalence on diagrams. Sequential composition can lead to the binding of names. For example, using sequential composition, we can build two diagrams

\begin{center}
      \includegraphics[page=21, width=0.4\textwidth]{twists_new}
\end{center}
that only differ in their internal names $x$ and $y$. 

\medskip
Semantically, this alpha-equivalence is a consequence of the fact that 
\[\supp(f\then\, g)=\dom(f)\cup\cod(g),\] 
which implies that both diagrams in the equation above have support $\{a,b\}$ and therefore, while renaming $x$ to $y$ in the left-hand diagram give the right-hand diagram, we also know from the definition of support that this renaming does not change the left-hand diagram. Hence both diagrams need to be equal.

\medskip
Syntactically, alpha-equivalence follows from the equation
\begin{center}
      \includegraphics[page=1, width=0.4\textwidth]{twists}
\end{center}
which allows us to insert and delete internal names. For example, this equation together with those about the interaction of renamings and generators, allows us to derive all equations in Fig~\ref{fig:Sigmaalpha}, which generate diagrammatic alpha-equivalence for the theory $\Sigma_r$ of relations.

  \begin{figure}
      \centering
      \includegraphics[page=21, width=0.4\textwidth]{twists_new}\qquad
      \includegraphics[page=26, width=0.4\textwidth]{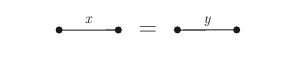}

      \includegraphics[page=22, width=0.4\textwidth]{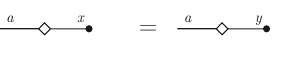}\qquad
      \includegraphics[page=25, width=0.4\textwidth]{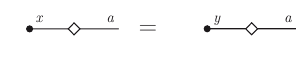}

      \includegraphics[page=23, width=0.4\textwidth]{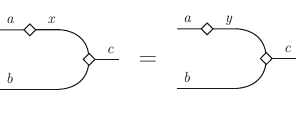}\qquad
      \includegraphics[page=24, width=0.4\textwidth]{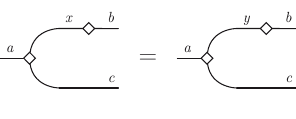}

      \includegraphics[page=27, width=0.4\textwidth]{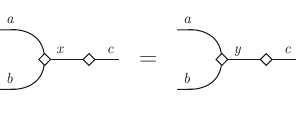}\qquad
      \includegraphics[page=28, width=0.4\textwidth]{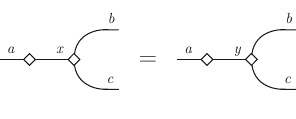}

      \includegraphics[page=29, width=0.4\textwidth]{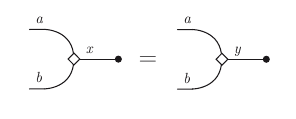}\qquad
      \includegraphics[page=30, width=0.4\textwidth]{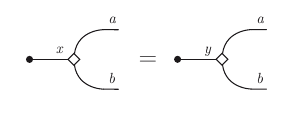}

      \includegraphics[page=15, width=0.4\textwidth]{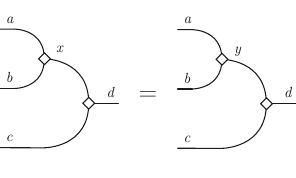}\qquad
      \includegraphics[page=17, width=0.4\textwidth]{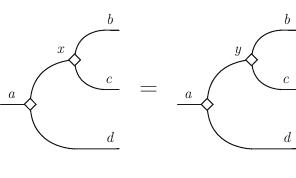}

      \includegraphics[page=31, width=0.4\textwidth]{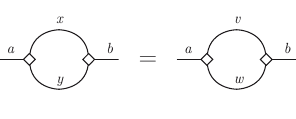}\qquad
      \includegraphics[page=20, width=0.4\textwidth]{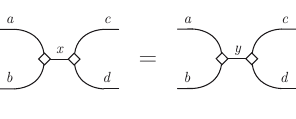}

    \caption{$\alpha$-equivalence}
    \label{fig:Sigmaalpha}
  \end{figure}

\begin{proposition}
The equations in Fig~\ref{fig:Sigmaalpha} hold in $\Sigma_{r}$.
\end{proposition}
\begin{proof}
... just do one example ...
\end{proof}

\begin{definition}
Let $\mathcal T$ be an NMT. Diagrammatic alpha equivalence is the smallest equivalence relation on $\mathcal T$-terms that is congruent with respect to $\uplus$ and satisfies $\supp (f\then\, g)=\dom(f)\cup\cod(g)$.
\end{definition}

\begin{remark}
If $t$ is a term in an NMT, we call the occurrence of names that do not occur on an input or on an output wire \emph{internal names}.
\end{remark}

\begin{proposition}
Diagrammatic alpha equivalence is the smallest equivalence relation on terms closed under
\[
\frac{\dom(f)\cup\cod(g)\subseteq\supp(\pi)\quad\pi\cdot A=_\alpha A'\quad \pi\cdot B=_\alpha B' }{A\,\then\, B =_\alpha A'\,\then\, B'}
\qquad\qquad
\frac{A=_\alpha A'\quad B=_\alpha B' }{A\uplus B =_\alpha A'\uplus B'}
\]
\end{proposition}
\begin{proof}
... POSSIBLY ADDING RULES FOR permutations and renamings ... 
\end{proof}

THIS PROBABLY CAN GO: 
\medskip
We call this equivalence \emph{diagrammatic alpha equivalence}. It is somewhat different from the one we are used to in lambda calculus, first-order logic and similar systems. In these systems, instead of defining, say, lambda abstraction as being of type $\names\times Terms\to Terms$ we can define it as being of type $[\names] Terms\to Terms$ where $[\names]$ is a type constructor that comes equipped with a surjective map $\names\times Terms\to[\names]Terms$ that quotients by alpha-equivalence \cite{gabbay-pitts}. In other words, alpha-equivalence can be incorporated in the type of the binder.

\medskip
Instead, in our situation, even though the binding is performed by sequential composition, it is not the case that sequential composition always binds a name. In particular, when sequentially composing with the identity no binding takes place.

\begin{remark}
The fact that $\uplus$ is partial introduces a certain symmetry in the interchange law \[(f_1\otimes f_2)\then(g_1\otimes g_2) = (f_1\then g_1)\otimes (f_2\then g_2)\]
Classically, if the right-hand side is defined then so is the left-hand side, while there are examples such that the left-hand side is defined and the right-hand side is not (look at the "S"). The reason is that $\then$ is partial but $\otimes$ is not. On the other hand, in the nominal case, $\otimes$ is partial as well and it may now happen that the right-hand side is defined while the left-hand side is not. One reason diagrammatic alpha-equivalence is interesting is that, up to alpha, the left-hand side is defined if the right-hand side is. 
\end{remark}

We say that $(f_1\otimes f_2)\,\then\,(g_1\otimes g_2)$ is defined up to alpha if $=_\alpha$ is an equivalence relation on diagrams and there are $f_1', f_2',g_1',g_2'$ such that $(f_1'\then g_1')\otimes (f_2'\then g_2')=_\alpha(f_1\then g_1)\otimes (f_2\then g_2)$ and $(f_1'\otimes f_2')\,\then\,(g_1'\otimes g_2')$ is defined.

\begin{proposition}
In the interchange law \[(f_1\otimes f_2)\then(g_1\otimes g_2) = (f_1\then g_1)\otimes (f_2\then g_2)\]
if the right-hand side is defined then the left-hand side is defined up to alpha-equivalence.
\end{proposition}

\begin{proof}
...
\end{proof}
\end{gray}

\begin{gray}
\section{Some remarks on equivalence of categories}

\subsection{Equivalences between finite sets and numbers} 

\renewcommand{\names}{\mathcal N}
\renewcommand{\Inj}{\mathsf n\mathbb I}
Let $\Inj$ be the category of finite subsets of $\names$ and injective maps.

\medskip\noindent
Let $\mathbb I$ be the category that has as objects the sets 
\[\underline n = \{1,\ldots n\}\]
for each $n\in\mathbb N$ and as arrows all injective functions. $\underline 0$ is the empty set.

\begin{proposition} \label{prop:F}
\begin{enumerate}
\item If $\mathbb I\to \Inj$ is fully faithful, then it preserves cardinalities.
%%%
\item If $\mathbb I\to \Inj$ is fully faithful, then it is essentially surjective.
%%%
\item\label{prop:F:ff-pincl} 
The fully faithful and inclusion preserving functors $\mathbb I\to \Inj$  are in bijection with bijections $\mathbb N\to\names$.
%%%
\item The fully faithful functors $\mathbb I\to \Inj$  are in bijection with families of bijections $(\underline n\to F_n)_{n\in\mathbb N}$ where each $F_n\subseteq\names$ is of cardinality $n$.
%%%
\item\label{prop:F:ff-pincl-iso} 
All fully faithful and inclusion preserving functors $\mathbb I\to \Inj$ are isomorphic.
%%%
\item\label{prop:F:iso-incl} 
Every functor $\mathbb I\to \Inj$  is isomorphic to a functor that preserves inclusions.
%%%
\item All fully faithful functors $\mathbb I\to \Inj$ are isomorphic.
\end{enumerate}
\end{proposition}

\begin{proof}[Proof (needs more details).]
\begin{enumerate}
\item case by case analysis ...
%%%
\item follows from previous item ...
%%%
\item every bijection $e:\mathbb N\to\names$ defines a fully faithful functor $\mathbb I\to \Inj$ via the assignment $\underline n\mapsto e[\underline n]$ ... conversely, if a functor $F:\mathbb I\to \Inj$ preserves inclusions, it defines a bijection $e:\mathbb N\to\names$ via $F(\underline 0)=\emptyset, \ldots F(\underline{n+1})=F(\underline n)\cup \{e(n+1)\}$... these two processes are inverse to each other ... 
%%%
\item a family of bijections $\kappa_n:(\underline n\to F_n)_{n\in\mathbb N}$ defines a functor $F\underline n=F_n$ where the action of $F$ on arrows is defined such that $\kappa$ is a natural transformation from the inclusion $\mathbb I\to \Inj$ to $F:\mathbb I\stackrel{}{\rightarrow}\Inj$%\to\Set$ 
... conversely, given a fully faithful functor $F:\mathbb I\stackrel{}{\rightarrow}\Inj$, there is only one way to define $\kappa_1:\underline 1 \to F\underline 1$ and $\kappa_{\underline n}(j)=F\tilde j(\kappa_1(1))$ where $\tilde j:1\to\underline n$ is defined via $\tilde j(1)=j$ ... these two processes are inverse to each other (CHECK) ... 
%%%
\item follows from item \ref{prop:F:ff-pincl}
%%%
\item this is similar to the corresponding result for Set-functors from Adamek and Trnkova, Automata and Algebras in Categories, Theorem III.4.5, page 132 ... could be worth to write out the proof in some detail to make sure the details transfer to our situation ...
%%%
\item follows from items \ref{prop:F:ff-pincl-iso} and \ref{prop:F:iso-incl} ...
 \end{enumerate}
\end{proof}

\begin{remark}
The proposition fails if we replace injections by bijections. For example, fully faithfulness then does not imply preservation of cardinality. On the other hand the proposition continues to hold if we replace injections by all functions. To adapt the proof of item \ref{prop:F:iso-incl}, we need to add the observation that fully faithful functors preserve injections.
\end{remark}

\medskip
Unfortunately, there is no fully faithful functor $\mathbb I\to \Inj$ that respects the structure of $\Inj$ as a nominal set since any such functor needs to make some arbitrary, non-equivariant choice of names. On the other hand, there are fully faithful functors
$\Inj\to\mathbb I$ internal in nominal sets. Let us first recall a general result.

\begin{remark}
Let $\mathcal C$ be a category and $I:\mathcal S\to\mathcal C$ be a skeleton, that is, $\mathcal S$ is a full subcategory containing exactly one object from each isomorphism class.  Denote by $K:\mathcal C\to\mathcal S$ a functor associating to $A$ the corresponding isomorphic object $KA$ in $\mathcal S$. To specify $K$ is to specify an isomorphism $\kappa_A:A\to IKA$ for each $A$. $K$ is defined on arrows so that $\kappa$ becomes a natural transformation. $K\kappa_A=\id_A$ and $KI=\Id$. The skeleton $\mathcal S$ is determined up to isomorphism. Every functor $\mathcal S\to \mathcal B$ extends to a functor $\mathcal C\to\mathcal B$, which is unique up to isomorphism, or  unique up to a choice of $\kappa_A:A\to IKA$. 
\end{remark}

\begin{example}
Let  $\mathsf n\mathbb B$ be the category of finite subsets of names with bijections and $\mathbb B$ the category of natural numbers with bijections. To identify a skeleton $\mathcal S$ of $\mathsf n\mathbb B$ means to give a list $(a_1,\ldots a_n)$ of names of each $n\in\mathbb N$. To define $\kappa:\Id \to IK$ then is to give a `local enumeration' $A\to |A|$. To summarise, the equivalences we are interested in are given by (i) $(a_1,\ldots a_n)$ and (ii) local enumerations $A\to |A|$. This can also be seen more directly. To give a fully faithful and essentially surjective $K: \mathsf n\mathbb B\to\mathbb B$, we choose a local enumeration $A\to |A|$.
\end{example}

\begin{proposition}
The fully faithful functors  \[\Inj\to\mathbb I\]  are in bijection with families of bijections $(A\to \underline{|A|})_{n\in\mathbb N}$, which we call local enumerations.  These functors are all equivariant and strong monoidal.
\end{proposition}

\begin{proof} (needs some details) First, fully faithfulness implies cardinality preservation as in the previous prop (CHECK). Cardinality preservation implies surjectiveness. Hence $G:\Inj\to\mathbb I$ is part of an equivalence of categories with some adjoint $F:\mathbb I\to \Inj$. Due to the previous remark $FG$ is determined by isomorphisms $\kappa_A:A\to FGA$ and due to the proposition $F$ is determined by isomorphisms $\lambda_{\underline n}\to F\underline n$ ...  (FILL IN THE DETAILS)
$G$ is equivariant on objects and arrows (CHECK). To show that \href{https://ncatlab.org/nlab/show/monoidal+functor}{https://ncatlab.org/nlab/show/monoidal+functor}, we note that there is an isomorphism $G(A)+G(B)\cong G(A\uplus B)$ satisfying associativity and unitality (CHECK). \footnote{we should revisit the notion of monoidal functor in the light of our treatment of internal categories ... }
\end{proof}

\begin{remark}
\[\Inj\to\mathbb I\] cannot be strict monoidal as $\uplus$ is commutative on $\Inj$ but $+$ is not commutative on $\mathbb I$.
\end{remark}

\subsection{translating between nominal and ordinary string diagrams}
Let $\NMT$ and $\SMT$ be the categories of nominal and symmetric monoidal categories. In the previous section we studied as an example
\[
\xymatrix{
\Inj\ar@/^/[rr]^{\Ordit} && \mathbb I \ar@/^/[ll]^{\Nomit}
}
\]
relating the $\NMT$ $\Inj$ with the $\SMT$ $\mathbb I$.

\medskip
$\Ordit$ is parameterised by an ordering, that is, by a family of bijections $\ordit_A:A\to\underline{|A|}$ for all finite $A\subseteq\names$.

\medskip
$\Nomit$ is parameterised by a choice of names, that is, by a family of sets $\Nomit(n)\subseteq \names$ of cardinality $n$ and bijections $\nomit_n:\underline n\to\Nomit(\underline n)$ for $n\in\mathbb N$.

\medskip 
$\Ordit$ and $\Nomit$ form an equivalence of categories for any choice of $\ordit$ and $\nomit$. In the following, we will assume that the choices have been made so that $\id=\ordit\circ\nomit$, which implies $\Id=\Ordit\circ\Nomit$

\medskip On objects, $\Ordit(\Nomit(\underline n))=\underline n$.
On arrows, we need to know sth more ... ?????

\medskip
Can this relationship be studied for nmts and smts in general? That is, are there functors
\[
\xymatrix{
\NMT\ar@/^/[rr]^{\ORD} && \SMT \ar@/^/[ll]^{\NOM}
}
\]
with good properties? Are there generic $\Ordit: \mathcal T\to\ORD(\mathcal T)$ and $\Nomit:\mathcal S\to \NOM(\mathcal S)$ so that $\mathcal T = \Inj = \NOM(\mathbb I)$ and $\mathcal S = \mathbb I = \ORD(\mathcal \Inj)$ are examples?

\medskip
$\ORD(\mathcal T)$ is taking the skeleton of $\mathcal T$. There is a functor $\Ordit:\mathcal T\to \ORD(\mathcal T)$ induced by some $\ordit:A\to\underline{|A|}$.

\medskip
$\NOM(\mathcal S)$ is induced by some $\nomit:\underline n\to\Nomit(\underline n)$ by closing freely under the permutation action (CHECK). $\NOM(\mathcal S)$ is independent of the choice of  $\nomit$ (CHECK). There is a functor $\Nomit:\mathcal S\to\NOM(\mathcal S)$ induced by $\nomit$ (CHECK). $\NOM(\mathcal S)$ is independent of the choice of  $\nomit$ (CHECK). $\Nomit$ and $\Ordit$ form an equivalence of categories (CHECK).

\medskip
Next we need to understand how to translate presentations of theories. 

\[\ordit:\mathcal T\to\ORD(\mathcal T)\]
\[\ordit((a\ b)\cdot t) = \ordit(t)\]

the last equation is only true for generators?

... need to sort out the details ...

\medskip\noindent\textbf{Thm 1?} $\Ordit$ and $\Nomit$ translate proofs
$\Ordit$ and $\Nomit$ translate diagrams and equations up to order and permuations?

\medskip\noindent\textbf{Thm 2?}
\[\Ordit(f)\equiv_o\Ordit(g) \Rightarrow f\equiv_n g\]
\[\Nomit(f)\equiv_n\Nomit(g) \Rightarrow f\equiv_o g\]
\end{gray}

\begin{gray}
\section{Lists of names}

lists of names ... oNMT ... oNMTs can be easily generated from SMTs ... so the main work of the transfer theorem is in the equivalence of oNMTs and NMTs

\begin{definition}[oNMT]
A "nominal $\PROP$" $\mathbb C$ is a small category, with a set $\mathbb C_0$ of `objects' and a set $\mathbb C_1$ of `arrows', defined as follows. We write $\then$ for the `sequential' composition (in the diagrammatic order) and $\uplus$ for the `parallel' or `monoidal' composition.
\begin{itemize}
\item $\mathbb C_0$  is the set of finite subsets of $\names$. The permutation action is given by $\pi\cdot A= \pi[A]=\{\pi(a) \mid a\in A\}$. 
%%%
\item $\mathbb C_1$ contains all bijections (`renamings') $\pi_A:A\to\pi\cdot A$ for all finite permutations $\pi:\names\to\names$ and is closed under the operation mapping an arrow $f:A\to B$ to $\pi\cdot f : \pi\cdot A\to \pi\cdot B$ defined as $\pi\cdot f = (\pi_A)^{-1}\then f\then \pi_B$. 
%%%
\item $A \uplus B$ is the union of $A$ and $B$ and defined whenever $A$ and $B$ are disjoint. This makes $(\Fun_0,\emptyset,\uplus)$ a partial commutative monoid. On arrows, we require $(\Fun_1,\emptyset,\uplus)$ to be a partial commutative monoid, with $f\uplus g$ defined whenever $\dom f\cap\dom g=\emptyset$ and $\cod f\cap\cod g=\emptyset$.
\item did we forget anything?
\end{itemize}
\end{definition}

The next theorem shows that oNMTs and NMTs are equivalent as internal monoidal categories in ...

\begin{theorem}
... we need to figure out what goes here exactly
\end{theorem}

\end{gray}

\begin{gray}
\section{about adjoints}

\subsection{willgosomewherelater}

\begin{ak}

Each $\PROP$ is also a  monoidal category internal in nominal sets where all objects and arrows have empty support. More generally, we have a `discrete' functor 
\footnote{The discrete functor $\Set\to\Nom$ has both left and right adjoints. The right adjoint maps a nominal set to its subset of elements with empty support. The left adjoint maps a nominal set to its set of orbits. On the other hand, the forgetful functor $\Nom\to\Set$ does not preserve (infinite) products and only has a right adjoint which maps a set $X$ to $X^\mathfrak S$ where $\mathfrak S$ is the set of all finite permutations $\Nom\to\Nom$.}
\[D:\Cat(\Set)\to\Cat(\Nom,1,\ast)\]
which inherits from the discrete functor $\Set\to\Nom$ the right adjoint (CHECK), but not the left adjoint, since we cannot map arrows to their orbits in a composition preserving way. This situation can be repaired using the notion of a local enumeration of names.

\medskip\noindent
The left adjoint $Orb:\Nom\to\Set$ of $D:\Set\to\Nom$ does not extend to $Orb:\Cat(\Nom,1,\ast)\Cat(\Set)$. For example, if $\mathbb C$ is the internal nominal category of finite subsets of names with bijections, then $Orb(\mathbb C_1)$ collapses any bijections $f,g$ where the domains have the same cardinality and the codomains have the same cardinality, so that the domain and codomain arrows cannot be defined.

% \medkskip\noindent

\end{ak}
\end{gray}

\begin{gray}
\section{Bits from section 6}

\begin{remark}\label{eqs:NomOrdDerivable}
The following equations can be obtained from the equations in Definitions \ref{def:NOM},\ref{def:ORD}:
\begin{align*}
[\boldsymbol a\rangle\, f \hcomp \langle\boldsymbol b|\boldsymbol b'\rangle \hcomp g \,\langle\boldsymbol c]
&= [\boldsymbol a\rangle\, f \,\langle\boldsymbol b] \hcomp [\boldsymbol b'\rangle\, g \,\langle\boldsymbol c]
\\
[\boldsymbol a\rangle\, \langle\boldsymbol b|\boldsymbol b'\rangle \,\langle\boldsymbol c]
&= [\boldsymbol a|\boldsymbol b] \hcomp [\boldsymbol b'|\boldsymbol c]
\\
[\boldsymbol a\rangle\langle\boldsymbol a]\, [\boldsymbol a|\boldsymbol b] \,[\boldsymbol b\rangle\langle\boldsymbol b]
&= [\boldsymbol a|\boldsymbol b]
\\
\langle\boldsymbol a]\, f \hcomp [\boldsymbol b|\boldsymbol c] \hcomp g \,[\boldsymbol d\rangle
&= \langle\boldsymbol a]\, f \,[\boldsymbol b\rangle \hcomp \langle\boldsymbol c]\, g \,[\boldsymbol d\rangle
\\
\langle\boldsymbol a]\, [\boldsymbol a'|\boldsymbol b'] \,[\boldsymbol b\rangle
&= \langle\boldsymbol a|\boldsymbol a'\rangle \hcomp \langle\boldsymbol b'|\boldsymbol b\rangle
\\
\langle\boldsymbol a][\boldsymbol a\rangle\, \langle\boldsymbol a|\boldsymbol a'\rangle \,\langle\boldsymbol a'][\boldsymbol a'\rangle
&= \langle\boldsymbol a|\boldsymbol a'\rangle
\\
\langle\boldsymbol a][\boldsymbol a\rangle\, f \,\langle\boldsymbol b][\boldsymbol b\rangle
&= \langle\boldsymbol c][\boldsymbol c\rangle\, f \,\langle\boldsymbol d][\boldsymbol d\rangle
\\
% \langle a]\id[a\rangle 
% & = \id
% \\
% \langle a,b]\delta_{ab}\uplus \delta_{ba}[a,b\rangle 
% & = tw
% \\
% [\boldsymbol a\rangle\langle\boldsymbol a | \boldsymbol b\rangle\langle\boldsymbol b]
% & = \id
% \\
% [\boldsymbol a\rangle\langle\boldsymbol a] f[\boldsymbol b\rangle\langle\boldsymbol b]
% & = f
% \\
% \langle\boldsymbol a][\boldsymbol a\rangle f\langle\boldsymbol b][\boldsymbol b\rangle
% & = f
% \\
% \langle\boldsymbol a]f\uplus g[\boldsymbol c\rangle
% &= 
% \langle\boldsymbol a] id[\boldsymbol a'\rangle \hcomp \langle\boldsymbol a'] f\uplus g [\boldsymbol b'\rangle\hcomp\langle\boldsymbol b'] id[\boldsymbol c\rangle
\end{align*}

% \begin{align*}
% \langle a,b]\delta_{ab}\uplus \delta_{ba}[a,b\rangle 
% & = tw \\
% \end{align*}

\begin{proof}
\begin{align*}
[\boldsymbol a \rangle f \hcomp \langle\boldsymbol b | \boldsymbol b'\rangle \hcomp g\langle\boldsymbol c]
&= [\boldsymbol a \rangle\, f \,\hcomp \langle\boldsymbol b | \boldsymbol b'\rangle \,\langle\boldsymbol b' ]\, \hcomp\, [\boldsymbol b'\rangle\, g \,\langle\boldsymbol c]
\\
&= [\boldsymbol a \rangle\, f \,\langle\boldsymbol b ]\, \hcomp \,[\boldsymbol b' | \boldsymbol b']\, \hcomp\, [\boldsymbol b'\rangle\, g \,\langle\boldsymbol c]
\\
&= [\boldsymbol a \rangle\, f \,\langle\boldsymbol b]\, \hcomp \,[\boldsymbol b'\rangle\, g \,\langle\boldsymbol c]
\end{align*}

\begin{align*}
\langle\boldsymbol a]\, f \hcomp [\boldsymbol b|\boldsymbol c] \hcomp g \,[\boldsymbol d\rangle
&= \langle\boldsymbol a]\, f \,\hcomp\, [\boldsymbol b|\boldsymbol c] \,[\boldsymbol c\rangle\ \hcomp\ \langle\boldsymbol c]\, g \,[\boldsymbol d\rangle
\\
&= \langle\boldsymbol a]\, f \,[\boldsymbol b\rangle\, \hcomp \,\langle\boldsymbol c | \boldsymbol c\rangle\, \hcomp\, \langle\boldsymbol c]\, g \,[\boldsymbol d\rangle
\\
&= \langle\boldsymbol a]\, f \,[\boldsymbol b\rangle\, \hcomp \,\langle\boldsymbol c]\, g \,[\boldsymbol d\rangle
\end{align*}

The choice of $\boldsymbol a, \boldsymbol b$ is arbitrary, because we can prove that for any other choice $\boldsymbol c, \boldsymbol d$, we have $\langle\boldsymbol a][\boldsymbol a\rangle\, f \,\langle\boldsymbol b][\boldsymbol b\rangle = \langle\boldsymbol c][\boldsymbol c\rangle\, f\,\langle\boldsymbol d][\boldsymbol d\rangle$:
\begin{align*}
\langle\boldsymbol a][\boldsymbol a\rangle\, f\,\langle\boldsymbol b][\boldsymbol b\rangle
& = \langle\boldsymbol a](\,[\boldsymbol a\rangle \langle \boldsymbol c|\boldsymbol c \rangle \hcomp f \hcomp \langle \boldsymbol d|\boldsymbol d \rangle \,\langle\boldsymbol b]\,)[\boldsymbol b\rangle
\\
& = \langle\boldsymbol a](\,[\boldsymbol a | \boldsymbol c] \hcomp  [ \boldsymbol c\rangle \, f \hcomp \langle \boldsymbol d|\boldsymbol d \rangle \,\langle\boldsymbol b]\,)[\boldsymbol b\rangle
\\
& = \langle\boldsymbol a | \boldsymbol a \rangle \hcomp \langle \boldsymbol c]([ \boldsymbol c\rangle \, f \hcomp \langle \boldsymbol d|\boldsymbol d \rangle \,\langle\boldsymbol b]\,)[\boldsymbol b\rangle
\\
& = \langle \boldsymbol c](\,[ \boldsymbol c\rangle \, f \hcomp \langle \boldsymbol d|\boldsymbol d \rangle \ \langle\boldsymbol b]\,)[\boldsymbol b\rangle
\\
& = \langle \boldsymbol c](\,[ \boldsymbol c\rangle \, f \,\langle \boldsymbol d] \,\hcomp\, [\boldsymbol d |\boldsymbol b]
\,)[\boldsymbol b\rangle
\\
& = \langle \boldsymbol c]\ \,[ \boldsymbol c\rangle \, f \,\langle \boldsymbol d]\ \,[\boldsymbol d\rangle \,\hcomp \langle\boldsymbol b |\boldsymbol b\rangle
\\
& = \langle \boldsymbol c][ \boldsymbol c\rangle \, f \,\langle \boldsymbol d][\boldsymbol d\rangle
\\
\end{align*}
\end{proof}

\end{remark}

\begin{gray}
\begin{remark} \begin{ak}not sure that is needed ... maybe it can stay a remark ...\end{ak}
If we move to an $\nPROP$ that has names $\names\cup\mathbb N$, we can interpret the new symbols as arrows 
\[
\xymatrix{
\underline{\boldsymbol a} 
\ar[r]^f &  \underline{\boldsymbol b}
\\
\underline n 
\ar[u]^{\langle\boldsymbol a]} 
\ar[r]^{\langle\boldsymbol a] f [\boldsymbol b\rangle} 
& \underline m 
\ar@{<-}[u]_{[\boldsymbol b\rangle}
}
\quad\quad \textrm{}\quad\quad
\xymatrix{
\underline{\boldsymbol a} 
\ar[r]^{[\boldsymbol a\rangle g \langle\boldsymbol b]} 
&  \underline{\boldsymbol b}
\\
\underline n 
\ar@{<-}[u]^{[\boldsymbol a\rangle} 
\ar[r]^{g} 
& \underline m 
\ar[u]_{\langle\boldsymbol b]}
}
\]
%Alternative notation
%\[
%\xymatrix{
%\underline{\boldsymbol a} 
%\ar[r]^f &  \underline{\boldsymbol b}
%\\
%\underline n 
%\ar[u]^{[\boldsymbol a\rangle} 
%\ar[r]^{[\boldsymbol a\rangle f \langle\boldsymbol b]} 
%& \underline m 
%\ar[u]_{\langle\boldsymbol b]}
%}
%\]
%and
%\[
%\xymatrix{
%\underline{\boldsymbol a} 
%\ar[r]^{\langle\boldsymbol a] g [\boldsymbol b\rangle} 
%&  \underline{\boldsymbol b}
%\\
%\underline n 
%\ar@{<-}[u]^{\langle\boldsymbol a]} 
%\ar[r]^{g} 
%& \underline m 
%\ar[u]_{[\boldsymbol b\rangle}
%}
%\]
Choosing a list $\boldsymbol a$ for a set $\underline{\boldsymbol a}$ corresponds to the choice of a local enumeration, see Definition~\ref{def:localenumeration}.
\end{remark} 
\end{gray}

\begin{gray}
\begin{example}
Let $tw:2\to 2$ be the twist in some $\PROP$ $ \mathcal S$. The arrows $[\boldsymbol a\rangle tw\langle\boldsymbol b]$ in $\NOM(\mathcal S)$ are bijections $\{a_1,a_2\}\to\{b_1,b_2\}$. The arrows $\langle\boldsymbol a][\boldsymbol a'\rangle tw\langle\boldsymbol b'][\boldsymbol b\rangle$ in $\ORD(\NOM(\mathcal S))$ correspond to either $tw$ or $\id$.
\end{example}
\end{gray}
\end{gray}

%%
%% Bibliography
%%

%% Please use bibtex, 

%\bibliography{calco-2019}

\end{document}